\documentclass[final,10pt,journal,twoside]{IEEEtran}
\usepackage{amsfonts}
\usepackage{microtype}
\usepackage{siunitx}
\usepackage{graphics, graphicx}
\usepackage{amsthm, amssymb}
\usepackage[cmex10]{amsmath}
\usepackage{ifpdf}
\usepackage{epstopdf}
\usepackage{cite}
\usepackage{algorithm, algpseudocode, algpascal}
\usepackage{array}
\usepackage{booktabs}
\usepackage{stfloats}
\usepackage{enumerate}

\input{tcilatex}
\allowdisplaybreaks[4]

\begin{document}

\title{Analysis, detection and control of secure and safe cyber-physical
control systems in a unified framework}
\author{Linlin~Li, Steven~X.~Ding, Maiying Zhong, Dong Zhao, Yang Shi\thanks{%
This work has been supported by the National Natural Science Foundation of
China under Grants 62322303, and 62233012.} \thanks{%
L. Li is with School of Automation and Electrical Engineering, University of
Science and Technology Beijing, Beijing 100083, P. R. China, Email:
linlin.li@ustb.edu.cn} \thanks{%
S. X. Ding is with the Institute for Automatic Control and Complex Systems
(AKS), University of Duisburg-Essen, Germany. Email: steven.ding@uni-due.de} 
\thanks{%
M. Zhong is with the College of Electrical Engineering and Automation,
Shandong University of Science and Technology, Qingdao 266590, China. Email:
mzhong@buaa.edu.cn} \thanks{%
D. Zhao is with School of Cyber Science and Technology, Beihang University,
Beijing 100191, China. E-mail: dzhao@buaa.edu.cn} \thanks{%
Y. Shi is with Department of Mechanical Engineering, University of Victoria,
Victoria, BC, V8W 2Y2, Canada. E-mail: yshi@uvic.ca}}
\maketitle

\begin{abstract}
This paper deals with analysis, simultaneous detection of faults and
attacks, fault-tolerant control and attack-resilient of cyber-physical
control systems. In our recent work, it has been observed that an attack
detector driven by an input residual signal is capable of reliably detecting
attacks. In particular, observing system dynamics from the perspective of
the system input-output signal space reveals that attacks and system
uncertainties act on different system subspaces. These results motivate our
exploration of secure and safe cyber-physical control systems in the unified
framework of control and detection. The unified framework is proposed to
handle control and detection issues uniformly and in subspaces of system
input-output data. Its mathematical and control-theoretic basis is system
coprime factorizations with Bezout identity at its core. We firstly explore
those methods and schemes of the unified framework, which serve as the major
control-theoretic tool in our work. It is followed by re-visiting and
examining established attack detection and resilient control schemes. The
major part of our work is the endeavours to develop a control-theoretic
paradigm, in which analysis, simultaneous detection of faults and attacks,
fault-tolerant and attack-resilient control of cyber-physical control
systems are addressed in a unified manner.
\end{abstract}

%




\section{Introduction}

\IEEEPARstart{T}{oday's} automatic control systems as the centerpiece of
industrial cyber-physical systems (CPSs) are fully equipped with intelligent
sensors, actuators and a modern information infrastructure. It is a logic
consequence of ever increasing demands for system performance and production
efficiency that today's automatic control systems are of an extremely high
degree of integration, automation and complexity \cite{Ding2020}.
Maintaining reliable and safe operations of cyber-physical control systems
(CPCSs) are of elemental importance for optimally managing industrial CPSs
over the whole operation life-cycle. As an indispensable maintenance
functionality, real-time monitoring, fault detection (FD) and fault-tolerant
control (FTC) \cite{CRB2001,PFC00,Blanke06,Ding2008,Ding2020} are widely
integrated in CPCSs and run parallel to the embedded control systems. In a
traditional automatic control system, FD and FTC are mainly dedicated to
maintaining functionalities of sensors and actuators as these key components
become defective \cite{AET_SMO_book_2011,Ding2020}.

Recently, a new type of malfunctions, the so-called cyber-attacks on CPSs,
have drawn attention to the urgent need for developing new monitoring,
diagnosis and resilient control strategies \cite%
{TEIXEIRA-zero-attack_2015,DIBAJI2019-survey,DeruiDing2021-survey,TGXHV2020,Zhou2021IEEE-Proc,ZHANG2021,Survey-Segovia-Ferreira2024}%
. Cyber-attacks can not only considerably affect functionalities of sensors
and actuators, but also impair communications among the system components
and sub-systems. Differing from technical faults, cyber-attacks are
artificially created. As the main categories of malicious attacks,
denial-of-service (DoS) attacks \cite{QLSY-TAC-2018,DIBAJI2019-survey} and
deception attacks \cite{DIBAJI2019-survey}, attract major research
attention. By blocking the communication links of system components over
networks, DoS attacks can degrade system performance \cite{RWDS-AUTO-22018}.
The deception attacks are usually designed to cause system performance
degradation through modifying data packets without being detected and may
lead to immense damages during system operations \cite%
{Survey-attack-detection2018,DIBAJI2019-survey,TGXHV2020}. In particular, it
was shown that attackers with model knowledge can exploit specified system
dynamics to construct perfectly stealthy sensor and actuator attacks that
evade classical Luenberger or Kalman-filter-based residuals \cite%
{Mo2010false,Pasqualetti2013}. These results revealed intrinsic
vulnerabilities of passive observer-based detectors and provided a unifying
framework linking cyber-attack detection to classical concepts from fault
detection. Building on this foundation, a large body of work like
observer-based FD methods, unknown-input observers, parity relations, and
structured observer banks have been proposed to detect the so-called replay,
zero dynamics, covert attacks, false data injection attacks, also called
integrity stealthy attacks \cite%
{TEIXEIRA-zero-attack_2015,MS-TAC-2016,Chen2018attacks,DIBAJI2019-survey,Zhang2020attacks,Yang2022attacks,Shang2022attacks}%
. Complementary to it, active detection approaches such as dynamic
watermarking inject designed excitation signals into the control inputs,
rendering attacks statistically detectable \cite%
{Mo2015-Watermarked-detection}.

As a response to the security issues of CPSs, resilient control, where the
objective is to maintain closed-loop stability and acceptable performance
despite ongoing attacks, has received increasing attention from both the
research and application domains. As appropriate countermeasures, the
so-called resource-aware secure control methods \cite%
{HeemelsCDC2012,DeruiDing2021-survey}, like event-triggered control
algorithms or switching control mechanisms, are prevalent resilient control
schemes. These methods make the CPSs highly resilient against attacks by
means of optimal management of data communications among subsystems in the
CPS. Beyond detection, observer-based methods also play a key role in
resilient control. Typical architectures integrate attack detection with
controller reconfiguration or compensation mechanisms, e.g. via switching
observers and controllers or attack-tolerant state estimators \cite%
{Sandberg2015}. Further methods include model predictive control \cite%
{SZCDCX-AUTO-2024}, adaptive control \cite{AY-IS-2018}, and sliding mode
control \cite{YZDYF-AUTO-2023}.

The argument that a detailed understanding of attack mechanisms is
fundamental for assessing CPS vulnerabilities and informing the development
of defense strategies has spurred significantly recent interest in attack
design \cite%
{FQCTZ-TAC-2020,Shang2022attacks,ChenTAC2018,AY-TAC-2018,GYH-TCNS-2021,Sui2021,ZHOU2023110723,ZhangTAC2023}%
. Attack design presupposes that attackers are in the possession of partial
or full system knowledge. In this context, a further type of cyber-attacks,
the so-called eavesdropping attacks, play a key role \cite%
{DIBAJI2019-survey,Survey-Segovia-Ferreira2024}. Although such attacks do
not cause changes in system dynamics and performance degradation, they
enable adversary to gain system knowledge, which is especially critical when
an attacker uses such knowledge for designing stealthy attacks. This shift
has revealed new and fundamental privacy risks that motivated recent
research on privacy preserving and protection issues \cite%
{HanPrivacy2018,LUreviewARC2019,Kawano2020TAC,Survey-Segovia-Ferreira2024}.

In our recent work, we observed that an observer-based attack detector with
an output residual generator is less capable than an attack detector driven
by an input residual signal \cite{DLautomatica2022}. Stimulated by this
observation, our further exploration revealed meaningful evidence \cite%
{LDZZ2024,Ding2026} that

\begin{itemize}
\item attacks can be equivalently expressed by an unknown input in the
system input, which is exactly the input residual signal, and

\item attacks and system uncertainties, including faults, affect the system
dynamics in a dual manner.
\end{itemize}

In particular, when the system dynamics are examined from the perspective of
the system input-output signal space, it is plainly seen that attacks and
system uncertainties act on different subspaces. This prospect is strongly
obscured, as viewing a CPCS from the perspective of a classic closed-loop
configuration. From that point of view, the influences of attacks and system
uncertainties on the system dynamic are fully corrupted, which considerably
hinder a reliable attack detection and even simultaneous detection of
attacks and faults. These results have motivated us to explore secure and
safe CPCSs in the unified framework of control and detection \cite%
{LDZZ2024,Ding2026}.

On account of the cognition that both feedback control and fault detection
systems are driven by (output) residual signal, the unified framework has
been initially proposed to manage control and detection issues uniformly, in
which control and detection problems are handled in subspaces of system
input-output data, rather than in a closed-loop configuration. Its
groundwork is the output-input residual signals, serving as indicators that
fully characterize the system nominal behavior and uncertain dynamics,
including uncertainties induced by faults and cyber-attacks. The
mathematical and control-theoretic basis of the framework is coprime
factorizations of plant models and feedback controllers, where the Bezout
identity, defining a unimodular linear isomorphism between the input-output
signals and the residual signals, is at its core \cite{Ding2026}.

Inspired by these observations, the main objective of this paper is to
leverage the unified framework of control and detection to study issues of
analysis, simultaneous detection, fault-tolerant and resilient control of
cyber-physical control systems under attacks and faults. Our work comprises
three steps. We will firstly explore those methods and schemes in the
unified framework, which serve as the major control-theoretic tool to deal
with the related issues. It is followed by re-visiting and examining state
of the art methods of the relevant techniques, in particular attack
detection and resilient control schemes. The last step, the major part of
our work, is the endeavours to develop a control-theoretic paradigm, in
which analysis, simultaneous detection of faults and attacks, fault-tolerant
and attack-resilient control of CPCSs are unified addressed. Our work will
deliver the following contributions:

\begin{itemize}
\item Establishment of control-theoretic setting for the subsequent work on
secure CPCSs with system image and residual subspaces in its core
(Definitions \ref{Def3-1}-\ref{Def3-3}), gap metric as the distance between
two subspaces, Bezout identities and the associated transformations, and
coprime factorizations of the plant and controllers (Theorem \ref{Theo3-1}
and Corollary \ref{Co3-1}) as the basic mathematical and control-theoretic
tool;

\item Exploration of secure CPCS configurations, particularly the
Plug-and-Play (PnP) structure that builds the basic system configuration for
the proposed fault-tolerant and attack-resilient control scheme (Subsection %
\ref{subsection3-3-1});

\item Alternative and enhanced forms of three well-established attack
detection and resilient control schemes, the MTD and watermark methods
(Subsections \ref{subsection3-3-2}-\ref{subsection3-3-3}), the auxiliary
system-based detection method (Subsection \ref{subsectionV-A-2}), and
characterization of zero dynamics attacks in the unified framework
(Subsection \ref{subsection3-3-4});

\item Introduction of the concept of attack information potential between
the plant and control station that builds the basis for attack detection and
resilient control study (\ref{subsectionIV-B-A}), and proof of two
elementary forms of the system dynamic under cyber-attacks, (i) on the plant
side (Theorem \ref{Theo4-3}), and (ii) on the control station side (Theorem %
\ref{Theo4-4});

\item Demonstration of the duality between the faulty dynamic and attack
dynamic (Subsection \ref{subsectionIV-B-C});

\item Introduction of the concepts of attack detectability, as a general and
detection method-independent definition of stealthy attacks, the actuability
of attacks that is overlooked in the literature, but important for
addressing attack design issues, and proof of the existence conditions
(Subsection \ref{subsectionIV-C-A});

\item Exploration of image attacks as a general form of undetectable
(stealthy) attacks that simplifies (stealthy) attacks design (Subsection \ref%
{subsectionIV-C-A});

\item Introduction of attack stability margin as a measure of the system
vulnerability (Subsection \ref{subsectionIV-C-B}) and proof of a sufficient
condition of the system vulnerability under image attacks (Theorem \ref%
{Theo4-8});

\item Exploration of three schemes of simultaneous detection of faults and
attacks (Subsection \ref{SubsectionV-A});

\item Exploration of a fault-tolerant and attack-resilient control system (%
\ref{SubsectionV-B}) with a performance analysis (Theorems \ref{Theo5-1}-\ref%
{Theo5-2}) and a proof of the stability condition (Theorem \ref{Theo5-3});

\item Exploration of a subspace-based privacy-preserving scheme whose core
is a privacy filter generating an auxiliary signal (\ref{SubsectionV-C});

\item Introduction of the concepts of system-level privacy and $\mathcal{L}%
_{2}$-privacy based on the similarity of two system subspaces and measured
by gap metric (Definitions \ref{Def5-1}-\ref{Def5-2});

\item Formulation and solution of the $\mathcal{L}_{2}$-privacy-filter
design as an optimization problem (Theorem \ref{Theo5-5}).
\end{itemize}

\noindent \textbf{Notations}. Throughout this paper, standard notations
known in control theory and linear algebra are adopted. In addition, $\ell
_{2}$ is the space of square-summable sequences over $\left( -\infty ,\infty
\right) ,$ $\ell _{2}\left( [0,\infty )\right) $ is its subspace of causal
sequences. $\mathcal{H}_{2}$ is the Hardy space of $z$-transforms of
sequences in $\ell _{2}\left( [0,\infty )\right) .$ $\mathcal{RH}_{\infty }$ 
$\left( \mathcal{RL}_{\infty }\right) $ denotes the space of real-rational,
proper transfer functions with poles in the unit disk (without poles on $%
\left\vert z\right\vert =1$). $G^{\ast }(z)=G^{T}\left( z^{-1}\right) $ is
the conjugate of $G\left( z\right) .$

\section{Preliminaries and Problem Formulation}

\label{sec2}

\subsection{Configuration of feedback control systems}

The linear discrete-time invariant (LDTI) feedback control systems are
considered in this work and modelled by%
\begin{equation}
\left\{ 
\begin{array}{c}
y\left( z\right) =G\left( z\right) u\left( z\right) +r_{y,0}\left( z\right)
\\ 
u\left( z\right) =K\left( z\right) y\left( z\right) +r_{u,0}\left( z\right) ,%
\end{array}%
\right.  \label{eq2-1}
\end{equation}%
where $u\in \mathbb{R}^{p}$ and $y\in \mathbb{R}^{m}$ represent the plant
input and output vectors, respectively. $G\left( z\right) $ is the transfer
function matrix with the minimal state space realisation $\left(
A,B,C,D\right) $ described by 
\begin{equation}
\left\{ 
\begin{array}{l}
x(k+1)=Ax(k)+Bu(k) \\ 
y(k)=Cx(k)+Du(k)+r_{y,0}(k).%
\end{array}%
\right.  \label{eq2-2}
\end{equation}%
Here, $x\in \mathbb{R}^{n}$ is the state vector, matrices $A,B,C,D$ are
known and of appropriate dimensions. The transfer function matrix $K\left(
z\right) $ denotes an LDTI control system that will be described in the next
subsection, and $\left( r_{u,0},r_{y,0}\right) $ model a given control
command and uncertainties in the plant, respectively, and will be specified
in the sequel.

\subsection{System coprime factorizations \label{SubsecII-A}}

Coprime factorizations of transfer function matrices \cite{Tay1998,Zhou96}
play an essential role in our subsequent study and serves as a basic tool.
Given the plant model $G,$ the factorizations%
\begin{equation}
G(z)=N(z)M^{-1}(z)=\hat{M}^{-1}(z)\hat{N}(z),  \label{eq2-3}
\end{equation}%
$M(z),N(z),\hat{M}(z),\hat{N}(z)\in \mathcal{RH}_{\infty },$ are called
right coprime factorization (RCF) and left coprime factorization (LCF) of $%
G(z),$ respectively, if there exist $X(z),Y(z),\hat{X}(z),\hat{Y}(z)\in 
\mathcal{RH}_{\infty }$ such that 
\begin{equation}
\left[ 
\begin{array}{cc}
X(z) & Y(z)%
\end{array}%
\right] \left[ 
\begin{array}{c}
M(z) \\ 
N(z)%
\end{array}%
\right] =I,\left[ 
\begin{array}{cc}
-\hat{N}(z) & \hat{M}(z)%
\end{array}%
\right] \left[ 
\begin{array}{c}
-\hat{Y}(z) \\ 
\hat{X}(z)%
\end{array}%
\right] =I.  \label{eq2-4}
\end{equation}%
Note that $\left( X(z),Y(z)\right) $ and $\left( \hat{X}(z),\hat{Y}%
(z)\right) $ are left- and right coprime pairs (LCP and RCP), respectively.
It is known that 
\begin{equation}
\left[ 
\begin{array}{cc}
X(z) & Y(z)%
\end{array}%
\right] \left[ 
\begin{array}{c}
-\hat{Y}(z) \\ 
\hat{X}(z)%
\end{array}%
\right] =0\Longleftrightarrow X(z)\hat{Y}(z)=Y(z)\hat{X}(z).  \label{eq2-5}
\end{equation}%
The double Bezout identity \cite{Tay1998,Zhou96} gives a compact form of (%
\ref{eq2-3})-(\ref{eq2-4}),%
\begin{equation}
\left[ 
\begin{array}{cc}
X(z) & \text{ }Y(z) \\ 
-\hat{N}(z) & \text{ }\hat{M}(z)%
\end{array}%
\right] \left[ 
\begin{array}{cc}
M(z) & \text{ }-\hat{Y}(z) \\ 
N(z) & \text{ }\hat{X}(z)%
\end{array}%
\right] =I.  \label{eq2-6}
\end{equation}%
The Bezout identity (\ref{eq2-6}) plays an essential role in our subsequent
work. In particular, the state space realizations of the eight transfer
functions in (\ref{eq2-6}) and the associated alternative system
representations build the cornerstone of our work.

\textbf{The RCF }$\left( M(z),N(z)\right) $\textbf{\ is a state feedback
control system }given by%
\begin{gather}
\left\{ 
\begin{array}{l}
x(k+1)=A_{F}x(k)+BVv(k),A_{F}=A+BF \\ 
\left[ 
\begin{array}{c}
u(k) \\ 
y(k)%
\end{array}%
\right] =\left[ 
\begin{array}{c}
Fx(k)+Vv(k) \\ 
C_{F}x(k)+DVv(k)%
\end{array}%
\right] ,C_{F}=C+DF,%
\end{array}%
\right.  \label{eq2-7} \\
\left[ 
\begin{array}{c}
u(z) \\ 
y(z)%
\end{array}%
\right] =I_{G}(z)v(z),I_{G}(z)=\left[ 
\begin{array}{c}
M(z) \\ 
N(z)%
\end{array}%
\right] ,  \label{eq2-7a} \\
M=\left( A_{F},BV,F,V\right) ,N=\left( A_{F},BV,C_{F},DV\right) .
\label{eq2-8}
\end{gather}%
The system (\ref{eq2-7})/(\ref{eq2-7a}) is also called stable image
representation (SIR)\textbf{.} As the dual system of the SIR, \textbf{the
LCF }$\left( \hat{M}(z),\hat{N}(z)\right) $\textbf{\ is an observer-based
residual generator }described by%
\begin{gather}
\left\{ 
\begin{array}{l}
\hat{x}(k+1)=A_{L}\hat{x}(k)+B_{K}\left[ 
\begin{array}{c}
u(k) \\ 
y(k)%
\end{array}%
\right] ,A_{L}=A-LC \\ 
r_{y}(k)=C_{K}\hat{x}(k)+D_{K}\left[ 
\begin{array}{c}
u(k) \\ 
y(k)%
\end{array}%
\right] ,C_{K}=WC \\ 
B_{K}=\left[ 
\begin{array}{cc}
B_{L} & \text{ }-L%
\end{array}%
\right] ,B_{L}=B-LD,D_{K}=\left[ 
\begin{array}{cc}
WD & \text{ }W%
\end{array}%
\right] ,%
\end{array}%
\right.  \label{eq2-9} \\
r_{y}(z)=K_{G}(z)\left[ 
\begin{array}{c}
u(z) \\ 
y(z)%
\end{array}%
\right] ,K_{G}(z)=\left[ 
\begin{array}{cc}
-\hat{N}(z) & \text{ }\hat{M}\left( z\right)%
\end{array}%
\right] ,  \label{eq2-10} \\
\hat{M}=\left( A_{L},-L,C_{K},W\right) ,\hat{N}=\left(
A_{L},B_{L},C_{K},WD\right) ,  \label{eq2-8a}
\end{gather}%
which is called stable kernel representation (SKR). The RCP and LCP, $\left( 
\hat{X}(z),\hat{Y}(z)\right) $ and $\left( X(z),Y(z)\right) ,$ are \textbf{%
observer-based state feedback control and input residual generator systems},
respectively, and modelled by%
\begin{gather}
\left\{ 
\begin{array}{l}
\hat{x}(k+1)=A_{F}\hat{x}(k)+LW^{-1}r_{y}(k), \\ 
\left[ 
\begin{array}{c}
u(k) \\ 
y(k)%
\end{array}%
\right] =\left[ 
\begin{array}{c}
F \\ 
C_{F}%
\end{array}%
\right] \hat{x}(k)+\left[ 
\begin{array}{c}
0 \\ 
W^{-1}%
\end{array}%
\right] r_{y}(k),%
\end{array}%
\right.  \label{eq2-11} \\
\left[ 
\begin{array}{c}
u(z) \\ 
y(z)%
\end{array}%
\right] =\left[ 
\begin{array}{c}
-\hat{Y}(z) \\ 
\hat{X}(z)%
\end{array}%
\right] r_{y}(z),  \label{eq2-12} \\
\hat{Y}=\left( A_{F},-LW^{-1},F,0\right) ,\hat{X}=\left(
A_{F},LW^{-1},C_{F},W^{-1}\right) ,  \label{eq2-8d} \\
\left\{ 
\begin{array}{l}
\hat{x}(k+1)=A_{L}\hat{x}(k)+B_{K}\left[ 
\begin{array}{c}
u(k) \\ 
y(k)%
\end{array}%
\right] \\ 
r_{u}(k)=V^{-1}F\hat{x}(k)+\left[ 
\begin{array}{cc}
V^{-1} & \text{ }0%
\end{array}%
\right] \left[ 
\begin{array}{c}
u(k) \\ 
y(k)%
\end{array}%
\right] ,%
\end{array}%
\right.  \label{eq2-13} \\
r_{u}(z)=\left[ 
\begin{array}{cc}
X(z) & \text{ }Y\left( z\right)%
\end{array}%
\right] \left[ 
\begin{array}{c}
u(z) \\ 
y(z)%
\end{array}%
\right] ,  \label{eq2-14} \\
X=\left( A_{L},-B_{L},V^{-1}F,V^{-1}\right) ,Y=\left(
A_{L},-L,V^{-1}F,0\right) .  \label{eq2-8c}
\end{gather}%
Here, $F,L$ are selected such that $A_{F}$ and $A_{L}$ are Schur matrices,
and $W$ and $V$ are invertible matrices \cite{Ding2020,Zhou96}.

\begin{remark}
For the sake of simplicity but without loss of generality, it is often assumed in
our subsequent work that $W=I,V=I.$ 
\end{remark}

Notice that (\ref{eq2-12})\textbf{\ }and (\ref{eq2-14})\textbf{\ }imply 
\begin{gather*}
u(z)=-\hat{Y}(z)r_{y}(z)=-\hat{Y}(z)\hat{X}^{-1}(z)y(z), \\
X(z)u(z)+Y(z)y(z)=r_{u}(z)=0 \\
\Longrightarrow u(z)=-X^{-1}(z)Y(z)y(z).
\end{gather*}%
That is, the control law $K(z)$ in (\ref{eq2-1}) is given by 
\begin{equation}
K(z)=-\hat{Y}(z)\hat{X}^{-1}(z)=-X^{-1}(z)Y(z)  \label{eq2-18}
\end{equation}%
with $\left( \hat{X}(z),\hat{Y}(z)\right) $ and $\left( X(z),Y(z)\right) $
as its RCF and LCF. In other words, the feedback controller in (\ref{eq2-18}%
) is equivalent to an observer-based state feedback controller or an input
residual generator 
\begin{equation}
r_{u}(z)=u(z)-F\hat{x}(z)=X(z)u(z)+Y\left( z\right) y(z).  \label{eq2-15}
\end{equation}%
It is of remarkable interest to notice that both residual generators, (\ref%
{eq2-10})\ and (\ref{eq2-14}), share the same state observer.

We now introduce a special type of LCF and RCF of the plant dynamic. They
are called normalized coprime factorisations, denoted by $\left( \hat{M}_{0},%
\hat{N}_{0}\right) $ and $\left( M_{0},N_{0}\right) ,$ respectively, and
satisfy%
\begin{align}
\left[ 
\begin{array}{cc}
\hat{N}_{0}(z) & \text{ }\hat{M}_{0}(z)%
\end{array}%
\right] \left[ 
\begin{array}{c}
\hat{N}_{0}^{\ast }(z) \\ 
\hat{M}_{0}^{\ast }(z)%
\end{array}%
\right] & =I,  \label{eq2-15a} \\
\left[ 
\begin{array}{cc}
M_{0}^{\ast }(z) & \text{ }N_{0}^{\ast }(z)%
\end{array}%
\right] \left[ 
\begin{array}{c}
M_{0}(z) \\ 
N_{0}(z)%
\end{array}%
\right] & =I.  \label{eq2-15b}
\end{align}%
The state space realizations of $\left( \hat{M}_{0},\hat{N}_{0}\right) $ and 
$\left( M_{0},N_{0}\right) $ are summarized in the following lemma \cite%
{Hoffmann1996}.

\begin{lemma}
\label{Le2-1}Given the system model (\ref{eq2-3}), $\left( \hat{M}_{0},\hat{N%
}_{0}\right) $ and $\left( M_{0},N_{0}\right) $ are LCP and RCP with $%
L=L_{0},W=W_{0},F=F_{0}$ and $V=V_{0},$ as given below:%
\begin{align}
L_{0}& =\left( BD^{T}+APC^{T}\right) \left( I+DD^{T}+CPC^{T}\right) ^{-1},
\label{N-observer gain} \\
W_{0}& =\left( I+DD^{T}+CPC^{T}\right) ^{-1/2},  \label{N-post-filter} \\
F_{0}& =-\left( I+D^{T}D+B^{T}SB\right) ^{-1}\left( D^{T}C+B^{T}SA\right) ,
\label{N-feedback gain} \\
V_{0}& =\left( I+D^{T}D+B^{T}SB\right) ^{-1/2}.  \label{N-pre-filter}
\end{align}%
Here, $P>0,S>0$ solve the following Riccati equations respectively,%
\begin{align}
P& =APA^{T}+BB^{T}-L_{0}\left( I+DD^{T}+CPC^{T}\right) L_{0}^{T},
\label{Ricatti eq-a} \\
S& =A^{T}SA+C^{T}C-F_{0}^{T}\left( I+D^{T}D+B^{T}SB\right) F_{0}.
\label{Ricatti eq-b}
\end{align}
\end{lemma}

The normalized SIR and SKR can be interpreted as a linear quadratic (LQ)
controller and a least squares (LS) estimator, respectively \cite{Tay1998},
and will serve as a useful tool in our subsequent study.

\begin{remark}
Hereafter, we may drop the domain variable $z$\ or $k$ when there is no risk
of confusion.
\end{remark}

\subsection{Youla parameterization and its observer-based realization}

Youla parameterization is a well-established method that, for given the LCFs
and RCFs of the plant $G$ and controller $K,$ parametrizes all stabilizing
controllers as 
\begin{equation}
\setlength{\abovedisplayskip}{5pt}\setlength{\belowdisplayskip}{5pt}K\hspace{%
-2pt}=\hspace{-3pt}-\hspace{-2pt}(\hat{Y}\hspace{-2pt}-\hspace{-2pt}MQ)(\hat{%
X}\hspace{-2pt}+\hspace{-2pt}NQ)^{-1}\hspace{-2pt}=\hspace{-3pt}-\hspace{-2pt%
}(X\hspace{-2pt}+\hspace{-2pt}Q\hat{N})^{-1}(Y\hspace{-2pt}-\hspace{-2pt}Q%
\hat{M})  \label{eq-Youla}
\end{equation}%
where $Q\in \mathcal{RH}_{\infty }$ is the parameter system \cite%
{Tay1998,Zhou96}. For our investigation purpose, the observer-based
realization of (\ref{eq-Youla}) is presented \cite{Ding2020}, 
\begin{gather}
\left\{ 
\begin{array}{l}
\hat{x}(k+1)=A_{F}\hat{x}(k)+Br_{Q}(k)+LW^{-1}r_{y}(k) \\ 
u(k)=F\hat{x}(k)+r_{Q}(k),%
\end{array}%
\right.  \label{eq2-17} \\
r_{Q}=Qr_{y}=Q\left( \hat{M}y-\hat{N}u\right) ,  \notag
\end{gather}%
which is equivalent to%
\begin{equation*}
\left[ 
\begin{array}{c}
u \\ 
y%
\end{array}%
\right] =\left[ 
\begin{array}{c}
-\hat{Y}+\hspace{-2pt}MQ \\ 
\hat{X}+\hspace{-2pt}NQ%
\end{array}%
\right] r_{y}
\end{equation*}%
with $r_{y}$ acting as the latent variable. Now, let the observer (\ref%
{eq2-17})\ be equivalently written into the observer form (\ref{eq2-13}). It
turns out%
\begin{equation}
u-F\hat{x}=Xu+Yy=Qr_{y},  \label{eq2-17a}
\end{equation}%
which is the input residual realization of the Youla parameterization (\ref%
{eq-Youla}).

\subsection{Problem formulation}

\begin{figure}[t]
\centering\includegraphics[width=7cm]{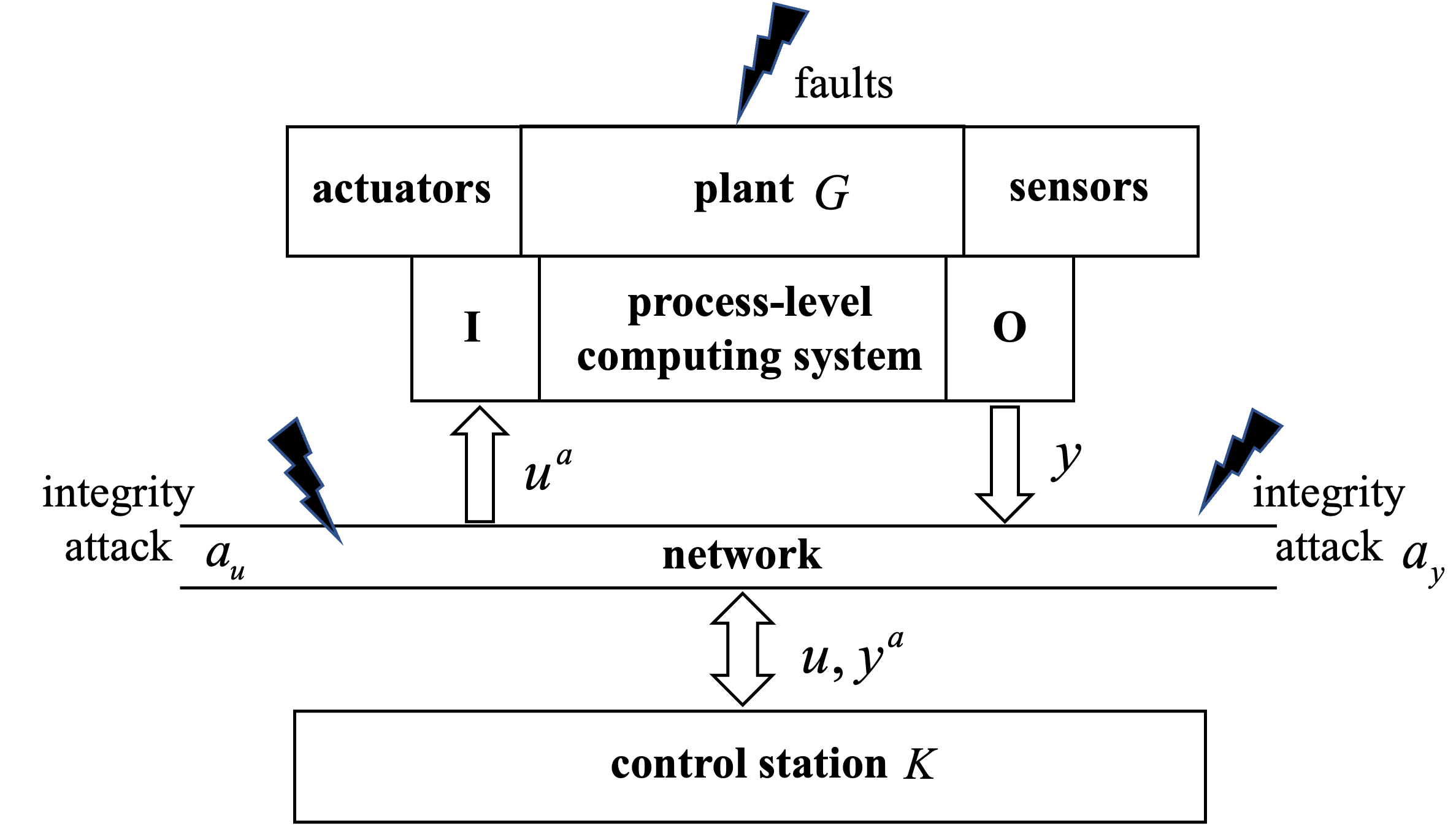}
\caption{The CPCS under consideration}
\label{Fig1}
\end{figure}
As schematically sketched in Fig. \ref{Fig1}, the feedback control system (%
\ref{eq2-1}) is implemented in a CPCS configuration, where the plant $G$ is
connected to the control station $K$ by a communication network. The control
station calculates $u$ and transmits it to the plant. The plant receives
control signal $u^{a}$ and sends the measurement output $y$ to the control
station. Specifically, the CPCS model is described by 
\begin{equation}
\left\{ 
\begin{array}{l}
y\left( z\right) =G\left( z\right) u^{a}\left( z\right)
+f_{0}(z)+r_{y,0}\left( z\right) \\ 
u\left( z\right) =K\left( z\right) y^{a}\left( z\right) +r_{u,0}\left(
z\right) ,%
\end{array}%
\right.  \label{eq2-20}
\end{equation}%
where $K$ is the feedback control law subject to (\ref{eq2-18}), $r_{u,0}=%
\bar{v}$ is the reference signal, $f_{0}$ denotes the influence of faults on
the system dynamic that will be specified in the sequel, and 
\begin{equation}
u^{a}=u+a_{u},y^{a}=y+a_{y}  \label{eq2-22}
\end{equation}%
represent the input and output signals transmitted over the network and
corrupted by the injected attack signals $\left( a_{u},a_{y}\right) .$ On
the assumption that the attacker is able to access the input and output
(transmission) channels and inject signals into the channels, it is supposed
that the cyber-attacks $\left( a_{u},a_{y}\right) $ are modelled by 
\begin{equation}
\left[ 
\begin{array}{c}
a_{u}(z) \\ 
a_{y}(z)%
\end{array}%
\right] =\Pi ^{a}(z)\left[ 
\begin{array}{c}
{u}(z) \\ 
y(z)%
\end{array}%
\right] +\left[ 
\begin{array}{c}
E_{u}\eta _{u}(z) \\ 
E_{y}\eta _{y}(z)%
\end{array}%
\right]  \label{eq2-21}
\end{equation}%
where $\Pi ^{a}$ is a stable dynamic system with $(u,y)$ as the input
variables, $\left( E_{u},E_{y}\right) $ are structure matrices indicating
which channels of the input $u$ and output $y$ are attacked, and $\left(
\eta _{u},\eta _{y}\right) $ are $\ell _{2}$-bounded, unknown and
independent of $(u,y).$ In this regard, $\left( E_{u}\eta _{u},E_{y}\eta
_{y}\right) $ and $\Pi ^{a}\left[ 
\begin{array}{cc}
u^{T} & y^{T}%
\end{array}%
\right] ^{T}$ are called additive and multiplicative attacks, respectively.

\begin{remark}
\label{Rem2-3}The term multiplicative attacks is rarely used in the
literature on secure CPCSs research. We adopt it to emphasize the fact that
such attacks are driven by the process data. Differing from additive
attacks, they affect system eigen-dynamics. For instance, a replay attack
defined by 
\begin{equation*}
y^{a}(k)=y\left( k-T\right) ,k>T(>0),
\end{equation*}%
is a multiplicative attack \cite{DLautomatica2022}, that is, 
\begin{gather*}
a_{y}(z)=-y(z)+z^{-T}y(z)\Longrightarrow y^{a}=z^{-T}y, \\
\Pi ^{a}=\left[ 
\begin{array}{c}
\Pi _{u}^{a} \\ 
\Pi _{y}^{a}%
\end{array}%
\right] ,\Pi _{y}^{a}=\left[ 
\begin{array}{cc}
0 & \text{ }z^{-T}-I%
\end{array}%
\right] 
\end{gather*}%
with $\Pi _{u}^{a}$ arbitrarily designed by attackers. DoS attacks are
random multiplicative attacks as well. In fact, multiplicative attacks are
often implicitly adopted for attack design \cite{TCNS2017Guo,LIU2022674}. 
\end{remark}

For the sake of simplicity, the uncertainty $r_{y,0}$ is assumed to be the
output of an LTI system driven by additive disturbances. Corresponding to
the state space realization of $G,$ (\ref{eq2-2}), the dynamic of $r_{y,0}$
is modelled by 
\begin{equation}
r_{y,0}=G_{d}d,G_{d}=C\left( zI-A\right) ^{-1}E_{d}+F_{d},  \label{eq2-23}
\end{equation}%
where two types of the unknown input $d\in \mathbb{R}^{k_{d}}$ are under
consideration,%
\begin{equation}
\left[ 
\begin{array}{c}
E_{d} \\ 
F_{d}%
\end{array}%
\right] d=\left\{ 
\begin{array}{l}
\left[ 
\begin{array}{c}
w \\ 
v%
\end{array}%
\right] ,\text{ process and measurement noises} \\ 
\left[ 
\begin{array}{c}
\bar{E}_{d}d_{P} \\ 
\bar{F}_{d}d_{S}%
\end{array}%
\right] ,\left\Vert d\right\Vert _{2}=\sqrt{\left\Vert d_{P}\right\Vert
_{2}^{2}+\left\Vert d_{S}\right\Vert _{2}^{2}}\leq \delta _{d}.\text{ }%
\end{array}%
\right.  \label{eq2-23a}
\end{equation}%
$w\sim \mathcal{N}(0,\Sigma _{w}),\nu \sim \mathcal{N}(0,\Sigma _{\nu })$
are white noise series, $\delta _{d}$ is the known upper-bound and $\bar{E}%
_{d},\bar{F}_{d}$ are known matrices. The vector $f_{0}$ represents the
influence of faults on the plant dynamic, which is modelled by 
\begin{equation}
f_{0}(z)=\Pi _{0}^{f}(z)\left[ 
\begin{array}{c}
{u}(z) \\ 
y(z)%
\end{array}%
\right] +{\bar{f}}_{0}(z).  \label{eq2-24}
\end{equation}%
The fault model (\ref{eq2-24}) is a general form of multiplicative faults.
Here, $\Pi _{0}^{f}$ is assumed to be a stable unknown dynamic system and
the additive fault ${\bar{f}}_{0}$ is $\ell _{2}$-bounded.

\begin{remark}
The additive fault affects the system dynamic in a way like the unknown
disturbances modelled by (\ref{eq2-23}). In fact, in the context of fault
diagnosis and fault-tolerant control, multiplicative faults are considerably
critical, because they might impair the system eigen-dynamic and even cause
instability. 
\end{remark}

Along the lines of the three-step explorations procedure towards the
objective of this paper, the following problems will be addressed. At first,
the core ideas, basic methods, and analysis and design schemes are examined
aiming at exploring those, which are particularly appropriate for studying
CPCSs under faults and attacks modelled by (\ref{eq2-20})-(\ref{eq2-24}).
This work will be complemented by reviewing well-established attack
detection and resilient control methods. In particular, it is of importance
to re-formulate these methods in the unified framework of control and
detection, and to examine how far they might be generalized and improved.
The main part of our work is the endeavour to establish a control-theoretic
paradigm, which enables, in a unified manner,

\begin{itemize}
\item analyzing (i) CPCS's structural and information properties, (ii)
system behavior under faults and attacks, (iii) existence conditions of
stealthy attacks, expressed in terms of the system dynamic and information
model and independent of applied detection schemes, and (iii) specified
system vulnerability;

\item developing schemes of simultaneous detection of faults and attacks;

\item exploring fault-tolerant and attack-resilient feedback control systems
with an analysis of the achievable control performance and stability
conditions, and finally

\item addressing system privacy issues in the unified framework of control
and detection.
\end{itemize}

\section{The unified framework of control and detection, and applications to
secure CPCSs}

The unified framework is proposed to deal with control and detection issues
uniformly. In this framework, the control and detection problems are handled
in the space of the system input and output data $\left( u,y\right) .$ Its
groundwork is both the output and input residual signals. Instead of the
loop model (\ref{eq2-20}), the so-called kernel-based system model is
introduced, based on which the system image and residual subspaces are
defined. Thereby, the input and output residual pair $\left(
r_{u},r_{y}\right) $ serves as latent variables that comprise information of
the system dynamic, including the nominal system behavior and uncertainties.
Consequently, the overall system dynamic, during nominal operation and under
faults and attacks, is fully characterized by the residual pair $\left(
r_{u},r_{y}\right) $ and the associated subspaces. In this context,
detection of faults and cyber-attacks, and fault-tolerant and
attack-resilient control can be addressed uniformly. In fact, both the
controllers and detectors are residual-driven, whose core is a (single)
observer. In the framework, the Bezout identity serves as the core
mathematical tool. On its basis, the coprime factorizations of the plant and
controller can be parameterized by a dynamic system triple $\left(
R(z),Q(z),T(z)\right) .$ On the other hand, the system image and residual
subspaces are invariant with respect to the dynamic system triple. This
property is particularly useful to design and construct secure CPCSs.

In this section, we firstly introduce the alternative loop model form, the
kernel-based model, and the concepts of system image and residual subspaces.
They build the basis of our subsequent work. We will then parameterize the
coprime factorizations of the plant and controller, and highlight their
relations to the system image and residual subspaces. Finally, we will
leverage the aforementioned results, serving as a control-theoretic and
mathematical tool, to propose alternative methods to some established
schemes in the research of secure CPCSs.

\subsection{Alternative loop model forms and signal subspaces\label%
{Subsec3-1}}

Consider the control system model (\ref{eq2-20}) with $K$ subject to (\ref%
{eq2-18}). On the assumption that the system is operating attack-free, let
us rewrite the model into 
\begin{equation}
\left\{ 
\begin{array}{l}
r_{y}=\hat{M}y-\hat{N}u=\hat{M}\left( r_{y,0}+f_{0}\right) \\ 
r_{u}=Xu+Yy=Xr_{u,0}=:v,%
\end{array}%
\right.  \label{eq3-1}
\end{equation}%
whose state space representation is given by 
\begin{align}
\hat{x}(k+1)& =A_{L}\hat{x}(k)+B_{K}\left[ 
\begin{array}{c}
u(k) \\ 
y(k)%
\end{array}%
\right] ,  \label{eq3-1b} \\
\left[ 
\begin{array}{c}
r_{u}(k) \\ 
r_{y}(k)%
\end{array}%
\right] & =\left[ 
\begin{array}{c}
V^{-1}F \\ 
C_{K}%
\end{array}%
\right] \hat{x}(k)+\left[ 
\begin{array}{cc}
V^{-1} & \text{ }0 \\ 
WD & \text{ }W%
\end{array}%
\right] \left[ 
\begin{array}{c}
u(k) \\ 
y(k)%
\end{array}%
\right] .  \notag
\end{align}%
The model (\ref{eq3-1}) is called kernel-based model. Note that the model (%
\ref{eq3-1}) is equivalent to 
\begin{gather*}
\left\{ 
\begin{array}{l}
W^{-1}r_{y}=y-\hat{y} \\ 
Vr_{u}=u-\hat{u}%
\end{array}%
\right. \Leftrightarrow \left\{ 
\begin{array}{l}
y=\hat{y}+W^{-1}r_{y} \\ 
u=\hat{u}+Vr_{u},%
\end{array}%
\right. \\
\left[ 
\begin{array}{c}
\hat{u}(k) \\ 
\hat{y}(k)%
\end{array}%
\right] =\left[ 
\begin{array}{c}
F \\ 
C%
\end{array}%
\right] \hat{x}(k)+\left[ 
\begin{array}{c}
0 \\ 
D%
\end{array}%
\right] u(k).
\end{gather*}%
That means, the kernel model is equivalent to the original loop model (\ref%
{eq2-20}). The kernel-based model (\ref{eq3-1}) is characterized by
modelling any uncertainties in the control loop in term of the residual
signals that are online computed. It is capable of fully describing the
system dynamic online and under all possible (incl. under faults and
attacks) operating conditions. Moreover, its core consists of a single
observer, i.e. the state observer (\ref{eq3-1b}). Accordingly, the needed
computation is moderate.

Attributed to the Bezout identity (\ref{eq2-6}), the model (\ref{eq3-1}) is
recast into%
\begin{equation}
\left[ 
\begin{array}{c}
u \\ 
y%
\end{array}%
\right] =\left[ 
\begin{array}{c}
M \\ 
N%
\end{array}%
\right] v+\left[ 
\begin{array}{c}
-\hat{Y} \\ 
\hat{X}%
\end{array}%
\right] r_{y}.  \label{eq3-11}
\end{equation}%
The models (\ref{eq3-1}) and (\ref{eq3-11}) are two dual representation
forms of the closed-loop (\ref{eq2-20}) during attack-free operations.
Observe that the output residual $r_{y},$ 
\begin{align}
r_{y}& =\hat{M}\left( r_{y,0}+f_{0}\right) =\hat{N}_{d}d+f,  \label{eq3-11b}
\\
\hat{N}_{d}& =\hat{M}G_{d}=\left( A-LC,E_{d}-LF_{d},C,F_{d}\right) ,  \notag
\\
f& =\Pi ^{f}\left[ 
\begin{array}{c}
{u} \\ 
y%
\end{array}%
\right] +\bar{f},\Pi ^{f}=\hat{M}\Pi _{0}^{f},\bar{f}=\hat{M}\bar{f}_{0}, 
\notag
\end{align}%
includes exclusively uncertain and faulty dynamics in the plant. In
contrast, the first term in (\ref{eq3-11}) describes the nominal system
dynamic. In other words, the model (\ref{eq3-11}) highlights a notable fact
that the uncertainties (including faults) and the reference signal $v$ act
on two different subspaces in the system input-output data space $\left(
u,y\right) $. In this context, we introduce the following definitions.

\begin{definition}
\label{Def3-1} Given the RCP and LCP $\left( M,N\right) $ and $\left( \hat{M},\hat{N}%
\right) $ of the plant $G,$ the $\mathcal{H}_{2}$ subspaces $\mathcal{I}_{G}$
and $\mathcal{K}_{G}$,%
\begin{align}
\mathcal{I}_{G}& =\left\{ \left[ 
\begin{array}{c}
u \\ 
y%
\end{array}%
\right] :\left[ 
\begin{array}{c}
u \\ 
y%
\end{array}%
\right] \hspace{-2pt}=\hspace{-2pt}\left[ 
\begin{array}{c}
M \\ 
N%
\end{array}%
\right] v,v\in \mathcal{H}_{2}^{p}\right\} ,  \label{eq3-12} \\
\mathcal{K}_{G}& =\left\{ \left[ 
\begin{array}{c}
u \\ 
y%
\end{array}%
\right] :\left[ 
\begin{array}{cc}
-\hat{N} & \hat{M}%
\end{array}%
\right] \left[ 
\begin{array}{c}
u \\ 
y%
\end{array}%
\right] =0,\left[ 
\begin{array}{c}
u \\ 
y%
\end{array}%
\right] \in \mathcal{H}_{2}^{m+p}\right\} ,  \label{eq3-13}
\end{align}%
are called (plant) image and kernel subspaces, respectively. 
\end{definition}

It is a well-known result that $\mathcal{I}_{G}=\mathcal{K}_{G}$ \cite%
{Vinnicombe-book}. The complementary subspace of $\mathcal{I}_{G},\mathcal{I}%
_{G}^{\bot },$ is of considerable importance in our subsequent work. It is
apparent from the relation $\hat{N}M-\hat{M}N=0$ that $\mathcal{I}_{G}^{\bot
}$ is characterized by 
\begin{equation*}
\forall \left[ 
\begin{array}{c}
u \\ 
y%
\end{array}%
\right] \in \mathcal{I}_{G}^{\bot },r_{y}=\left[ 
\begin{array}{cc}
-\hat{N} & \hat{M}%
\end{array}%
\right] \left[ 
\begin{array}{c}
u \\ 
y%
\end{array}%
\right] \neq 0.
\end{equation*}%
On the basis of the Bezout identity (\ref{eq2-6}) and considering that $%
r_{y} $ reflects the uncertain and faulty dynamics in the plant, we
introduce the concept of residual subspace as the complementary subspace of $%
\mathcal{I}_{G}.$

\begin{definition}
\label{Def3-2}Given the RCP $\left( \hat{X},\hat{Y}\right) $ of $K,$ the $%
\mathcal{H}_{2}$ subspace%
\begin{equation}
\mathcal{R}_{G}=\left\{ \left[ 
\begin{array}{c}
u \\ 
y%
\end{array}%
\right] :\left[ 
\begin{array}{c}
u \\ 
y%
\end{array}%
\right] =\left[ 
\begin{array}{c}
-\hat{Y} \\ 
\hat{X}%
\end{array}%
\right] r,r\in \mathcal{H}_{2}^{m}\right\} ,
\end{equation}%
is called canonical residual subspace. 
\end{definition}

\begin{remark}
In the above definition, we've abused the expression \textit{canonical
residual subspace} to distinguish it from the general form of residual
subspace to be introduced in the next subsection.
\end{remark}

It is noteworthy that the canonical residual subspace is the image subspace
of the controller as well, i.e.%
\begin{equation*}
\mathcal{I}_{K}=\left\{ \left[ 
\begin{array}{c}
u \\ 
y%
\end{array}%
\right] :\left[ 
\begin{array}{c}
u \\ 
y%
\end{array}%
\right] =\left[ 
\begin{array}{c}
-\hat{Y} \\ 
\hat{X}%
\end{array}%
\right] r,r\in \mathcal{H}_{2}^{m}\right\} .
\end{equation*}%
Apparently, the residual subspace $\mathcal{R}_{G}$ is the complement of $%
\mathcal{I}_{G},$ and 
\begin{equation}
\mathcal{H}_{2}^{p+m}=\mathcal{I}_{G}\oplus \mathcal{R}_{G}=\mathcal{I}%
_{G}\oplus \mathcal{I}_{K}.  \label{eq3-16}
\end{equation}%
The system response (\ref{eq3-11}) is indeed the consequence of (\ref{eq3-16}%
) and it is sufficient to use the observer-based residual generator (\ref%
{eq2-10}) to detect faults, which pinpoints the data in the residual
subspace.

It is known that $\mathcal{I}_{G}$ and $\mathcal{R}_{G}$ are closed
subspaces in $\mathcal{H}_{2}$ space \cite{Vinnicombe-book}. Hence, by means
of the orthogonal projection method and the concept gap metric \cite%
{Vinnicombe-book} we are able to measure the similarity of two subspaces.
The orthogonal projection method and gap metric are a mathematical tool that
plays an essential role in the unified framework of control and detection 
\cite{LD-Automatica-2020,Ding2026,DL2026}. Below, we summarize some basic
definitions, concepts and computations of the orthogonal projection method,
which will be applied in our subsequent work. For details, the reader is
referred to \cite{Georgiou88,Feintuch_book,Vinnicombe-book}. Let $\left(
M_{0},N_{0}\right) $ be the normalized RCP of $G.$ The operator $\mathcal{P}%
_{\mathcal{I}_{G}},$ 
\begin{equation*}
\mathcal{P}_{\mathcal{I}_{G}}:=\left[ 
\begin{array}{c}
M_{0} \\ 
N_{0}%
\end{array}%
\right] \mathcal{P}_{\mathcal{H}_{2}}\left[ 
\begin{array}{c}
M_{0} \\ 
N_{0}%
\end{array}%
\right] ^{\ast }:\mathcal{H}_{2}\rightarrow \mathcal{H}_{2},
\end{equation*}%
defines an orthogonal projection onto the system image subspace $\mathcal{I}%
_{G}$, where $\mathcal{P}_{\mathcal{H}_{2}}$ is the operator of the
orthogonal projection onto $\mathcal{H}_{2}$ \cite{Georgiou88}. Given system
data $\left( u,y\right) ,$ the distance from $\left( u,y\right) $ to $%
\mathcal{I}_{G}$ is defined as 
\begin{align*}
\setlength{\abovedisplayskip}{6pt}\setlength{\belowdisplayskip}{6pt}%
dist\left( \left[ 
\begin{array}{c}
u \\ 
y%
\end{array}%
\right] ,\mathcal{I}_{G}\right) & =\inf_{\left[ 
\begin{array}{c}
u_{\mathcal{I}_{G}} \\ 
y_{\mathcal{I}_{G}}%
\end{array}%
\right] \in \mathcal{I}_{G}}\left\Vert \left[ 
\begin{array}{c}
u \\ 
y%
\end{array}%
\right] -\left[ 
\begin{array}{c}
u_{\mathcal{I}_{G}} \\ 
y_{\mathcal{I}_{G}}%
\end{array}%
\right] \right\Vert _{2} \\
& =\left\Vert \left( \mathcal{I}-\mathcal{P}_{\mathcal{I}_{G}}\right) \left[ 
\begin{array}{c}
u \\ 
y%
\end{array}%
\right] \right\Vert _{2}.
\end{align*}%
It holds (Pythagorean equation) 
\begin{equation}
\left\Vert \left[ 
\begin{array}{c}
u \\ 
y%
\end{array}%
\right] \right\Vert _{2}^{2}=\left\Vert \mathcal{P}_{\mathcal{I}_{G}}\left[ 
\begin{array}{c}
u \\ 
y%
\end{array}%
\right] \right\Vert _{2}^{2}+\left\Vert \left( \mathcal{I}-\mathcal{P}_{%
\mathcal{I}_{G}}\right) \left[ 
\begin{array}{c}
u \\ 
y%
\end{array}%
\right] \right\Vert _{2}^{2}.  \label{eq3-50}
\end{equation}%
For two systems $G_{1}$ and $G_{2},$ the gap metric between $\mathcal{I}%
_{G_{1}},\mathcal{I}_{G_{2}}\in \mathcal{H}_{2}^{m+p}$ is defined by%
\begin{eqnarray*}
\delta \left( \mathcal{I}_{G_{1}},\mathcal{I}_{G_{2}}\right) &=&\max \left\{ 
\vec{\delta}\left( \mathcal{I}_{G_{1}},\mathcal{I}_{G_{2}}\right) ,\vec{%
\delta}\left( \mathcal{I}_{G_{2}},\mathcal{I}_{G_{1}}\right) \right\} \\
&=&\left\Vert \mathcal{P}_{\mathcal{I}_{G_{1}}}-\mathcal{P}_{\mathcal{I}%
_{G_{2}}}\right\Vert ,
\end{eqnarray*}%
where $\vec{\delta}\left( \mathcal{\cdot },\mathcal{\cdot }\right) $ is
called directed gap, $\mathcal{P}_{\mathcal{I}_{G_{1}}}:\mathcal{H}%
_{2}\rightarrow \mathcal{H}_{2},\mathcal{P}_{\mathcal{I}_{G_{2}}}:\mathcal{H}%
_{2}\rightarrow \mathcal{H}_{2},$ are orthogonal projection operators,
respectively, and $\left\Vert \mathcal{P}_{\mathcal{I}_{G_{1}}}-\mathcal{P}_{%
\mathcal{I}_{G_{2}}}\right\Vert $ is the operator norm \cite%
{Georgiou88,Feintuch_book,Vinnicombe-book}. The gap metric $\delta \left( 
\mathcal{I}_{G_{1}},\mathcal{I}_{G_{2}}\right) $ is of the properties 
\begin{gather}
\delta \left( \mathcal{I}_{G_{1}},\mathcal{I}_{G_{2}}\right) \in \lbrack
0,1],\delta \left( \mathcal{I}_{G_{1}},\mathcal{I}_{G_{2}}\right)
=1\Longrightarrow \mathcal{I}_{G_{1}}\bot \mathcal{I}_{G_{2}},
\label{eq3-51a} \\
\delta \left( \mathcal{I}_{G_{1}},\mathcal{I}_{G_{2}}\right)
=0\Longrightarrow \mathcal{I}_{G_{1}}=\mathcal{I}_{G_{2}},  \label{eq3-51b}
\end{gather}%
and for $\delta \left( \mathcal{I}_{G_{1}},\mathcal{I}_{G_{2}}\right) \in
\left( 0,1\right) ,$%
\begin{equation*}
\delta \left( \mathcal{I}_{G_{1}},\mathcal{I}_{G_{2}}\right) =\vec{\delta}%
\left( \mathcal{I}_{G_{1}},\mathcal{I}_{G_{2}}\right) =\vec{\delta}\left( 
\mathcal{I}_{G_{2}},\mathcal{I}_{G_{1}}\right) .
\end{equation*}

\subsection{Parameterization of the coprime factorizations and Bezout
identity\label{SubsecIII-B}}

Notice that the RCFs and LCFs of $G$ and $K$ are parameterized by four
matrices $\left( F,L,V,W\right) .$ The following theorem describes the
corresponding variations in the RCFs and LCFs caused by varying these four
parameter matrices.

\begin{theorem}
\label{Theo3-1}Given $\left( M_{i},N_{i}\right) ,\left( \hat{M}_{i},\hat{N}%
_{i}\right) ,\left( \hat{X}_{i},\hat{Y}_{i}\right) ,\left(
X_{i},Y_{i}\right) $ satisfying (\ref{eq2-8}), (\ref{eq2-8a}), (\ref{eq2-8d}%
) and (\ref{eq2-8c}), respectively, $i=1,2,$ it holds%
\begin{gather}
\left[ 
\begin{array}{cc}
\hat{M}_{2} & \text{ }\hat{N}_{2}%
\end{array}%
\right] =R_{21}\left[ 
\begin{array}{cc}
\hat{M}_{1} & \text{ }\hat{N}_{1}%
\end{array}%
\right] ,  \label{eq3-33} \\
\left[ 
\begin{array}{c}
M_{2} \\ 
N_{2}%
\end{array}%
\right] =\left[ 
\begin{array}{c}
M_{1} \\ 
N_{1}%
\end{array}%
\right] T_{21},  \label{eq3-34} \\
\left[ 
\begin{array}{cc}
X_{1} & \text{ }Y_{1}%
\end{array}%
\right] =T_{21}\left[ 
\begin{array}{cc}
X_{2} & \text{ }Y_{2}%
\end{array}%
\right] +\bar{T}_{21}\left[ 
\begin{array}{cc}
-\hat{N}_{2} & \text{ }\hat{M}_{2}%
\end{array}%
\right] ,  \label{eq3-35} \\
\left[ 
\begin{array}{c}
-\hat{Y}_{1} \\ 
\hat{X}_{1}%
\end{array}%
\right] =\left[ 
\begin{array}{c}
-\hat{Y}_{2} \\ 
\hat{X}_{2}%
\end{array}%
\right] R_{21}+\left[ 
\begin{array}{c}
M_{2} \\ 
N_{2}%
\end{array}%
\right] \bar{R}_{21},  \label{eq3-36}
\end{gather}%
where $R_{21}=R_{12}^{-1}$ and $T_{21}=T_{12}^{-1}$ are invertible over $%
\mathcal{RH}_{\infty },$
\begin{gather*}
R_{21}=W_{2}\left( I-C\left( zI-A_{L_{2}}\right) ^{-1}\left(
L_{2}-L_{1}\right) \right) W_{1}^{-1}\in \mathcal{RH}_{\infty }, \\
T_{21}=V_{1}^{-1}\left( I+\left( F_{2}-F_{1}\right) \left(
zI-A_{F_{2}}\right) ^{-1}B\right) V_{2}\in \mathcal{RH}_{\infty }, \\
\bar{T}_{21}=V_{1}^{-1}\left( 
\begin{array}{c}
\left( F_{2}-F_{1}\right) \left( zI-A_{F_{2}}\right) ^{-1}L_{2}- \\ 
F_{1}\left( zI-A_{L_{1}}\right) ^{-1}\left( L_{1}-L_{2}\right) 
\end{array}%
\right) \in \mathcal{RH}_{\infty }, \\
\bar{R}_{21}=\left( 
\begin{array}{c}
\left( F_{1}-F_{2}\right) \left( zI-A_{F_{1}}\right) ^{-1}L_{1}- \\ 
F_{2}\left( zI-A_{L_{2}}\right) ^{-1}\left( L_{2}-L_{1}\right) 
\end{array}%
\right) W_{1}^{-1}\in \mathcal{RH}_{\infty }.
\end{gather*}
\end{theorem}

While (\ref{eq3-33})-(\ref{eq3-34}) are known result, the proof of (\ref%
{eq3-35})-(\ref{eq3-36}) can be found in \cite{DLautomatica2022,Ding2026}.

An immediate result of Theorem \ref{Theo3-1} is that the closed-loop dynamic
(\ref{eq3-11}) is subject to 
\begin{gather}
\left[ 
\begin{array}{c}
u \\ 
y%
\end{array}%
\right] =\left[ 
\begin{array}{c}
M_{1} \\ 
N_{1}%
\end{array}%
\right] v+\left[ 
\begin{array}{c}
-\hat{Y}_{1} \\ 
\hat{X}_{1}%
\end{array}%
\right] r_{y}  \notag \\
=\left[ 
\begin{array}{c}
M_{2} \\ 
N_{2}%
\end{array}%
\right] T_{12}v+\left( \left[ 
\begin{array}{c}
-\hat{Y}_{2} \\ 
\hat{X}_{2}%
\end{array}%
\right] R_{21}+\left[ 
\begin{array}{c}
M_{2} \\ 
N_{2}%
\end{array}%
\right] \bar{R}_{21}\right) r_{y}  \notag \\
=\left[ 
\begin{array}{c}
M_{2} \\ 
N_{2}%
\end{array}%
\right] v_{2}+\left( \left[ 
\begin{array}{c}
-\hat{Y}_{2} \\ 
\hat{X}_{2}%
\end{array}%
\right] +\left[ 
\begin{array}{c}
M_{2} \\ 
N_{2}%
\end{array}%
\right] Q_{2}\right) r_{y,2},  \label{eq3-15} \\
v_{2}=T_{12}v,r_{y,2}=R_{21}\left[ 
\begin{array}{cc}
-\hat{N}_{1} & \text{ }\hat{M}_{1}%
\end{array}%
\right] \left[ 
\begin{array}{c}
u \\ 
y%
\end{array}%
\right] ,Q_{2}=\bar{R}_{21}R_{21}^{-1}.  \notag
\end{gather}%
We are interested in a special case that is useful for our subsequent study.

\begin{corollary}
\label{Co3-1}Given $\left( \hat{M},\hat{N}\right) ,\left( M,N\right) ,\left(
X,Y\right) $ and $\left( \hat{X},\hat{Y}\right) ,$ with parameter matrices $%
\left( F,L,V,W\right) ,$ and their normalized realizations, $\left( \hat{M}%
_{0},\hat{N}_{0}\right) ,\left( M_{0},N_{0}\right) ,\left(
X_{0},Y_{0}\right) $, and associated with them, $\left( \hat{X}_{0},\hat{Y}_{0}\right) ,$ it holds%
\begin{gather}
\left[ 
\begin{array}{c}
u \\ 
y%
\end{array}%
\right] =\left[ 
\begin{array}{c}
M \\ 
N%
\end{array}%
\right] v+\left[ 
\begin{array}{c}
-\hat{Y} \\ 
\hat{X}%
\end{array}%
\right] r_{y}  \notag \\
=\left[ 
\begin{array}{c}
M_{0} \\ 
N_{0}%
\end{array}%
\right] v_{0}+\left( \left[ 
\begin{array}{c}
-\hat{Y}_{0} \\ 
\hat{X}_{0}%
\end{array}%
\right] +\left[ 
\begin{array}{c}
M_{0} \\ 
N_{0}%
\end{array}%
\right] Q\right) r_{0},  \label{eq3-17} \\
\left[ 
\begin{array}{c}
v \\ 
r_{y}%
\end{array}%
\right] =\left[ 
\begin{array}{cc}
X & \text{ }Y \\ 
-\hat{N} & \text{ }\hat{M}%
\end{array}%
\right] \left[ 
\begin{array}{c}
u \\ 
y%
\end{array}%
\right] \Longleftrightarrow   \notag \\
\left[ 
\begin{array}{c}
v_{0}-Rr_{0} \\ 
r_{0}%
\end{array}%
\right] =\left[ 
\begin{array}{cc}
X_{0}-R\hat{N}_{0} & \text{ }Y_{0}+R\hat{M}_{0} \\ 
-\hat{N}_{0} & \text{ }\hat{M}_{0}%
\end{array}%
\right] \left[ 
\begin{array}{c}
u \\ 
y%
\end{array}%
\right] ,  \label{eq3-17a} \\
v_{0}=\left[ 
\begin{array}{cc}
X_{0} & \text{ }Y_{0}%
\end{array}%
\right] \left[ 
\begin{array}{c}
u \\ 
y%
\end{array}%
\right] ,r_{0}=\left[ 
\begin{array}{cc}
-\hat{N}_{0} & \text{ }\hat{M}_{0}%
\end{array}%
\right] \left[ 
\begin{array}{c}
u \\ 
y%
\end{array}%
\right] ,  \notag \\
Q=\bar{R}_{0}R_{0}^{-1},R=T_{0}^{-1}\bar{T}_{0},  \notag \\
R_{0}=W_{0}\left( I-C\left( zI-A_{L_{0}}\right) ^{-1}\left( L_{0}-L\right)
\right) W^{-1},  \notag \\
\bar{R}_{0}=\left( 
\begin{array}{c}
\left( F-F_{0}\right) \left( zI-A_{F}\right) ^{-1}L- \\ 
F_{0}\left( zI-A_{L_{0}}\right) ^{-1}\left( L_{0}-L\right) 
\end{array}%
\right) W^{-1},  \notag \\
T_{0}=V^{-1}\left( I+\left( F_{0}-F\right) \left( zI-A_{F_{0}}\right)
^{-1}B\right) V_{0},  \notag \\
\bar{T}_{0}=V^{-1}\left( 
\begin{array}{c}
\left( F_{0}-F\right) \left( zI-A_{F_{0}}\right) ^{-1}L_{0}- \\ 
F\left( zI-A_{L}\right) ^{-1}\left( L-L_{0}\right) 
\end{array}%
\right) .  \notag
\end{gather}
\end{corollary}

\begin{proof}
The proof follows from Theorem \ref{Theo3-1} and (\ref{eq3-15}).
\end{proof}

Observe that\ the second term of (\ref{eq3-15}) is a special case of the
observer-based Youla parameterization (\ref{eq2-17}). In fact, on account of
the general form of the Bezout identity \cite{Zhou96},%
\begin{equation}
\left[ 
\begin{array}{cc}
X+Q\hat{N} & \text{ }Y-Q\hat{M} \\ 
-\hat{N} & \text{ }\hat{M}%
\end{array}%
\right] \left[ 
\begin{array}{cc}
M & \text{ }-\hat{Y}+MQ \\ 
N & \text{ }\hat{X}+NQ%
\end{array}%
\right] =I  \label{eq2-6a}
\end{equation}%
for $Q\in \mathcal{RH}_{\infty },$ the loop model (\ref{eq2-1}), the models (%
\ref{eq3-1}) and (\ref{eq3-11}) are given by%
\begin{gather}
\left\{ 
\begin{array}{l}
y=Gu+r_{y,0} \\ 
u=Ky+r_{u,0},r_{u,0}=\left( X+Q\hat{N}\right) ^{-1}v,%
\end{array}%
\right.  \label{eq2-1a} \\
\left\{ 
\begin{array}{l}
\hat{M}y-\hat{N}u=r_{y} \\ 
Xu+Yy=r_{u},r_{u}=v+Qr_{y},%
\end{array}%
\right.  \label{eq3-1a} \\
\left[ 
\begin{array}{c}
u \\ 
y%
\end{array}%
\right] =\left[ 
\begin{array}{c}
M \\ 
N%
\end{array}%
\right] v+\left( \left[ 
\begin{array}{c}
-\hat{Y} \\ 
\hat{X}%
\end{array}%
\right] +\left[ 
\begin{array}{c}
M \\ 
N%
\end{array}%
\right] Q\right) r_{y}.  \label{eq3-11a}
\end{gather}%
Here, the feedback control law $K$ is described by (\ref{eq-Youla}). It
should be emphasized that the controller (\ref{eq2-1a}) can be realized in
the observer-based form (\ref{eq2-17}), which plays a special role in our
subsequent work.

On account of the Bezout identity (\ref{eq2-6a}) and loop dynamic (\ref%
{eq3-11a}), the canonical residual subspace introduced in Definition \ref%
{Def3-2} is extended to the following general form.

\begin{definition}
\label{Def3-3}Given the RCPs $\left( M,N\right) $ and $\left( \hat{X},\hat{Y}%
\right) $ of $G$ and $K,$ respectively, the $\mathcal{H}_{2}$ subspace $%
\mathcal{R}_{G}$,%
\begin{equation}
\mathcal{R}_{G}=\left\{ 
\begin{array}{c}
\left[ 
\begin{array}{c}
u \\ 
y%
\end{array}%
\right] :\left[ 
\begin{array}{c}
u \\ 
y%
\end{array}%
\right] =\left( \left[ 
\begin{array}{c}
-\hat{Y} \\ 
\hat{X}%
\end{array}%
\right] +\left[ 
\begin{array}{c}
M \\ 
N%
\end{array}%
\right] Q\right) r \\ 
r\in \mathcal{H}_{2}^{m},Q\in \mathcal{RH}_{\infty }%
\end{array}%
\right\} ,  \label{eq3-14}
\end{equation}%
is called residual subspace. 
\end{definition}

Corollary \ref{Co3-2} is an immediate result of Theorem \ref{Theo3-1}.

\begin{corollary}
\label{Co3-2} Given the RCPs $\left( M,N\right) $ and $\left( \hat{X},\hat{Y}%
\right) $ of $%
G$ and $K$, the image subspace (\ref{eq3-12}) and residual
subspace (\ref%
{eq3-14}) are invariant with respect to $\left(
F,L,V,W\right) .$
\end{corollary}

The proof follows directly from (\ref{eq3-34}) and (\ref{eq3-36}) in Theorem %
\ref{Theo3-1} as well as Definitions \ref{Def3-2} and \ref{Def3-3}. Thus, it
is omitted.

\subsection{Applications to secure CPCS study}

In the literature on secure CPCSs, several active defence schemes have
emerged to detect cyber-attacks, in particular those stealthy attacks \cite%
{Survey-Segovia-Ferreira2024}. Among others, Moving Target Defense (MTD) 
\cite{Kanellopoulos2019MTD}, dynamic watermarking \cite%
{Mo2015-Watermarked-detection}, and auxiliary subsystem \cite%
{MT-method-CDC2015,Zhang-CDC2017} methods are well-established. In this
subsection, the aforementioned results are leveraged to the MTD and
watermark methods, while the auxiliary subsystem scheme will be addressed in
Subsection \ref{SubsectionV-A}. We will also exemplify the concept of
output-nulling invariant subspace \cite{Anderson1975}, which builds the
control-theoretic basis for designing zero-dynamics attacks, a well-studied
class of stealthy attacks \cite{Sui2021}, in the unified framework of
control and detection. To begin with, the unified framework is exploited
aiming at secure CPCS configurations.

\subsubsection{Secure CPCS configurations\label{subsection3-3-1}}

In their survey paper \cite{Survey-Segovia-Ferreira2024}, Segovia-Ferreira
et al. pointed out that modifying the system architecture, e.g. using
different hardware, software, firmware, or protocols, is capable of
improving the resilience of the system to absorb or survive the attack
impact. To author's best knowledge, there are rarely reported efforts on
modifying control system configurations for the attack-resilient purpose. We
notice that the observer-based realization of a Youla parameterized
controller (\ref{eq2-17}) enables two-site control architectures. On this
basis, it is possible to implement a controller in two separate units
running on the plant and the control station, respectively. Such
architectures naturally obscure the overall loop dynamic behavior from an
attacker. Specifically, the observer-based realization (\ref{eq2-17}) offers
two possible variations,

\begin{itemize}
\item Variation I: on the plant side, the observer-based system%
\begin{equation*}
\left\{ 
\begin{array}{l}
\hat{x}(k+1)=A_{L}\hat{x}(k)+B_{K}\left[ 
\begin{array}{c}
u(k) \\ 
y(k)%
\end{array}%
\right] \\ 
u(k)=F\hat{x}(k)+r_{Q}(k)+v(k) \\ 
r_{y}(k)=C_{K}\hat{x}(k)+D_{K}\left[ 
\begin{array}{c}
u(k) \\ 
y(k)%
\end{array}%
\right] ,%
\end{array}%
\right.
\end{equation*}%
is implemented and the residual $r_{y}$ is transmitted to the control
station, and on the control station, $Qr_{y}+v$ is calculated and
transmitted to the plant;

\item Variation II: on the control station side, the system 
\begin{equation}
\left\{ 
\begin{array}{l}
\hat{x}_{0}(k+1)=A_{L}\hat{x}_{0}(k)+B_{K}\left[ 
\begin{array}{c}
u_{0}(k) \\ 
y(k)%
\end{array}%
\right] \\ 
u_{0}(k)=F\hat{x}_{0}(k)+v(k)%
\end{array}%
\right.  \label{eq3-31}
\end{equation}%
is realized and $u_{0}$ is transmitted to the plant, where the system 
\begin{equation*}
\left\{ 
\begin{array}{l}
\hat{x}(k+1)=A_{L}\hat{x}(k)+B_{K}\left[ 
\begin{array}{c}
u(k) \\ 
y(k)%
\end{array}%
\right] \\ 
r_{y}(k)=W\left( y(k)-C\hat{x}(k)-Du(k)\right) \\ 
u(k)=u_{0}(k)+r_{Q}(k)%
\end{array}%
\right.
\end{equation*}%
and $Qr_{y}$ are implemented, and the output $y$ is sent to the control
station.
\end{itemize}

Variation I is of the advantages that (i) higher security, since the
essential system performance like the system stability is guaranteed also in
case of communication failure, (ii) system privacy preserving, where the
residual signal $r_{y}$ instead of $y$ is transmitted over the communication
network. A potential and significant application is the cloud
computing-based fault-tolerant control of CPCSs \cite{Salah2025}, where the
fault-tolerant control system $Qr_{y}$ is online computed in the cloud
computing system. With the aid of Variation II, the degree of design freedom
can be remarkably increased so that capable attack and fault detection,
tolerant and resilient control systems are designed, as exemplified in the
subsequent sections. Note that the controller (\ref{eq3-31}) implemented on
the control station can be alternatively an output feedback controller $%
u_{0}=K_{0}y+\bar{v}.$ It is noteworthy that the above two system
configurations are often adopted in the context of plug-and-play (PnP)
control systems. In engineering practice, control systems are often equipped
with a default controller whose configuration and parameters are fixed. In
order to meet the control performance requirements, the PnP configuration
and control strategy are often employed to recover the system performance
when the system is suffered from faults and disturbances \cite%
{LLDYP-2019,LYKDP-2018,Salah2025}. The following theorem is an alternative
form of the Youla parameterization that serves as a theoretic basis of a PnP
controller.

\begin{theorem}
\label{theo3-2}Given the feedback control system (\ref{eq2-1}) equipped with
a default (stabilizing) controller $u_{0}=K_{0}y,$ then 
\begin{equation}
u=u_{0}+Qr_{y},r_{y}=\hat{M}y-\hat{N}u,  \label{eq3-17}
\end{equation}%
$Q\in \mathcal{RH}_{\infty },$ gives a parameterization form of all
stabilizing controllers, where $\left( \hat{M},\hat{N}\right) $ is an LCP of 
$G.$ 
\end{theorem}

\begin{proof}
In \cite{LYKDP-2018}, a proof based on the observer-based implementation
form is given. Below, an alternative proof is delineated using the kernel
model (\ref{eq3-1}). As stabilizing controller, $K_{0}$ can be written as $%
K_{0}\hspace{-2pt}=-(X+\hspace{-2pt}Q_{0}\hat{N})^{-1}(Y\hspace{-2pt}-%
\hspace{-2pt}Q_{0}\hat{M}).$ Let the Youla parameterization form of $K$ be
written into%
\begin{equation*}
K=-(X+\hspace{-2pt}Q_{0}\hat{N}+\bar{Q}\hat{N})^{-1}(Y\hspace{-2pt}-\hspace{%
-2pt}Q_{0}\hat{M}-\bar{Q}\hat{M}),\bar{Q}\in \mathcal{RH}_{\infty }.
\end{equation*}%
It turns out%
\begin{gather*}
(X+\hspace{-2pt}Q_{0}\hat{N}+\bar{Q}\hat{N})u=-(Y\hspace{-2pt}-\hspace{-2pt}%
Q_{0}\hat{M}-\bar{Q}\hat{M})y\Longrightarrow \\
u=-(X+\hspace{-2pt}Q_{0}\hat{N})^{-1}\left( (Y\hspace{-2pt}-\hspace{-2pt}%
Q_{0}\hat{M})y+\bar{Q}r_{y}\right) .
\end{gather*}%
Imposing $\bar{Q}=-(X+\hspace{-2pt}Q_{0}\hat{N})Q,Q\in \mathcal{RH}_{\infty
},$ finally results in the parameterization form (\ref{eq3-17}).
\end{proof}

The parameterization (\ref{eq3-17}) will be adopted in our subsequent work
on fault-tolerant and attack-resilient control.

\subsubsection{An application to MTD technique\label{subsection3-3-2}}

Roughly speaking, the MTD technique introduces time-varying or deliberately
switched system configurations to obscure system information over time \cite%
{Kanellopoulos2019MTD,MT-method-IEEE-TAC2021}. \cite%
{Survey-Segovia-Ferreira2024} describes the MTD technique as a strategy of
increasing the complexity and cost of attack design and limiting the
exposure of the system vulnerabilities. Below, we briefly highlight the
advantageous application of the unified framework to the design of MTD
schemes.

Consider the closed-loop dynamic (\ref{eq3-11a}) with the stabilizing
controller (\ref{eq-Youla}). According to Corollary \ref{Co3-2}, the process
data $\left( u,y\right) $ are invariant with respect to the variations of
the four parameters $\left( F,L,V,W\right) .$ Specifically, it follows from
Theorem \ref{Theo3-1} that the variations of $\left( F,L,V,W\right) $ can be
represented by a pre-filter $T(z),$ a post-filter $R(z)$ and a
parameterization matrix $Q(z).$ Since 
\begin{align*}
G& =\left( R\hat{M}\right) ^{-1}R\hat{N}=NT\left( MT\right) ^{-1}, \\
K& =\left( -\hat{Y}R^{-1}+MQ\right) \left( \hat{X}R^{-1}+NQ\right) ^{-1} \\
& =\left( T^{-1}X-Q\hat{N}\right) ^{-1}\left( T^{-1}Y+Q\hat{M}\right) ,
\end{align*}%
with $\left( R,T\right) \in \mathcal{RH}_{\infty }$ being invertible over $%
\mathcal{RH}_{\infty },$ the closed-loop dynamic (\ref{eq3-11a}) is
generally parameterized by a triple of $\mathcal{RH}_{\infty }$-systems $%
\left( Q,R,T\right) $ with $\left( R,T\right) $ invertible over $\mathcal{RH}%
_{\infty }.$ This result can be further extended to linear time-varying
(LTV) systems. To be specific, it is a known result that the general form of
the Bezout identity (\ref{eq2-6a}) exists for LTV systems as well with $%
\mathcal{Q}$ as the parameterization matrix, and the coprime factorizations
presented in Subsection \ref{SubsecII-A} can be well extended to the case
with LTV post- and pre-filters, $\mathcal{R}$ and $\mathcal{T}$ \cite%
{Feintuch_book,LDZ2024}, where $\left( \mathcal{Q},\mathcal{R},\mathcal{T}%
\right) $ are the operators representing stable LTV systems with the state
space representations%
\begin{align*}
\mathcal{Q}& :\left\{ 
\begin{array}{l}
x_{Q}(k+1)=A_{Q}(k)x_{Q}(k)+B_{Q}(k)r_{y}(k) \\ 
r_{Q}(k)=C_{Q}(k)x_{Q}(k)+D_{Q}(k)r_{y}(k),%
\end{array}%
\right. \\
\mathcal{R}& :\left\{ 
\begin{array}{l}
x_{R}(k+1)=A_{R}(k)x_{R}(k)+B_{R}(k)r_{y}(k) \\ 
r_{R}(k)=C_{R}(k)x_{R}(k)+D_{R}(k)r_{y}(k),%
\end{array}%
\right. \\
\mathcal{T}& :\left\{ 
\begin{array}{l}
x_{T}(k+1)=A_{T}(k)x_{T}(k)+B_{T}(k)v(k) \\ 
\bar{v}(k)=C_{T}(k)x_{T}(k)+D_{T}(k)v(k),%
\end{array}%
\right.
\end{align*}%
and $\left( \mathcal{R},\mathcal{T}\right) $ invertible. With the aid of the
above results, we are able to design more capable MTD systems. For instance,
the LTV post- and pre-filters $\left( \mathcal{R},\mathcal{T}\right) $
generalize time-varying input/output transformations $\left(
R(k),T(k)\right) $, a state of the art MTD scheme applied in the research of
secure CPCSs \cite{Survey-Segovia-Ferreira2024}. In this case, the
controller and observer invert $\left( \mathcal{R},\mathcal{T}\right) $
dynamically, while the attacker sees a time-varying system. The use of LTV
parameter system $\mathcal{Q}$ injects unpredictability in the closed-loop
dynamic so that it is hard for attackers to identify the plant model, the
controller and observer. Mathematically, $\left( \mathcal{Q},\mathcal{R},%
\mathcal{T}\right) $ build an admissible set in the Banach space of bounded
causal linear operators, which is vastly larger than the set formed by
time-varying input/output transformations $\left( R(k),T(k)\right) .$

\subsubsection{An application to watermark technique\label{subsection3-3-3}}

The basic idea of watermarking consists in injecting an excitation signal as
a watermark, often a stochastic series, into the control inputs. The
injected watermark causes measurable correlations that expose malicious
tampering when the outputs fail to reflect the watermark \cite%
{Mo2015-Watermarked-detection}. The watermark technique is a popular method
of detecting cyber-attacks, in particular dealing with the so-called replay
attacks. Below, we propose a watermark-based detection scheme in the unified
framework.

Suppose that the CPCS with noises modelled by (\ref{eq2-23a}) is described
by its kernel model 
\begin{equation*}
\left\{ 
\begin{array}{l}
r_{y}=\hat{M}_{n}y-\hat{N}_{n}u=r_{y,n} \\ 
r_{u}=Xu+Yy=X_{n}u+Y_{n}y+Q_{n}r_{y,n},%
\end{array}%
\right. 
\end{equation*}%
where $\left( \hat{M}_{n},\hat{N}_{n}\right) ,\left( X_{n},Y_{n}\right) $
are the LCPs of the plant and the controller with Kalman-filter gain $%
L=L_{K},r_{y,n}\sim \mathcal{N}(0,\Sigma _{r_{y,n}})$ is the generated
innovation series, and $r_{u}$ is given by 
\begin{gather*}
u=Ky=-X^{-1}Yy\Longrightarrow  \\
r_{u}=Xu+Yy=X_{n}u+Y_{n}y+Q_{n}r_{y,n}.
\end{gather*}%
Here, the last equation is attributed to Theorem \ref{Theo3-1} with 
\begin{equation*}
Q_{n}=-F\left( zI-A_{L}\right) ^{-1}(L-L_{K}),
\end{equation*}%
and $\left( F,L\right) $ denoting the setting of $\left( X,Y\right) $. The
system dynamic is described by 
\begin{gather}
\left[ 
\begin{array}{c}
u \\ 
y%
\end{array}%
\right] =\left[ 
\begin{array}{c}
M \\ 
N%
\end{array}%
\right] r_{u_{0}}+\left[ 
\begin{array}{c}
-\hat{Y}_{n}+MQ_{n} \\ 
\hat{X}_{n}+NQ_{n}%
\end{array}%
\right] r_{y,n},  \notag \\
\Longrightarrow y(k)\sim \mathcal{N}(r_{u_{0}}(k),\Sigma
_{y_{r}}(k)),r_{u_{0}}=X_{n}u+Y_{n}y  \label{eq3-18a}
\end{gather}%
with $\Sigma _{y_{r}}$ as the covariance matrix of $y_{r},$%
\begin{equation}
y_{r}=\left( \hat{X}_{n}+NQ_{n}\right) r_{y,n}.  \label{eq3-18b}
\end{equation}%
On the assumption that the CPCS is working in the steady state, a replay
attack is modelled by (refer to Remark \ref{Rem2-3}) 
\begin{equation}
y^{a}(k)=y(k)+a_{y}(k),a_{y}(k)=-y(k)+y\left( k-T\right) .  \label{eq3-18c}
\end{equation}%
We now configure the CPCS as follows. Suppose that the control station
receives $y^{a}$ from the plant and generates the residual and control
signals, 
\begin{equation*}
r_{Q}=Q\left[ 
\begin{array}{cc}
-\hat{N}_{n} & \text{ }\hat{M}_{n}%
\end{array}%
\right] \left[ 
\begin{array}{c}
u \\ 
y^{a}%
\end{array}%
\right] ,u=Ky^{a},
\end{equation*}%
and send the signal $u_{Q},$ 
\begin{equation*}
u_{Q}=u+r_{Q},
\end{equation*}%
to the plant, where $r_{Q}$ serves as a watermark signal with $Q\in \mathcal{%
RH}_{\infty }.$ On the plant side, the input signal $u$ is recovered by 
\begin{eqnarray*}
u &=&u_{Q}^{a}-Qr_{y}, \\
r_{y} &=&\left[ 
\begin{array}{cc}
-\hat{N}_{n} & \text{ }\hat{M}_{n}%
\end{array}%
\right] \left[ 
\begin{array}{c}
u_{Q}^{a}-Qr_{y} \\ 
y%
\end{array}%
\right] .
\end{eqnarray*}%
To detect replay attacks, a detector, 
\begin{equation}
r_{u}=\left[ 
\begin{array}{cc}
X_{n} & \text{ }Y_{n}%
\end{array}%
\right] \left[ 
\begin{array}{c}
u_{Q}^{a}-Qr_{y} \\ 
y%
\end{array}%
\right] -Q_{n}r_{y},  \label{eq3-18}
\end{equation}%
is implemented on the plant side. It is apparent that in case of
attack-free, $r_{u}=0.$ Against it, under replay attacks it turns out%
\begin{align}
r_{u}& =\left[ 
\begin{array}{cc}
X_{n} & \text{ }Y_{n}%
\end{array}%
\right] \left[ 
\begin{array}{c}
Ky^{a}+r_{Q}+a_{u}-Qr_{y} \\ 
y%
\end{array}%
\right] -Q_{n}r_{y}  \label{eq3-20} \\
& =X_{n}a_{u}-Y_{n}a_{y}+\left( Q-Q_{n}\right) \left( \hat{N}_{n}a_{u}+\hat{M%
}_{n}a_{y}\right)   \notag \\
& =\left[ 
\begin{array}{cc}
X_{n}-\bar{Q}\hat{N}_{n} & \text{ }Y_{n}+\bar{Q}\hat{M}_{n}%
\end{array}%
\right] \left[ 
\begin{array}{c}
a_{u} \\ 
z^{-T}y%
\end{array}%
\right] ,\bar{Q}=Q-Q_{n}.  \label{eq3-19}
\end{align}%
It is of interest to notice that the closed-loop dynamic under the attacks
is governed by%
\begin{equation*}
\left[ 
\begin{array}{c}
u \\ 
y%
\end{array}%
\right] =\left[ 
\begin{array}{c}
M \\ 
N%
\end{array}%
\right] r_{u}+\left[ 
\begin{array}{c}
-\hat{Y}_{n}+MQ_{n} \\ 
\hat{X}_{n}+NQ_{n}%
\end{array}%
\right] r_{y,n},
\end{equation*}%
which yields, recalling (\ref{eq3-18a})-(\ref{eq3-18c}), 
\begin{gather*}
r_{u}(k)\sim \mathcal{N}(r_{a_{u}}(k),\Sigma _{r_{a_{y}}}(k)), \\
r_{a_{u}}=\left( X_{n}-\bar{Q}\hat{N}_{n}\right) a_{u},r_{a_{y}}=\left(
Y_{n}+\bar{Q}\hat{M}_{n}\right) r_{y,a}, \\
r_{y,a}\sim \mathcal{N}(0,2\Sigma _{r_{y,n}}),
\end{gather*}%
where $\Sigma _{r_{a_{y}}}$ is the covariance matrix of $r_{a_{y}}.$ Note
that the above relations are attributed to the facts that 
\begin{gather*}
a_{y}(k)=-y(k)+y\left( k-T\right) =r_{y,a}(k)\sim \mathcal{N}(0,2\Sigma
_{r_{y,n}}), \\
\mathcal{E}y(k)=\mathcal{E}y(k-T)=r_{u_{0}},
\end{gather*}%
and $a_{u}$ is a deterministic signal. As a result, the detection logic is
defined as%
\begin{equation*}
r_{u}=\left\{ 
\begin{array}{l}
0,\text{ attack-free} \\ 
\sim \mathcal{N}(r_{a_{u}}(k),\Sigma _{r_{a_{y}}}(k))\neq 0,\text{ alarm.}%
\end{array}%
\right. 
\end{equation*}%
We would like to call the reader's attention to a useful by-product of the
detection scheme. Selecting $\bar{Q}$ such that 
\begin{equation*}
\left\Vert \left[ 
\begin{array}{cc}
X_{n} & \text{ }Y_{n}%
\end{array}%
\right] -\bar{Q}\left[ 
\begin{array}{cc}
-\hat{N}_{n} & \text{ }\hat{M}_{n}%
\end{array}%
\right] \right\Vert _{\infty }\rightarrow \min 
\end{equation*}%
enhances the attack resilience of the system.

\subsubsection{Zero-dynamics\label{subsection3-3-4}}

Zero-dynamics attacks are a type of stealthy attacks that have been
intensively researched in the literature, dedicated to attack detection and
design \cite{TEIXEIRA-zero-attack_2015,Hoehn2016,Sui2021}. Generally
speaking, this type of attacks exploits the plant's inherent zero dynamics
to remain perfectly stealthy against (output) residual-based detectors. As
well delineated in \cite{Sui2021}, the control-theoretic basis for
constructing zero-dynamics attacks is the output-nulling invariant subspace 
\cite{Anderson1975}. This motivates us to examine the relation between the
output-nulling invariant subspace and the RCF/LCF of $G$ as well as the
associated image/kernel subspaces. To this end, consider the output-nulling
invariant subspace $\mathcal{V}$ in the sense of \cite{Anderson1975}, for
some $u,v\in \mathcal{V}$%
\begin{equation*}
Av+Bu\mathcal{\subset V},Cv+Du=0,
\end{equation*}%
which is equivalent to 
\begin{equation}
\left( A+BF\right) \mathcal{V\subset V},\left( C+DF\right) \mathcal{V}=0,
\label{eq3-21}
\end{equation}%
for some $F.$ Subsequently, associated to each $x\in \mathcal{V},$ we have
an input-output pair $\left( u,y\right) $ so that $y=0.$ According to (\ref%
{eq3-21}), this can be expressed in terms of the RCF of $G$ as 
\begin{equation*}
u=Mv,Nv=0\Longrightarrow y=NM^{-1}u=0.
\end{equation*}%
Thus, the set of the input and output pairs associated with the
output-nulling invariant subspace corresponds to a subspace in the image
subspace $\mathcal{I}_{G}$%
\begin{equation*}
\left\{ 
\begin{array}{c}
\left[ 
\begin{array}{c}
u \\ 
y%
\end{array}%
\right] :\left[ 
\begin{array}{c}
u \\ 
y%
\end{array}%
\right] =\left[ 
\begin{array}{c}
M \\ 
N%
\end{array}%
\right] v, \\ 
Nv=0,v\neq 0%
\end{array}%
\right\} \subset \mathcal{I}_{G}.
\end{equation*}%
The dual form is the input-nulling invariant subspace. The associated
input-output pair $\left( u,y\right) $ satisfies, for some $u,$ 
\begin{equation*}
\hat{N}u=0,\hat{M}y=\hat{N}u=0\Longrightarrow y=0.
\end{equation*}%
It is apparent that the set of all associated input-output pairs builds a
subspace in the kernel subspace $\mathcal{K}_{G},$%
\begin{equation*}
\left\{ 
\begin{array}{c}
\left[ 
\begin{array}{c}
u \\ 
y%
\end{array}%
\right] :\left[ 
\begin{array}{cc}
-\hat{N} & \text{ }\hat{M}%
\end{array}%
\right] \left[ 
\begin{array}{c}
u \\ 
y%
\end{array}%
\right] =0, \\ 
\hat{N}u=0,u\neq 0%
\end{array}%
\right\} \subset \mathcal{K}_{G}.
\end{equation*}%
Motivated by the recent endeavours to study local stealthy attacks \cite%
{TEIXEIRA-zero-attack_2015,MikhaylenkoTAC2022,Sui2021}, we now extend the
aforementioned results to the case, when only partial input and output
channels are under consideration. Specifically, let 
\begin{equation}
B=\left[ 
\begin{array}{cc}
B_{1} & \text{ }B_{2}%
\end{array}%
\right] ,C=\left[ 
\begin{array}{c}
C_{1} \\ 
C_{2}%
\end{array}%
\right] ,D=\left[ 
\begin{array}{cc}
D_{11} & \text{ }D_{12} \\ 
D_{21} & \text{ }D_{22}%
\end{array}%
\right]  \label{eq3-23}
\end{equation}%
where the input channels modelled by $B_{1}$ are of interest, e.g.
representing the control inputs that are attacked, and $\left( C_{1},\left[ 
\begin{array}{cc}
D_{11} & \text{ }D_{12}%
\end{array}%
\right] \right) $ represents the output channels that should not be affected
by the inputs acted on the input channels $B_{1}$ like attacks. Below, we
examine the existence conditions expressed in terms of RCF/LCF and
image/kernel subspaces schematically and without rigorous mathematical
details. For our purpose, corresponding to (\ref{eq3-23}) $\left( u,y\right) 
$ are written into%
\begin{equation*}
u=\left[ 
\begin{array}{c}
u_{1} \\ 
u_{2}%
\end{array}%
\right] =\left[ 
\begin{array}{c}
F_{1}x+v_{1} \\ 
0%
\end{array}%
\right] ,y=\left[ 
\begin{array}{c}
y_{1} \\ 
y_{2}%
\end{array}%
\right] ,
\end{equation*}%
which yields an SIR of $G,$%
\begin{gather*}
\left[ 
\begin{array}{c}
u \\ 
y%
\end{array}%
\right] =\left[ 
\begin{array}{c}
M \\ 
N%
\end{array}%
\right] v=\left[ 
\begin{array}{c}
M_{1} \\ 
N_{1} \\ 
N_{2}%
\end{array}%
\right] v_{1}, \\
M_{1}=\left[ 
\begin{array}{c}
F_{1}\left( zI-A-B_{1}F_{1}\right) ^{-1}B_{1}+I \\ 
0%
\end{array}%
\right] , \\
\left[ 
\begin{array}{c}
N_{1} \\ 
N_{2}%
\end{array}%
\right] =\left[ 
\begin{array}{c}
\left( C_{1}+D_{11}F_{1}\right) \left( zI-A-B_{1}F_{1}\right)
^{-1}B_{1}+D_{11} \\ 
\left( C_{2}+D_{21}F_{1}\right) \left( zI-A-B_{1}F_{1}\right)
^{-1}B_{1}+D_{21}%
\end{array}%
\right] .
\end{gather*}%
As a result, $y_{1}=0,$ if and only if, for some $v_{1}\neq 0,N_{1}v_{1}=0.$
It is obvious that%
\begin{equation*}
\left\{ 
\begin{array}{c}
\left[ 
\begin{array}{c}
u \\ 
y%
\end{array}%
\right] :\left[ 
\begin{array}{c}
u \\ 
y%
\end{array}%
\right] =\left[ 
\begin{array}{c}
M \\ 
N%
\end{array}%
\right] v=\left[ 
\begin{array}{c}
u_{1} \\ 
0 \\ 
0 \\ 
y_{2}%
\end{array}%
\right] , \\ 
N_{1}v_{1}=0,v_{1}\neq 0%
\end{array}%
\right\} \subset \mathcal{I}_{G}.
\end{equation*}%
The dual results are summarized as follows. For some $u_{1}\neq 0,u_{2}=0,$%
\begin{gather*}
\hat{N}u=\left[ 
\begin{array}{c}
\hat{N}_{1}u_{1} \\ 
\hat{N}_{2}u_{1}%
\end{array}%
\right] =\left[ 
\begin{array}{c}
0 \\ 
\hat{N}_{2}u_{1}%
\end{array}%
\right] , \\
\hat{M}y=\left[ 
\begin{array}{cc}
\hat{M}_{1} & \text{ }\hat{M}_{2}%
\end{array}%
\right] \left[ 
\begin{array}{c}
y_{1} \\ 
y_{2}%
\end{array}%
\right] =\left[ 
\begin{array}{cc}
\hat{M}_{11} & \text{ }\hat{M}_{12} \\ 
\hat{M}_{21} & \text{ }\hat{M}_{22}%
\end{array}%
\right] \left[ 
\begin{array}{c}
y_{1} \\ 
y_{2}%
\end{array}%
\right] , \\
\hat{N}_{1}=C_{1}\left( zI-A+L_{1}C_{1}\right) ^{-1}\left(
B_{1}-L_{1}D_{11}\right) +D_{11}, \\
\hat{N}_{2}=C_{2}\left( zI-A+L_{1}C_{1}\right) ^{-1}\left(
B_{1}-L_{1}D_{21}\right) +D_{21}, \\
\hat{M}_{1}=\left[ 
\begin{array}{c}
I \\ 
0%
\end{array}%
\right] -C\left( zI-A+L_{1}C_{1}\right) ^{-1}L_{1},\hat{M}_{2}=\left[ 
\begin{array}{c}
0 \\ 
I%
\end{array}%
\right] ,
\end{gather*}%
which finally results in 
\begin{gather*}
\hat{N}u=\hat{M}y=\left[ 
\begin{array}{c}
0 \\ 
y_{2}%
\end{array}%
\right] \Longrightarrow y_{1}=0,y_{2}=\hat{N}_{2}u_{1} \\
\left\{ 
\begin{array}{c}
\left[ 
\begin{array}{c}
u \\ 
y%
\end{array}%
\right] :\left[ 
\begin{array}{cc}
-\hat{N} & \text{ }\hat{M}%
\end{array}%
\right] \left[ 
\begin{array}{c}
u \\ 
y%
\end{array}%
\right] =0\Leftrightarrow \\ 
\left[ 
\begin{array}{cc}
-\hat{N}_{2} & \text{ }I%
\end{array}%
\right] \left[ 
\begin{array}{c}
u_{1} \\ 
y_{2}%
\end{array}%
\right] =0,\hat{N}_{1}u_{1}=0,u_{1}\neq 0%
\end{array}%
\right\} \subset \mathcal{K}_{G}.
\end{gather*}%
The above results are achieved on the assumption that merely system model
knowledge $\left( A,B_{1},C_{1},D_{11}\right) $ is used. In Subsection \ref%
{Subsection IV-C}, application of the above results to constructing
zero-dynamics attacks will be addressed.

\section{Analysis of System Dynamics under Cyber-attacks and Faults}

\label{sec3}In this section, dynamics of the system (\ref{eq2-20}) and their
impairment under cyber-attacks and faults are firstly analyzed. On this
basis, the issues of stealthy attacks like definitions, detection conditions
and their construction are closely examined.

\subsection{Analysis of faulty dynamics}

We begin with an analysis of the closed-loop dynamic under faults modelled
by (\ref{eq2-24}).

\begin{theorem}
\label{Theo4-1}Given the kernel-based loop model (\ref{eq3-1}) with the
fault modelled by (\ref{eq2-24}), then the closed-loop system is stable if 
\begin{equation}
\Phi ^{f}=\left( I-\Pi ^{f}\left[ 
\begin{array}{c}
-\hat{Y} \\ 
\hat{X}%
\end{array}%
\right] \right) ^{-1}\in \mathcal{RH}_{\infty }.  \label{eq3-2}
\end{equation}%
Moreover, the closed-loop dynamic is governed by 
\begin{gather}
\left[ 
\begin{array}{c}
u \\ 
y%
\end{array}%
\right] =\left[ 
\begin{array}{c}
M \\ 
N%
\end{array}%
\right] v+\left[ 
\begin{array}{c}
-\hat{Y} \\ 
\hat{X}%
\end{array}%
\right] \left( \Phi ^{f}\hat{N}_{d}d+\Phi ^{f}\bar{f}+\Psi ^{f}v\right) ,
\label{eq3-3} \\
\Psi ^{f}=\Phi ^{f}\Pi ^{f}\left[ 
\begin{array}{c}
M \\ 
N%
\end{array}%
\right] .  \notag
\end{gather}%
\end{theorem}

\begin{proof}
It follows from (\ref{eq3-11}) that 
\begin{equation*}
\left[ 
\begin{array}{c}
u \\ 
y%
\end{array}%
\right] =\left[ 
\begin{array}{c}
M \\ 
N%
\end{array}%
\right] v+\left[ 
\begin{array}{c}
-\hat{Y} \\ 
\hat{X}%
\end{array}%
\right] r_{y}.
\end{equation*}%
During faulty operations, 
\begin{gather*}
r_{y}=\hat{M}\left( r_{y,0}+f_{0}\right) =\hat{N}_{d}d+\Pi ^{f}\left[ 
\begin{array}{c}
{u} \\ 
y%
\end{array}%
\right] +\bar{f}\Longrightarrow \\
\left[ 
\begin{array}{c}
u \\ 
y%
\end{array}%
\right] =\left[ 
\begin{array}{c}
M \\ 
N%
\end{array}%
\right] v+\left[ 
\begin{array}{c}
-\hat{Y} \\ 
\hat{X}%
\end{array}%
\right] \left( \hat{N}_{d}d+\bar{f}+\Pi ^{f}\left[ 
\begin{array}{c}
{u} \\ 
y%
\end{array}%
\right] \right) \Longrightarrow \\
\left[ 
\begin{array}{c}
u \\ 
y%
\end{array}%
\right] =\left( I-\left[ 
\begin{array}{c}
-\hat{Y} \\ 
\hat{X}%
\end{array}%
\right] \Pi ^{f}\right) ^{-1}\left( 
\begin{array}{c}
\left[ 
\begin{array}{c}
M \\ 
N%
\end{array}%
\right] v \\ 
+\left[ 
\begin{array}{c}
-\hat{Y} \\ 
\hat{X}%
\end{array}%
\right] \left( \hat{N}_{d}d+\bar{f}\right)%
\end{array}%
\right) ,
\end{gather*}%
leading to (\ref{eq3-3}). It is clear that the closed-loop is stable if (\ref%
{eq3-2}) is true.
\end{proof}

As expected, equation (\ref{eq3-3}) showcases that the uncertain dynamics
caused by the unknown input $d$ and fault $f$ act on the system residual
subspace. Their influence on the closed-loop dynamic is decoupled from the
nominal system response 
\begin{equation*}
\left[ 
\begin{array}{c}
u_{n} \\ 
y_{n}%
\end{array}%
\right] :=\left[ 
\begin{array}{c}
M \\ 
N%
\end{array}%
\right] v.
\end{equation*}%
Consequently, it is sufficient to detect the fault using the output residual 
$r_{y}$, that is%
\begin{equation}
r_{y}=\Phi ^{f}\hat{N}_{d}d+\Phi ^{f}\bar{f}+\Psi ^{f}v.
\end{equation}%
Note that $r_{y}$ acts as the latent variable building the residual subspace 
$\mathcal{R}_{G}$ defined in Definition \ref{Def3-2}.

Next, an alternative form of the closed-loop dynamic is derived, which
underlines the influence of the faults on the system dynamic and is useful
for our subsequent study on secure CPCSs.

\begin{theorem}
\label{Theo4-2}The kernel-based loop model (\ref{eq3-1}) with the fault
modelled by (\ref{eq2-24}) can be equivalently expressed by 
\begin{gather}
\left\{ 
\begin{array}{l}
y\left( z\right) =G_{\Delta }\left( z\right) u\left( z\right) +r_{y,f}\left(
z\right)  \\ 
u\left( z\right) =K\left( z\right) y\left( z\right) +r_{u,0}\left( z\right) ,%
\end{array}%
\right.   \label{eq3-5} \\
G_{\Delta }=\left( \hat{M}+\Delta _{\hat{M}}\right) ^{-1}\left( \hat{N}%
+\Delta _{\hat{N}}\right)   \notag \\
=\left( N+\Delta _{N}\right) \left( M+\Delta _{M}\right) ^{-1},
\label{eq3-6} \\
\left[ 
\begin{array}{c}
\Delta _{M} \\ 
\Delta _{N}%
\end{array}%
\right] =\left[ 
\begin{array}{c}
-\hat{Y}\Psi ^{f} \\ 
\hat{X}\Psi ^{f}%
\end{array}%
\right] ,\left[ 
\begin{array}{cc}
\Delta _{\hat{N}} & \text{ }-\Delta _{\hat{M}}%
\end{array}%
\right] =\Pi ^{f},  \label{eq3-6a} \\
r_{y,f}=\left( \hat{M}+\Delta _{\hat{M}}\right) ^{-1}\left( \hat{N}_{d}d+%
\bar{f}\right) .  \notag
\end{gather}
\end{theorem}

\begin{proof}
It follows from (\ref{eq3-1}) that 
\begin{gather*}
r_{y}=\left[ 
\begin{array}{cc}
-\hat{N} & \text{ }\hat{M}%
\end{array}%
\right] \left[ 
\begin{array}{c}
u \\ 
y%
\end{array}%
\right] =\hat{N}_{d}d+\Pi ^{f}\left[ 
\begin{array}{c}
{u} \\ 
y%
\end{array}%
\right] +\bar{f} \\
\Longrightarrow \left( \left[ 
\begin{array}{cc}
-\hat{N} & \text{ }\hat{M}%
\end{array}%
\right] -\Pi ^{f}\right) \left[ 
\begin{array}{c}
u \\ 
y%
\end{array}%
\right] =\hat{N}_{d}d+\bar{f}, \\
Xu+Yy=r_{u}=v,
\end{gather*}%
from which (\ref{eq3-5}) with 
\begin{equation*}
y=G_{\Delta }u+r_{y,f}\left( z\right) ,G_{\Delta }=\left( \hat{M}+\Delta _{%
\hat{M}}\right) ^{-1}\left( \hat{N}+\Delta _{\hat{N}}\right)
\end{equation*}%
is proven. We now examine (\ref{eq3-6}), which is equivalent to 
\begin{equation*}
\left[ 
\begin{array}{cc}
-\hat{N}-\Delta _{\hat{N}} & \text{ }\hat{M}+\Delta _{\hat{M}}%
\end{array}%
\right] \left( \left[ 
\begin{array}{c}
M \\ 
N%
\end{array}%
\right] +\left[ 
\begin{array}{c}
\Delta _{M} \\ 
\Delta _{N}%
\end{array}%
\right] \right) =0.
\end{equation*}%
Notice that 
\begin{gather*}
\left[ 
\begin{array}{cc}
-\Delta _{\hat{N}} & \text{ }\Delta _{\hat{M}}%
\end{array}%
\right] \left[ 
\begin{array}{c}
M \\ 
N%
\end{array}%
\right] =-\Pi ^{f}\left[ 
\begin{array}{c}
M \\ 
N%
\end{array}%
\right] , \\
\left[ 
\begin{array}{cc}
-\Delta _{\hat{N}} & \text{ }\Delta _{\hat{M}}%
\end{array}%
\right] \left[ 
\begin{array}{c}
\Delta _{M} \\ 
\Delta _{N}%
\end{array}%
\right] =-\Pi ^{f}\left[ 
\begin{array}{c}
-\hat{Y} \\ 
\hat{X}%
\end{array}%
\right] \Psi ^{f}, \\
\left[ 
\begin{array}{cc}
-\hat{N} & \text{ }\hat{M}%
\end{array}%
\right] \left[ 
\begin{array}{c}
\Delta _{M} \\ 
\Delta _{N}%
\end{array}%
\right] =\Psi ^{f}\Longrightarrow \\
\left[ 
\begin{array}{cc}
-\hat{N}-\Delta _{\hat{N}} & \text{ }\hat{M}+\Delta _{\hat{M}}%
\end{array}%
\right] \left( \left[ 
\begin{array}{c}
M \\ 
N%
\end{array}%
\right] +\left[ 
\begin{array}{c}
\Delta _{M} \\ 
\Delta _{N}%
\end{array}%
\right] \right) \\
=\left( I-\Pi ^{f}\left[ 
\begin{array}{c}
-\hat{Y} \\ 
\hat{X}%
\end{array}%
\right] \right) \Psi ^{f}-\Pi ^{f}\left[ 
\begin{array}{c}
M \\ 
N%
\end{array}%
\right] =0.
\end{gather*}%
Hence, (\ref{eq3-6}) is proven.
\end{proof}

Theorem \ref{Theo4-2} implies that the variation caused by the fault (\ref%
{eq2-24}) can be equivalently modelled by the coprime factor uncertainty of
the plant \cite{Vinnicombe-book}. In robust control theory, $\left( \Delta _{%
\hat{M}},\Delta _{\hat{N}}\right) $ and $\left( \Delta _{M},\Delta
_{N}\right) $ are called coprime factor uncertainty of the plant \cite%
{Vinnicombe-book}, i.e. 
\begin{align*}
I_{G}& =\left[ 
\begin{array}{c}
M \\ 
N%
\end{array}%
\right] +\Delta I_{G},\Delta I_{G}=\left[ 
\begin{array}{c}
\Delta _{M} \\ 
\Delta _{N}%
\end{array}%
\right] , \\
K_{G}& =\left[ 
\begin{array}{cc}
-\hat{N} & \text{ }\hat{M}%
\end{array}%
\right] +\Delta K_{G},\Delta K_{G}=\left[ 
\begin{array}{cc}
-\Delta _{\hat{N}} & \text{ }\Delta _{\hat{M}}%
\end{array}%
\right] .
\end{align*}

\subsection{Analysis of system dynamics under cyber-attacks}

In this section, we consider the CPCSs (\ref{eq2-20}) with $K$ subject to (%
\ref{eq2-18}) and attack model (\ref{eq2-21}). It is supposed that the
system is operating under fault-free conditions.

\subsubsection{The information patterns and essential models\label%
{subsectionIV-B-A}}

The available data on the plant and control station sides are different.
While the accessible input-output data on the plant side are $\left(
u^{a},y\right) ,$ the control station has the input-output data $\left(
u,y^{a}\right) $ available. This information mismatching results in attack
information potential 
\begin{equation}
IP_{a}:=\left[ 
\begin{array}{c}
u^{a} \\ 
y%
\end{array}%
\right] -\left[ 
\begin{array}{c}
u \\ 
y^{a}%
\end{array}%
\right] =\left[ 
\begin{array}{c}
a_{u} \\ 
-a_{y}%
\end{array}%
\right]  \label{eq4-3a}
\end{equation}%
that indicates the existence of cyber-attacks. In order to avoid
nomenclatural inconsistency due to\ the information mismatching, we define
the residual $\left( r_{u},r_{y}\right) $ in the sequel by means of the
kernel-based model according to the information pattern on the plant side, 
\begin{equation}
\left\{ 
\begin{array}{l}
\hat{M}y-\hat{N}u^{a}=r_{y} \\ 
Xu^{a}+Yy=r_{u}%
\end{array}%
\right. ,r_{y}=\hat{M}r_{y,0}=\hat{N}_{d}d,r_{u}=v.  \label{eq4-2}
\end{equation}%
The loop dynamic expressed in terms of $\left( u^{a},y\right) \ $is
governed, considering the information pattern on the control station and
information potential (\ref{eq4-3a}), by 
\begin{gather}
\left[ 
\begin{array}{cc}
X & \text{ }Y \\ 
-\hat{N} & \text{ }\hat{M}%
\end{array}%
\right] \left[ 
\begin{array}{c}
u^{a} \\ 
y%
\end{array}%
\right] =\left[ 
\begin{array}{c}
Xa_{u}-Ya_{y}+v \\ 
r_{y}%
\end{array}%
\right] \Longrightarrow  \notag \\
\left[ 
\begin{array}{c}
u^{a} \\ 
y%
\end{array}%
\right] =\left[ 
\begin{array}{c}
M \\ 
N%
\end{array}%
\right] \left( Xa_{u}-Ya_{y}+v\right) +\left[ 
\begin{array}{c}
-\hat{Y} \\ 
\hat{X}%
\end{array}%
\right] r_{y}.  \label{eq4-3}
\end{gather}%
Attributed to the Bezout identity (\ref{eq2-6}), the attacks can be
expressed by%
\begin{gather}
\left[ 
\begin{array}{c}
a_{u} \\ 
-a_{y}%
\end{array}%
\right] =\left[ 
\begin{array}{c}
M \\ 
N%
\end{array}%
\right] a_{r_{u}}+\left[ 
\begin{array}{c}
-\hat{Y} \\ 
\hat{X}%
\end{array}%
\right] a_{r_{y}},  \label{eq4-4} \\
a_{r_{u}}=\left[ 
\begin{array}{cc}
X & \text{ }Y%
\end{array}%
\right] \left[ 
\begin{array}{c}
a_{u} \\ 
-a_{y}%
\end{array}%
\right] ,a_{r_{y}}=\left[ 
\begin{array}{cc}
-\hat{N} & \text{ }\hat{M}%
\end{array}%
\right] \left[ 
\begin{array}{c}
a_{u} \\ 
-a_{y}%
\end{array}%
\right] ,  \label{eq4-4a}
\end{gather}%
which models the influence of $\left( a_{u},a_{y}\right) $ on the image and
residual subspaces $\mathcal{I}_{G}$ and $\mathcal{R}_{G},$ respectively.
Relations (\ref{eq4-3a})-(\ref{eq4-4}) are essential for our subsequent
work. Note that, attribute to (\ref{eq4-4}), the attack information
potential can be expressed by 
\begin{equation}
IP_{a}=\left[ 
\begin{array}{c}
M \\ 
N%
\end{array}%
\right] a_{r_{u}}+\left[ 
\begin{array}{c}
-\hat{Y} \\ 
\hat{X}%
\end{array}%
\right] a_{r_{y}}.  \label{eq4-3b}
\end{equation}

\subsubsection{System dynamics under cyber-attacks\label{subsectionIV-B-B}}

We now examine the loop dynamic under cyber-attacks modelled by (\ref{eq2-21}%
).

\begin{theorem}
\label{Theo4-3}Given the loop model (\ref{eq4-2}) and the attack model (\ref%
{eq2-21}), the plant dynamics under additive and multiplicative attacks, 
\begin{equation}
\left[ 
\begin{array}{c}
a_{u} \\ 
a_{y}%
\end{array}%
\right] =\left[ 
\begin{array}{c}
E_{u}\eta _{u} \\ 
E_{y}\eta _{y}%
\end{array}%
\right] ,\left[ 
\begin{array}{c}
a_{u} \\ 
a_{y}%
\end{array}%
\right] =\Pi ^{a}\left[ 
\begin{array}{c}
{u} \\ 
y%
\end{array}%
\right] ,  \label{eq4-7}
\end{equation}%
are respectively governed by 
\begin{gather}
\left[ 
\begin{array}{c}
u^{a} \\ 
y%
\end{array}%
\right] =\left[ 
\begin{array}{c}
M \\ 
N%
\end{array}%
\right] \left( v+\bar{a}_{r_{u}}\right) +\left[ 
\begin{array}{c}
-\hat{Y} \\ 
\hat{X}%
\end{array}%
\right] r_{y},  \label{eq4-8} \\
\bar{a}_{r_{u}}=\left[ 
\begin{array}{cc}
X & \text{ }Y%
\end{array}%
\right] \left[ 
\begin{array}{c}
E_{u}\eta _{u}(z) \\ 
-E_{y}\eta _{y}(z)%
\end{array}%
\right] ,  \notag \\
\left[ 
\begin{array}{c}
u^{a} \\ 
y%
\end{array}%
\right] =\left[ 
\begin{array}{c}
M \\ 
N%
\end{array}%
\right] \left( v_{\Delta }+\Psi ^{a}r_{y}\right) +\left[ 
\begin{array}{c}
-\hat{Y} \\ 
\hat{X}%
\end{array}%
\right] r_{y},  \label{eq4-9} \\
v_{\Delta }=\left( I-\left[ 
\begin{array}{cc}
X & \text{ }Y%
\end{array}%
\right] \Phi ^{a}\left[ 
\begin{array}{c}
M \\ 
N%
\end{array}%
\right] \right) ^{-1}v,  \notag \\
\Psi ^{a}=\left( I-\left[ 
\begin{array}{cc}
X & \text{ }Y%
\end{array}%
\right] \Phi ^{a}\left[ 
\begin{array}{c}
M \\ 
N%
\end{array}%
\right] \right) ^{-1}\left[ 
\begin{array}{cc}
X & \text{ }Y%
\end{array}%
\right] \Phi ^{a}\left[ 
\begin{array}{c}
-\hat{Y} \\ 
\hat{X}%
\end{array}%
\right] ,  \notag \\
\Phi ^{a}=\left[ 
\begin{array}{cc}
I & \text{ }0 \\ 
0 & \text{ }-I%
\end{array}%
\right] \Pi ^{a}\left( I+\left[ 
\begin{array}{cc}
I & \text{ }0 \\ 
0 & \text{ }0%
\end{array}%
\right] \Pi ^{a}\right) ^{-1}.  \notag
\end{gather}
\end{theorem}

\begin{proof}
The dynamic under additive attacks (\ref{eq4-8}) follows immediately from (%
\ref{eq4-3}). We prove (\ref{eq4-9}), the system response to the
multiplicative attacks. Observe that 
\begin{equation*}
\left[ 
\begin{array}{c}
a_{u} \\ 
-a_{y}%
\end{array}%
\right] =\left[ 
\begin{array}{cc}
I & \text{ }0 \\ 
0 & \text{ }-I%
\end{array}%
\right] \Pi ^{a}\left[ 
\begin{array}{c}
u^{a}-a_{u} \\ 
y%
\end{array}%
\right] .
\end{equation*}%
By some routine computations, it turns out 
\begin{equation}
\left[ 
\begin{array}{c}
a_{u} \\ 
-a_{y}%
\end{array}%
\right] =\Phi ^{a}\left[ 
\begin{array}{c}
u^{a} \\ 
y%
\end{array}%
\right] .  \label{eq4-15}
\end{equation}%
Substituting it into (\ref{eq4-3}) yields 
\begin{equation*}
\left[ 
\begin{array}{c}
u^{a} \\ 
y%
\end{array}%
\right] =\left[ 
\begin{array}{c}
M \\ 
N%
\end{array}%
\right] \left( \left[ 
\begin{array}{cc}
X & \text{ }Y%
\end{array}%
\right] \Phi ^{a}\left[ 
\begin{array}{c}
u^{a} \\ 
y%
\end{array}%
\right] +v\right) +\left[ 
\begin{array}{c}
-\hat{Y} \\ 
\hat{X}%
\end{array}%
\right] r_{y},
\end{equation*}%
from which (\ref{eq4-9}) follows.
\end{proof}

The system responses (\ref{eq4-8}) and (\ref{eq4-9}) demonstrate that the
cyber-attacks act on the system image subspace, and their influence is
decoupled from the residual behavior 
\begin{equation*}
\left[ 
\begin{array}{c}
u_{r}^{a} \\ 
y_{r}%
\end{array}%
\right] =\left[ 
\begin{array}{c}
-\hat{Y} \\ 
\hat{X}%
\end{array}%
\right] r_{y}.
\end{equation*}%
On this account, it is impossible to detect cyber-attacks on the plant side
by means of the output residual generator. Instead, the input residual 
\begin{equation*}
r_{u}=Xu^{a}+Yy
\end{equation*}%
is a capable detector of the cyber-attacks.

Considering that attack detection is generally performed on the control
station, we now examine the system dynamic in the information pattern $%
\left( u,y^{a}\right) ,$ i.e. from the viewpoint of the control station. It
follows from the information potential (\ref{eq4-3a}) and (\ref{eq4-3})-(\ref%
{eq4-4}) that%
\begin{equation}
\left[ 
\begin{array}{c}
u \\ 
y^{a}%
\end{array}%
\right] =\left[ 
\begin{array}{c}
M \\ 
N%
\end{array}%
\right] v+\left[ 
\begin{array}{c}
-\hat{Y} \\ 
\hat{X}%
\end{array}%
\right] \left( r_{y}-a_{r_{y}}\right) .  \label{eq4-5}
\end{equation}%
According to (\ref{eq4-15}), 
\begin{gather}
\left[ 
\begin{array}{c}
a_{u} \\ 
-a_{y}%
\end{array}%
\right] =\Phi ^{a}\left( \left[ 
\begin{array}{c}
u \\ 
y^{a}%
\end{array}%
\right] +\left[ 
\begin{array}{c}
a_{u} \\ 
-a_{y}%
\end{array}%
\right] \right) +\left[ 
\begin{array}{c}
\bar{\eta}_{u} \\ 
\bar{\eta}_{y}%
\end{array}%
\right]  \label{eq4-15a} \\
\Longrightarrow \left[ 
\begin{array}{c}
a_{u} \\ 
-a_{y}%
\end{array}%
\right] =\left( I-\Phi ^{a}\right) ^{-1}\left( \Phi ^{a}\left[ 
\begin{array}{c}
u \\ 
y^{a}%
\end{array}%
\right] +\left[ 
\begin{array}{c}
\bar{\eta}_{u} \\ 
\bar{\eta}_{y}%
\end{array}%
\right] \right) ,  \notag \\
\left[ 
\begin{array}{c}
\bar{\eta}_{u} \\ 
\bar{\eta}_{y}%
\end{array}%
\right] =\left[ 
\begin{array}{cc}
I & \text{ }0 \\ 
0 & \text{ }-I%
\end{array}%
\right] \left( I+\Pi ^{a}\left[ 
\begin{array}{cc}
I & \text{ }0 \\ 
0 & \text{ }0%
\end{array}%
\right] \right) ^{-1}\left[ 
\begin{array}{c}
E_{u}\eta _{u} \\ 
E_{y}\eta _{y}%
\end{array}%
\right] .  \notag
\end{gather}%
Substituting it into $a_{r_{y}}$ in (\ref{eq4-5}) yields 
\begin{gather}
\left[ 
\begin{array}{c}
u \\ 
y^{a}%
\end{array}%
\right] =\Gamma ^{-1}\left( \left[ 
\begin{array}{c}
M \\ 
N%
\end{array}%
\right] v+\left[ 
\begin{array}{c}
-\hat{Y} \\ 
\hat{X}%
\end{array}%
\right] \left( r_{y}-\Theta \left[ 
\begin{array}{c}
\bar{\eta}_{u} \\ 
\bar{\eta}_{y}%
\end{array}%
\right] \right) \right) ,  \label{eq4-17} \\
\Gamma =I+\left[ 
\begin{array}{c}
-\hat{Y} \\ 
\hat{X}%
\end{array}%
\right] \left[ 
\begin{array}{cc}
-\hat{N} & \text{ }\hat{M}%
\end{array}%
\right] \left( I-\Phi ^{a}\right) ^{-1}\Phi ^{a},  \label{eq4-19} \\
\Theta =\left[ 
\begin{array}{cc}
-\hat{N} & \text{ }\hat{M}%
\end{array}%
\right] \left( I-\Phi ^{a}\right) ^{-1},  \notag
\end{gather}%
from which we have the following theorem.

\begin{theorem}
\label{Theo4-4}Given the loop model (\ref{eq4-2}) and the attack model (\ref%
{eq2-21}), then it holds%
\begin{gather}
\left[ 
\begin{array}{c}
u \\ 
y^{a}%
\end{array}%
\right] =\left[ 
\begin{array}{c}
M \\ 
N%
\end{array}%
\right] v+\left[ 
\begin{array}{c}
-\hat{Y} \\ 
\hat{X}%
\end{array}%
\right] \left( v_{\Delta }^{a}+\Gamma ^{a}r_{y}+\bar{a}_{\Delta }\right) ,
\label{eq4-17b} \\
v_{\Delta }^{a}=-\Gamma ^{-1}\left[ 
\begin{array}{c}
-\hat{Y} \\ 
\hat{X}%
\end{array}%
\right] \left[ 
\begin{array}{cc}
-\hat{N} & \text{ }\hat{M}%
\end{array}%
\right] \left( I-\Phi ^{a}\right) ^{-1}\Phi ^{a}\left[ 
\begin{array}{c}
M \\ 
N%
\end{array}%
\right] v,  \notag \\
\Gamma ^{a}=\left( I+\left[ 
\begin{array}{cc}
-\hat{N} & \text{ }\hat{M}%
\end{array}%
\right] \left( I-\Phi ^{a}\right) ^{-1}\Phi ^{a}\left[ 
\begin{array}{c}
-\hat{Y} \\ 
\hat{X}%
\end{array}%
\right] \right) ^{-1},  \notag \\
\bar{a}_{\Delta }=-\Gamma ^{a}\Theta \left[ 
\begin{array}{c}
\bar{\eta}_{u} \\ 
\bar{\eta}_{y}%
\end{array}%
\right] .  \notag
\end{gather}
\end{theorem}

\begin{proof}
The equation (\ref{eq4-17b}) follows from (\ref{eq4-17}) by some routine
computations.
\end{proof}

The result (\ref{eq4-17b}) showcases that, due to the information potential (%
\ref{eq4-3a}), the cyber-attacks, faults and unknown input act exclusively
on the residual subspace of $\left( u,y^{a}\right) ,$ which is the
information pattern on the control station side, and are decoupled from the
image subspace.

\subsubsection{Duality between system dynamics under faults and attacks\label%
{subsectionIV-B-C}}

Theorems \ref{Theo4-2}-\ref{Theo4-3} build the basis for our endeavours to
leverage the unified framework to address simultaneous detection,
fault-tolerant and resilient control of faults and attacks. Before that, we
firstly reveal the duality between system dynamics under faults and attacks,
which is helpful for the subsequent work.

\begin{theorem}
\label{Theo4-5} The closed-loop model (\ref{eq2-20}) under cyber-attackers
modelled by (\ref{eq2-21}) is equivalent to 
\begin{gather}
\left\{ 
\begin{array}{l}
y=Gu^{a}+r_{y,0} \\ 
u^{a}=K_{\Delta }y+r_{u,a},%
\end{array}%
\right.   \label{eq4-10} \\
r_{u,a}=\left( X+\Delta _{X}\right) ^{-1}\left( v+a_{\eta }\right) ,a_{\eta
}=\left[ 
\begin{array}{cc}
X & \text{ }Y%
\end{array}%
\right] \left[ 
\begin{array}{c}
\bar{\eta}_{u} \\ 
-\bar{\eta}_{y}%
\end{array}%
\right] ,  \label{eq4-13} \\
K_{\Delta }=-\left( X+\Delta _{X}\right) ^{-1}\left( Y+\Delta _{Y}\right) ,
\label{eq4-11} \\
\left[ 
\begin{array}{cc}
\Delta _{X} & \text{ }\Delta _{Y}%
\end{array}%
\right] =-\left[ 
\begin{array}{cc}
X & \text{ }Y%
\end{array}%
\right] \Phi ^{a}.  \label{eq4-12}
\end{gather}
\end{theorem}

\begin{proof}
On account of (\ref{eq4-15a}), the attack model (\ref{eq2-21}) is recast
into 
\begin{equation}
\left[ 
\begin{array}{c}
a_{u} \\ 
-a_{y}%
\end{array}%
\right] =\Phi ^{a}\left[ 
\begin{array}{c}
u^{a} \\ 
y%
\end{array}%
\right] +\left[ 
\begin{array}{c}
\bar{\eta}_{u} \\ 
\bar{\eta}_{y}%
\end{array}%
\right] .  \label{eq4-16}
\end{equation}%
Similar to the proof of Theorem \ref{Theo4-3}, consider 
\begin{gather*}
\left[ 
\begin{array}{c}
u^{a} \\ 
y%
\end{array}%
\right] =\left[ 
\begin{array}{c}
M \\ 
N%
\end{array}%
\right] \left( \left[ 
\begin{array}{cc}
X & \text{ }Y%
\end{array}%
\right] \Phi ^{a}\left[ 
\begin{array}{c}
u^{a} \\ 
y%
\end{array}%
\right] +a_{\eta }+v\right) \\
+\left[ 
\begin{array}{c}
-\hat{Y} \\ 
\hat{X}%
\end{array}%
\right] r_{y},a_{\eta }=\left[ 
\begin{array}{cc}
X & \text{ }Y%
\end{array}%
\right] \left[ 
\begin{array}{c}
\bar{\eta}_{u} \\ 
-\bar{\eta}_{y}%
\end{array}%
\right] ,
\end{gather*}%
and write it into 
\begin{equation*}
\left( \left[ 
\begin{array}{cc}
X & \text{ }Y%
\end{array}%
\right] -\left[ 
\begin{array}{cc}
X & \text{ }Y%
\end{array}%
\right] \Phi ^{a}\right) \left[ 
\begin{array}{c}
u^{a} \\ 
y%
\end{array}%
\right] =a_{\eta }+v.
\end{equation*}%
Defining $r_{u,a},\left[ 
\begin{array}{cc}
\Delta _{X} & \text{ }\Delta _{Y}%
\end{array}%
\right] $ according to (\ref{eq4-13})-(\ref{eq4-12}) yields%
\begin{equation*}
u^{a}=-\left( X+\Delta _{X}\right) ^{-1}\left( Y+\Delta _{Y}\right)
y+r_{u,a}.
\end{equation*}%
The theorem is proven.
\end{proof}

Comparing with the result in Theorem \ref{Theo4-2} makes it painfully clear
that the faults and cyber-attacks impair the closed-loop dynamic in a dual
manner. Fig. 2 and Fig. 3 sketch the duality of the equivalent system
configurations under additive and multiplicative faults and cyber-attacks,
respectively. This important property insightfully reveals the structural
relations between the faults and cyber-attacks. It inspires us to design
detection and control systems towards simultaneous detection, tolerant and
resilient control of the faults and cyber-attacks in a dual manner, which
builds the control-theoretic foundation for establishing a novel framework
of designing fault-tolerant and attack-resilient control systems in a
unified way.

\begin{figure}[t]
\centering\includegraphics[width=5.5cm]{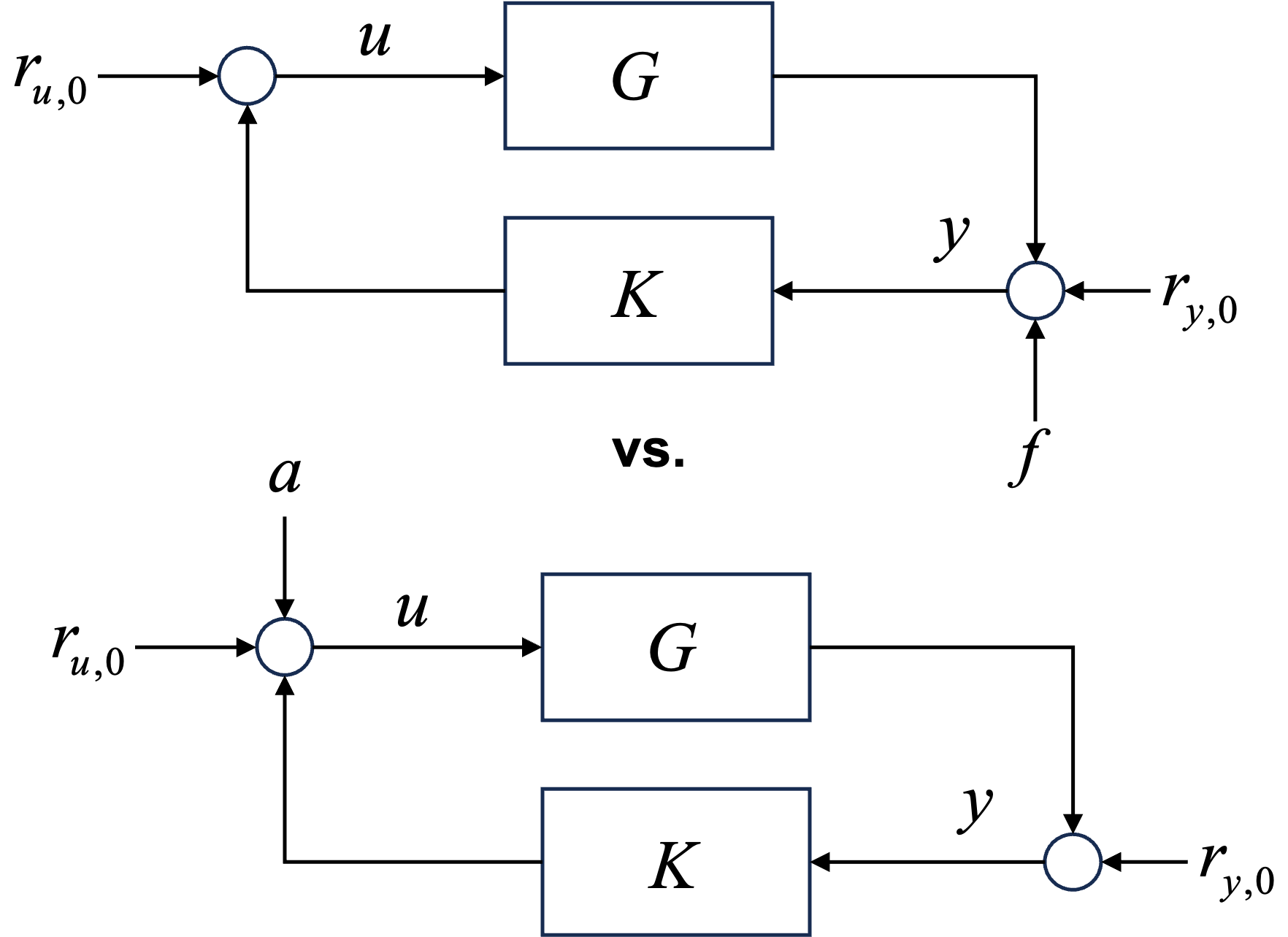}
\caption{The duality of the equivalent system configurations under the
additive faults and cyber-attacks}
\label{Fig2}
\end{figure}

\begin{figure}[t]
\centering\includegraphics[width=5.5cm]{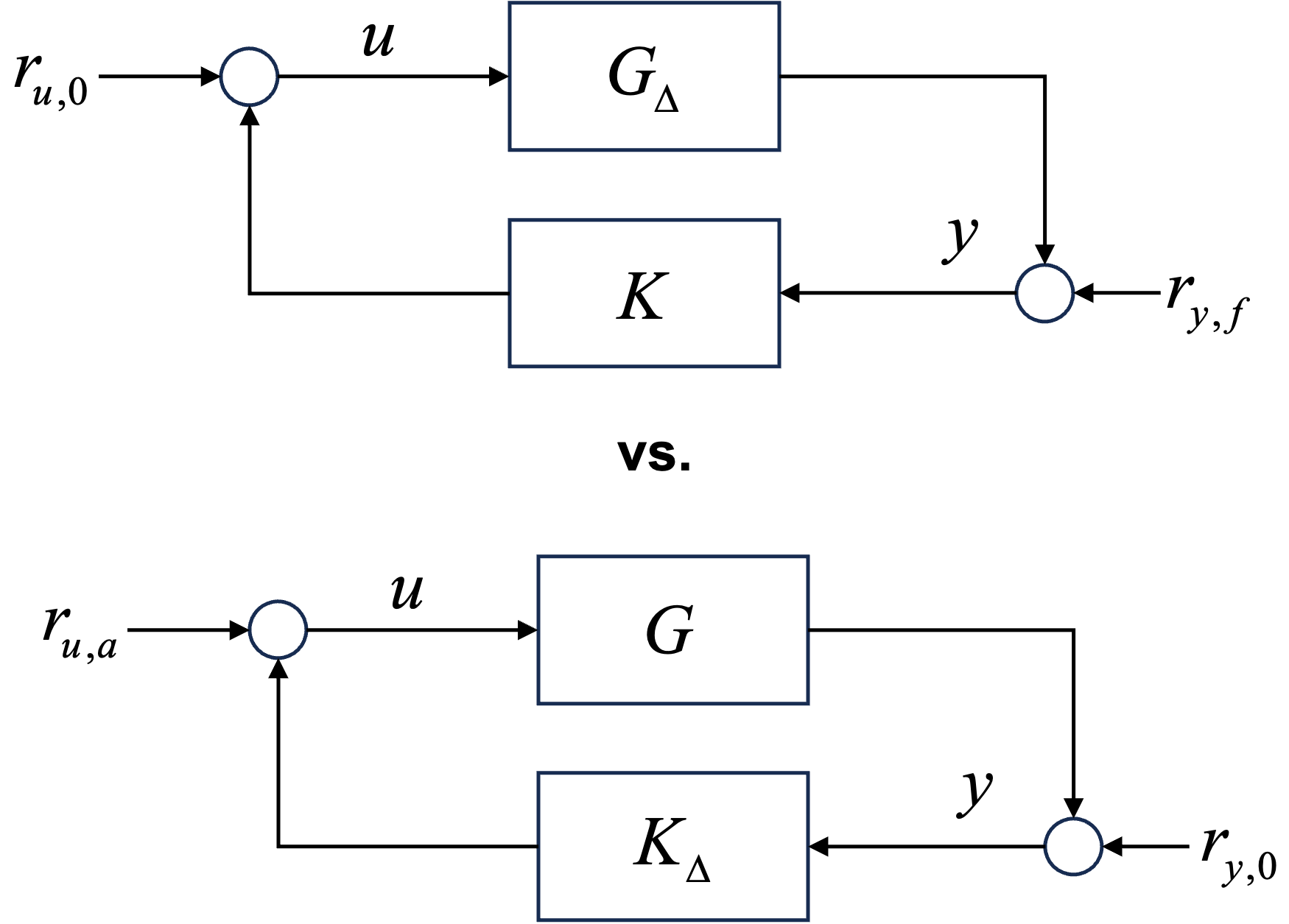}
\caption{The duality of the equivalent system configurations under the
multiplicative faults and cyber-attacks}
\label{Fig3}
\end{figure}

\subsection{Detectability and actuability of attacks, and vulnerability
analysis\label{Subsection IV-C}}

Issues associated with the concept stealthy attacks have continually
attracted remarkable research attention since the publication of several
real cases of cyber-attacks on industrial infrastructures \cite%
{DIBAJI2019-survey,DeruiDing2021-survey,TGXHV2020,Zhou2021IEEE-Proc}. These
are, among others, detection and design of stealthy attacks, vulnerability
of CPCSs. In this subsection, the analysis results presented in the previous
sections are leveraged to reveal new aspects of stealthy cyber-attacks,
viewed in the unified framework of control and detection. The focus is on
alternative definitions, existence conditions of stealthy attacks and
vulnerability analysis.

\subsubsection{Definitions and existence conditions of detectability and
actuability of attacks\label{subsectionIV-C-A}}

In the literature, stealthy attacks, for instance, replay \cite%
{Mo2015-Watermarked-detection}, zero dynamics \cite%
{TEIXEIRA-zero-attack_2015} and covert \cite{Smith2015} attacks, are usually
defined with respect to a certain type of detector \cite%
{DIBAJI2019-survey,DeruiDing2021-survey,TGXHV2020,Zhou2021IEEE-Proc,Survey-Segovia-Ferreira2024}%
. As detectors, observers or Kalman-filters serve as residual generators. It
is assumed, often implicitly, that the detection is performed in the
information pattern $\left( u,y^{a}\right) ,$ i.e. on the control station
side. In our recent work, the concept of kernel attacks was introduced \cite%
{DLautomatica2022}, which can be equivalently expressed as the so-called
image attacks \cite{LiAutomatica_2026}, thanks to the equivalence between
the kernel and image subspaces, $\mathcal{I}_{G}=\mathcal{K}_{G},$ as
described in Subsection \ref{Subsec3-1}. The concept of kernel/image attacks
not only gives a uniform form of stealthy attacks like covert, zero dynamics
and replay attacks, but also, and more importantly, characterizes stealthy
attacks independent of the used detection schemes. Instead, kernel/image
attacks are defined in terms of plant model knowledge, $G$ or $\left(
A,B,C,D\right) ,$ and on the assumption of the information pattern $\left(
u,y^{a}\right) .$ Below, the concept of detectability of attacks is
introduced. It is not only an alternative to stealthy attacks, but also more
general and independent of the type of detectors. In other words, any
stealthy attack should be an undetectable attack, but a stealthy attack
under the use of an observer-based detector may be detected using other
types of detectors.

We would like to call the reader's attention to (\ref{eq4-3}) and Theorem %
\ref{Theo4-3}. It can be observed that a certain class of attacks is
effectless. Although they act on the system, they do not cause any change in
the system dynamic. In the literature, this fact is obviously overlooked.
This motivates us to introduce the concept of actuability of attacks, a dual
form to the undetectability of attacks. For our purpose, we consider the
system model (\ref{eq2-1}) under cyber-attacks with the attack model (\ref%
{eq2-21}), 
\begin{equation}
\left\{ 
\begin{array}{l}
y=Gu^{a}+r_{y,0},G=\left( A,B,C,D\right) \\ 
u=Ky^{a}+r_{u,0},K=-\hat{Y}\hat{X}^{-1}=X^{-1}Y.%
\end{array}%
\right.  \label{eq4-18}
\end{equation}

\begin{definition}
\label{De4-1}Given the control system (\ref{eq4-18}), the attack pair $%
\left( a_{u},a_{y}\right) $ is called undetectable, if 
\begin{equation}
\left[ 
\begin{array}{c}
u^{I} \\ 
y^{I}%
\end{array}%
\right] \equiv \left[ 
\begin{array}{c}
u_{0} \\ 
y_{0}%
\end{array}%
\right] ,  \label{eq4-23a}
\end{equation}%
where $\left( u_{0},y_{0}\right) $ represent the system dynamic under
attack-free operation conditions, and $\left( u^{I},y^{I}\right) $ denote
the accessible input-output data $\left( u,y\right) $ with respect to the
information pattern, 
\begin{equation*}
\left( u^{I},y^{I}\right) =\left\{ 
\begin{array}{l}
\left( u,y^{a}\right) ,\text{ detection on the control station side} \\ 
\left( u^{a},y\right) ,\text{ detection on the plant side.}%
\end{array}%
\right. 
\end{equation*}
\end{definition}

Equation (\ref{eq4-23a}) means, the attack does not cause any changes
comparing with the attack-free dynamic, 
\begin{equation*}
\left[ 
\begin{array}{c}
u_{0} \\ 
y_{0}%
\end{array}%
\right] =\left[ 
\begin{array}{c}
M \\ 
N%
\end{array}%
\right] v+\left[ 
\begin{array}{c}
-\hat{Y} \\ 
\hat{X}%
\end{array}%
\right] r_{y}.
\end{equation*}

\begin{definition}
\label{De4-2}Given the control system (\ref{eq4-18}), the attack pair $%
\left( a_{u},a_{y}\right) $ is called actuable, if 
\begin{equation}
\left[ 
\begin{array}{c}
u^{a} \\ 
y%
\end{array}%
\right] \not\equiv \left[ 
\begin{array}{c}
u_{0} \\ 
y_{0}%
\end{array}%
\right] .  \label{eq4-23b}
\end{equation}
\end{definition}

It is apparent that the undetectability and actuability given in the above
definitions are exclusively expressed in terms of the system behavior. Now,
we are in a good position to address the existence conditions of
undetectable attacks and actuable attacks. For the former, it is assumed
that attackers possess model knowledge $\left( A,B,C,D\right) $ of the plant.

\begin{theorem}
\label{Theo4-6}Given the control system (\ref{eq4-2}) and the attack model (%
\ref{eq2-21}), the attacks $\left( a_{u},a_{y}\right) $ are undetectable
under the information pattern $\left( u,y^{a}\right) $ if and only if%
\begin{equation}
\left[ 
\begin{array}{c}
a_{u} \\ 
-a_{y}%
\end{array}%
\right] =\left[ 
\begin{array}{c}
M \\ 
N%
\end{array}%
\right] \xi ,  \label{eq4-24a}
\end{equation}%
where $\xi $ is an arbitrary $\mathcal{L}_{2}$ bounded signal. 
\end{theorem}

\begin{proof}
Since $\forall \xi $%
\begin{equation*}
\left[ 
\begin{array}{cc}
-\hat{N} & \text{ }\hat{M}%
\end{array}%
\right] \left[ 
\begin{array}{c}
M \\ 
N%
\end{array}%
\right] \xi =0,
\end{equation*}%
according (\ref{eq4-5}), it is apparent that (\ref{eq4-24a}) is a sufficient
condition. To prove the necessity, we examine the existence condition for 
\begin{equation}
\left[ 
\begin{array}{cc}
-\hat{N} & \text{ }\hat{M}%
\end{array}%
\right] \left[ 
\begin{array}{c}
a_{u} \\ 
-a_{y}%
\end{array}%
\right] =0  \label{eq4-25}
\end{equation}%
following Definition \ref{De4-1} and (\ref{eq4-5}). Let $\left(
a_{u},a_{y}\right) $ be $\mathcal{L}_{2}$ bounded, which can be written as 
\begin{equation}
\left[ 
\begin{array}{c}
a_{u} \\ 
-a_{y}%
\end{array}%
\right] =\left[ 
\begin{array}{cc}
M & \text{ }-\hat{Y} \\ 
N & \text{ }\hat{X}%
\end{array}%
\right] \left[ 
\begin{array}{c}
\xi _{1} \\ 
\xi _{2}%
\end{array}%
\right]
\end{equation}%
for some $\mathcal{L}_{2}$ bounded $\left( \xi _{1},\xi _{2}\right) .$
Consequently, 
\begin{equation*}
\left[ 
\begin{array}{cc}
-\hat{N} & \text{ }\hat{M}%
\end{array}%
\right] \left[ 
\begin{array}{c}
a_{u} \\ 
-a_{y}%
\end{array}%
\right] =\xi _{2},
\end{equation*}%
from which it can be immediately concluded that (\ref{eq4-25}) is true only
if $\xi _{2}=0,$ which implies 
\begin{equation*}
\left[ 
\begin{array}{c}
a_{u} \\ 
-a_{y}%
\end{array}%
\right] =\left[ 
\begin{array}{c}
M \\ 
N%
\end{array}%
\right] \xi _{1}=:\left[ 
\begin{array}{c}
M \\ 
N%
\end{array}%
\right] \xi .
\end{equation*}
\end{proof}

\begin{remark}
Since the stability of the closed-loop is assured, the signal $\xi $ can be
a function of $\left( u,y\right) $ as well. That is, the result (\ref{eq4-24a}%
) holds for both additive and multiplication attacks. 
\end{remark}

Note that (\ref{eq4-24a}) implies that the attack pair $\left(
a_{u},a_{y}\right) $ is undetectable, when 
\begin{equation}
\left[ 
\begin{array}{c}
a_{u} \\ 
-a_{y}%
\end{array}%
\right] \in \mathcal{I}_{G}.
\end{equation}%
Hence, we call undetectable attacks image attacks as well to emphasize the
fact that they lie in the system image subspace $\mathcal{I}_{G}$ \cite%
{DLautomatica2022,LiAutomatica_2026}. Recall that the system image subspace
is invariant with respect to the gain matrix $F.$ Thus, it is sufficient for
attackers to construct the image subspace using model knowledge $\left(
A,B,C,D\right) .$ In fact, \cite{DLautomatica2022} firstly proposed the
concept kernel attacks for the undetectable (stealthy) attacks, as it was
proven that $\left( a_{u},a_{y}\right) $ should belong to the kernel
subspace $\mathcal{K}_{G},$ i.e. satisfy (\ref{eq4-25}), and then introduced
the concept of image attacks using the equivalence $\mathcal{K}_{G}=\mathcal{%
I}_{G}.$ It was proven that the covert, zero dynamics and replay attacks are
kernel attackers \cite{DLautomatica2022}. It is noteworthy that the
aforementioned results are true for any class of cyber-attacks, including
those that attack a part of the input and output channels. Specifically, let
the input and output channels $\left( B,C,D\right) $ be split into two
groups as described by (\ref{eq3-23}) with the attacked channels $\left(
B_{1},D_{21},C_{2}\right) $ and unattacked channels $\left(
B_{2},D_{22},\right) $ and $\left( C_{1},\left[ 
\begin{array}{cc}
D_{11} & \text{ }D_{12}%
\end{array}%
\right] \right) .$ Correspondingly, write 
\begin{equation*}
a_{u}=\left[ 
\begin{array}{c}
a_{u_{1}} \\ 
0%
\end{array}%
\right] ,a_{y}=\left[ 
\begin{array}{c}
0 \\ 
a_{y_{1}}%
\end{array}%
\right] ,
\end{equation*}%
and the existence conditions for the kernel and image attacks are
respectively,%
\begin{eqnarray*}
\left[ 
\begin{array}{cc}
-\hat{N} & \text{ }\hat{M}%
\end{array}%
\right] \left[ 
\begin{array}{c}
a_{u_{1}} \\ 
0 \\ 
0 \\ 
a_{y_{1}}%
\end{array}%
\right] &=&0\Longleftrightarrow \left[ 
\begin{array}{c}
a_{u_{1}} \\ 
0 \\ 
0 \\ 
a_{y_{1}}%
\end{array}%
\right] \in \mathcal{K}_{G}, \\
\text{for some }\xi ,\left[ 
\begin{array}{c}
a_{u_{1}} \\ 
0 \\ 
0 \\ 
a_{y_{1}}%
\end{array}%
\right] &=&\left[ 
\begin{array}{c}
M \\ 
N%
\end{array}%
\right] \xi \in \mathcal{I}_{G}.
\end{eqnarray*}%
As described in Subsection \ref{subsection3-3-4}, zero dynamics attacks
studied in \cite{TEIXEIRA-zero-attack_2015,Hoehn2016,Sui2021} are a special
case of the above result.

We now attend to the attack actuability.

\begin{theorem}
\label{Theo4-7}Given the control system (\ref{eq4-2}) and the attack model (%
\ref{eq2-21}), the attacks $\left( a_{u},a_{y}\right) $ are unactuable if
and only if 
\begin{equation}
\left[ 
\begin{array}{c}
a_{u} \\ 
-a_{y}%
\end{array}%
\right] =\left[ 
\begin{array}{c}
-\hat{Y} \\ 
\hat{X}%
\end{array}%
\right] r,  \label{eq4-40}
\end{equation}%
where $r\neq 0$ is an arbitrary $\mathcal{L}_{2}$ bounded signal. 
\end{theorem}

\begin{proof}
Since $\forall r$%
\begin{equation*}
\left[ 
\begin{array}{cc}
X & \text{ }Y%
\end{array}%
\right] \left[ 
\begin{array}{c}
-\hat{Y} \\ 
\hat{X}%
\end{array}%
\right] r=0,
\end{equation*}%
it follows from (\ref{eq4-3}) that (\ref{eq4-40}) is a sufficient condition.
To prove the necessity, let $\left( a_{u},a_{y}\right) $ be 
\begin{equation*}
\left[ 
\begin{array}{c}
a_{u} \\ 
-a_{y}%
\end{array}%
\right] =\left[ 
\begin{array}{cc}
M & \text{ }-\hat{Y} \\ 
N & \text{ }\hat{X}%
\end{array}%
\right] \left[ 
\begin{array}{c}
r_{1} \\ 
r_{2}%
\end{array}%
\right]
\end{equation*}%
for some $\mathcal{L}_{2}$ bounded $\left( r_{1},r_{2}\right) .$
Consequently, 
\begin{equation*}
\left[ 
\begin{array}{cc}
X & \text{ }Y%
\end{array}%
\right] \left[ 
\begin{array}{c}
a_{u} \\ 
-a_{y}%
\end{array}%
\right] =r_{1}.
\end{equation*}%
As a result, according to (\ref{eq4-3}), the attacks are unactuable only if $%
r_{1}=0,$ which implies 
\begin{equation*}
\left[ 
\begin{array}{c}
a_{u} \\ 
-a_{y}%
\end{array}%
\right] =\left[ 
\begin{array}{c}
-\hat{Y} \\ 
\hat{X}%
\end{array}%
\right] r_{2}.
\end{equation*}%
The theorem is thus proven.
\end{proof}

Remember that $\left( \hat{X},-\hat{Y}\right) $ spans the image subspace of
the controller $\mathcal{I}_{K}.$ Hence, the condition (\ref{eq4-40}) is
equivalent to 
\begin{equation*}
\left[ 
\begin{array}{c}
a_{u} \\ 
-a_{y}%
\end{array}%
\right] \in \mathcal{I}_{K},
\end{equation*}%
which implies, on the other hand, if%
\begin{equation*}
\left[ 
\begin{array}{cc}
X & \text{ }Y%
\end{array}%
\right] \left[ 
\begin{array}{c}
a_{u} \\ 
-a_{y}%
\end{array}%
\right] \neq 0,
\end{equation*}%
then $\left( a_{u},a_{y}\right) $ are actuable.

It is interesting to notice the following two claims as the immediate
results of Theorems \ref{Theo4-6}-\ref{Theo4-7}.

\begin{corollary}
\label{Co4-1}Given the control system (\ref{eq4-2}) and the attack model (%
\ref{eq2-21}), (i) any actuable attacks $\left( a_{u},a_{y}\right) $ are
detectable in the information pattern $\left( u^{a},y\right) ,$ (ii) the
closed-loop dynamic $\left( u^{a},y\right) $ under any image attack, 
\begin{equation*}
\left[ 
\begin{array}{c}
a_{u} \\ 
-a_{y}%
\end{array}%
\right] =\left[ 
\begin{array}{c}
M \\ 
N%
\end{array}%
\right] \xi ,
\end{equation*}%
is governed by 
\begin{equation}
\left[ 
\begin{array}{c}
u^{a} \\ 
y%
\end{array}%
\right] =\left[ 
\begin{array}{c}
M \\ 
N%
\end{array}%
\right] \left( \xi +v\right) +\left[ 
\begin{array}{c}
-\hat{Y} \\ 
\hat{X}%
\end{array}%
\right] r_{y}.  \label{eq4-41}
\end{equation}
\end{corollary}

\begin{proof}
The claim (i) follows immediately from (\ref{eq4-3}), Theorem \ref{Theo4-7},
and Definition \ref{De4-1}. The claim (ii) with relation (\ref{eq4-41}) is
the result of Theorem \ref{Theo4-6} and the relation%
\begin{equation*}
\left[ 
\begin{array}{cc}
X & \text{ }Y%
\end{array}%
\right] \left[ 
\begin{array}{c}
M \\ 
N%
\end{array}%
\right] =I.
\end{equation*}
\end{proof}

The claim (i) is useful for detecting (stealthy) attacks, as will be
addressed in Subsection \ref{SubsectionV-A}, while the claim (ii) is
essential for dealing with attack design issues. In fact, the condition (\ref%
{eq4-25}) for kernel attacks builds the control-theoretic basis for attack
detection, and its dual (and equivalent) form, the condition (\ref{eq4-24a})
for the image attacks, is a capable model for designing cyber-attacks. We
would like to call the reader's attention to the fact that ANY (actuable)
attacks $\left( a_{u},a_{y}\right) $ act on the system dynamic $\left(
u^{a},y\right) $ via the signal $r_{a_{u}},$ namely 
\begin{equation}
\left[ 
\begin{array}{c}
u^{a} \\ 
y%
\end{array}%
\right] =\left[ 
\begin{array}{c}
M \\ 
N%
\end{array}%
\right] \left( r_{a_{u}}+v\right) +\left[ 
\begin{array}{c}
-\hat{Y} \\ 
\hat{X}%
\end{array}%
\right] r_{y},  \label{eq4-42}
\end{equation}%
Moreover, write $\left( a_{u},a_{y}\right) $ as 
\begin{equation*}
\left[ 
\begin{array}{c}
a_{u} \\ 
-a_{y}%
\end{array}%
\right] =\left[ 
\begin{array}{cc}
M & \text{ }-\hat{Y} \\ 
N & \text{ }\hat{X}%
\end{array}%
\right] \left[ 
\begin{array}{c}
\xi _{1} \\ 
\xi _{2}%
\end{array}%
\right] ,
\end{equation*}%
the closed-loop dynamic is governed by 
\begin{equation}
\left[ 
\begin{array}{c}
u^{a} \\ 
y%
\end{array}%
\right] =\left[ 
\begin{array}{c}
M \\ 
N%
\end{array}%
\right] \left( \xi _{1}+v\right) +\left[ 
\begin{array}{c}
-\hat{Y} \\ 
\hat{X}%
\end{array}%
\right] r_{y},  \label{eq4-43}
\end{equation}%
independent of $\xi _{2}.$ In other words, \textbf{it is sufficient to
construct image (stealthy) attacks exclusively.}

\begin{remark}
\label{Rem4-2}In our aforementioned discussion, the attack detectability and
actuability have been defined and analyzed in the context that cyber-attacks
are dedicated to degrade system control performance. Notice that in
industrial practice, system monitoring and fault diagnosis are often
implemented on the side of control station and in the information pattern $%
\left( u,y^{a}\right) .$ According to Theorem \ref{Theo4-4}, 
\begin{align*}
\left[ 
\begin{array}{c}
u \\ 
y^{a}%
\end{array}%
\right] & =\left[ 
\begin{array}{c}
M \\ 
N%
\end{array}%
\right] v+\left[ 
\begin{array}{c}
-\hat{Y} \\ 
\hat{X}%
\end{array}%
\right] \left( r_{y}-a_{r_{y}}\right) \Longrightarrow  \\
r_{y,cs}& =\left[ 
\begin{array}{cc}
-\hat{N} & \text{ }\hat{M}%
\end{array}%
\right] \left[ 
\begin{array}{c}
u \\ 
y^{a}%
\end{array}%
\right] =r_{y}-a_{r_{y}}.
\end{align*}%
Consequently, a fault diagnosis based on the residual $r_{y,cs}$ generated
at the control station can be significantly disturbed by the attack signal $%
a_{r_{y}}.$ Suppose that cyber-attacks are directed at the fault diagnosis
system rather than the control performance. Attackers may design the attack
pair $\left( a_{u},a_{y}\right) $ in such a way that 
\begin{equation*}
r_{a_{u}}=Xa_{u}-Ya_{y}=0,a_{r_{y}}=-\hat{M}a_{y}+\hat{N}a_{u}\neq 0,
\end{equation*}%
i.e. $\left( a_{u},a_{y}\right) $ is undetectable by means of the residual $%
r_{u},$ and actuable via the residual $r_{y,cs}.$ This can remarkably
increase false alarm rate, and even result in alarm flooding \cite{Wang2016}%
. The consequence is strongly reduced production rate and, when
fault-tolerant control action is triggered, degradation of control
performance as well. This engineering practical aspect is totally overlooked
in the research on secure CPCSs.
\end{remark}

\subsubsection{Vulnerability analysis and image attack design\label%
{subsectionIV-C-B}}

In the literature, in the context of vulnerability analysis, major attention
has been mainly paid to degrade the tracking performance, enlarge the state
estimation error, and increase linear quadratic Gaussian (LQG) control
performance \cite{JY-TC-2024,Shang2022attacks,ZhangTAC2023,LY-TAC-2022},
while limit research efforts have been made on deteriorating the feedback
control performance \cite{MS-TAC-2016}. Furthermore, most of the existing
attack design schemes are implemented based on some strict assumptions,
which limits the applications.

In this subsection, we analyze system vulnerability and delineate a scheme
of designing image attacks. Apart from the above-mentioned sufficiency of
constructing image attacks, further two advantages of this design scheme are
(i) the design is performed in the context of a constriant-free optimization
problem with respect to the system vulnerability, (ii) instead of $\left(
a_{u},a_{y}\right) \in \mathbb{R}^{p+m},$ the latent variable $\xi \in 
\mathbb{R}^{p}$ serves as the optimization variable, which may considerably
reduce optimization computations.

\begin{remark}
In the interature, optimal design of stealthy attacks is typically
formulated as an optimization problem with the stealthiness as a constraint.
It can be proven that such constraint conditions can be reformulated as the
kernel condition (\ref{eq4-25}) for strictly stealthy attacks or  
\begin{equation*}
\left\Vert \left[ 
\begin{array}{cc}
-\hat{N} & \text{ }\hat{M}%
\end{array}%
\right] \left[ 
\begin{array}{c}
a_{u} \\ 
-a_{y}%
\end{array}%
\right] \right\Vert \leq \varepsilon 
\end{equation*}%
for stealthy attacks, where $\left\Vert \cdot \right\Vert $ denotes some
norm or index, and $\varepsilon $ is a small constant.
\end{remark}

Vulnerability analysis is receiving increasing attention in the research
area of secure CPCSs and attack design \cite%
{MS-TAC-2016,Sui2021,LIUautomatica2024,ZHANGautomatica2024}. Informally
speaking, vulnerability describes attacker's ability to induce significant
degradation of a control performance like optimal control and estimation
objective, while remaining within the system's admissible operating range.
Below, we firstly introduce the concept of attack stability margin and, on
this basis, define the system vulnerability. In this context, image attack
design is addressed. To this end, we consider the control system (\ref{eq4-2}%
) and the attack model (\ref{eq2-21}), and assume that attackers possess the
process knowledge $\left( A,B,C,D\right) $ and are able to access the system
data $\left( u,y\right) $.

Theorem \ref{Theo4-5} reveals that the impact of cyber-attacks on the CPCS
dynamic is equivalent to variations in controller. Specifically, the altered
control gain is subject to%
\begin{equation*}
K_{\Delta }=-\left( X+\Delta _{X}\right) ^{-1}\left( Y+\Delta _{Y}\right) ,
\end{equation*}%
with $\left( \Delta _{X},\Delta _{Y}\right) $ denoting the variation caused
by attacks. One critical endangerment of such variations is the system
stability. In robust control theory, stability margin is a well-established
concept of quantifying the impact of model uncertainty on the closed-loop
stability \cite{KP-book,Zhou96}. According to Theorems \ref{Theo4-3} and \ref%
{Theo4-5}, attacks exclusively act on the system image subspace, the
stability margin of the system dynamic under attacks is to be defined in
terms of $\left( M,N\right) .$ Specifically, dual to the stability margin
with respect to the model uncertainty defined by \cite{KP-book,Zhou96} 
\begin{equation*}
\left\Vert \left[ 
\begin{array}{c}
K \\ 
I%
\end{array}%
\right] \left( I-GK\right) ^{-1}\hat{M}^{-1}\right\Vert _{\infty
}^{-1}=\left\Vert \left[ 
\begin{array}{c}
-\hat{Y} \\ 
\hat{X}%
\end{array}%
\right] \right\Vert _{\infty }^{-1},
\end{equation*}%
the loop dynamic%
\begin{gather}
\left[ 
\begin{array}{c}
I \\ 
G%
\end{array}%
\right] \left( I-K_{\Delta }G\right) ^{-1}\left( X+\Delta _{X}\right) ^{-1} 
\notag \\
=\left[ 
\begin{array}{c}
M \\ 
N%
\end{array}%
\right] \left( \left( X+\Delta _{X}\right) M+\left( Y+\Delta _{Y}\right)
N\right) ^{-1}  \notag \\
=\left[ 
\begin{array}{c}
M \\ 
N%
\end{array}%
\right] \left( I+\left[ 
\begin{array}{cc}
\Delta _{X} & \text{ }\Delta _{Y}%
\end{array}%
\right] \left[ 
\begin{array}{c}
M \\ 
N%
\end{array}%
\right] \right) ^{-1}  \label{eq4-45}
\end{gather}%
is conceded, from which it can be concluded that the loop dynamic is stable
when 
\begin{gather*}
\left\Vert \left[ 
\begin{array}{cc}
\Delta _{X} & \text{ }\Delta _{Y}%
\end{array}%
\right] \left[ 
\begin{array}{c}
M \\ 
N%
\end{array}%
\right] \right\Vert _{\infty }<1\Longrightarrow \\
\left\Vert \left[ 
\begin{array}{cc}
\Delta _{X} & \text{ }\Delta _{Y}%
\end{array}%
\right] \right\Vert _{\infty }\left\Vert \left[ 
\begin{array}{c}
M \\ 
N%
\end{array}%
\right] \right\Vert _{\infty }<1.
\end{gather*}%
In this context, we introduce the definition of attack stability margin.

\begin{definition}
\label{De4-3}Given the attacked CPCS model,   
\begin{equation}
\left\{ 
\begin{array}{l}
y=Gu^{a}+r_{y,0} \\ 
u^{a}=K_{\Delta }y+r_{u,a},%
\end{array}%
\right.   \label{eq4-44}
\end{equation}%
where $K_{\Delta }$ and $r_{u,a}$ are described in Theorem \ref{Theo4-5},
the value $b_{a},$ 
\begin{equation}
b_{a}=\left\Vert \left[ 
\begin{array}{c}
M \\ 
N%
\end{array}%
\right] \right\Vert _{\infty }^{-1},  \label{eq4-46}
\end{equation}%
is called attack stability margin. 
\end{definition}

It is apparent that the loop dynamic may become unstable, when $\left\Vert %
\left[ 
\begin{array}{cc}
\Delta _{X} & \text{ }\Delta _{Y}%
\end{array}%
\right] \right\Vert _{\infty }\geq b_{a}$. In other words, the value 
\begin{equation*}
\left\Vert \left[ 
\begin{array}{c}
M \\ 
N%
\end{array}%
\right] \right\Vert _{\infty }=b_{a}^{-1}
\end{equation*}%
reflects the degree of vulnerability. We now delineate a design scheme of
image attacks and the vulnerability in the context of system stability.

Let an image attack be described by%
\begin{equation*}
\left[ 
\begin{array}{c}
a_{u} \\ 
a_{y}%
\end{array}%
\right] =\Pi ^{a}\left[ 
\begin{array}{c}
u \\ 
y%
\end{array}%
\right] ,\Pi ^{a}:=\left[ 
\begin{array}{c}
M_{a} \\ 
N_{a}%
\end{array}%
\right] \Pi _{u}^{a},
\end{equation*}%
where $\left( M_{a},N_{a}\right) $ represents an SIR of the plant
constructed by the attackers with parameters $\left( F_{a},V_{a}\right) $,
and $\Pi _{u}^{a}$ is to be determined. The attack model (\ref{eq4-15}) is
then further written into 
\begin{gather*}
\left[ 
\begin{array}{c}
a_{u} \\ 
a_{y}%
\end{array}%
\right] =\Pi ^{a}\left( I+\left[ 
\begin{array}{cc}
I & 0 \\ 
0 & 0%
\end{array}%
\right] \Pi ^{a}\right) ^{-1}\left[ 
\begin{array}{c}
u^{a} \\ 
y%
\end{array}%
\right] \\
=\left[ 
\begin{array}{c}
M_{a} \\ 
N_{a}%
\end{array}%
\right] \Pi _{u}^{a}\left( I+\left[ 
\begin{array}{c}
M_{a} \\ 
0%
\end{array}%
\right] \Pi _{u}^{a}\right) ^{-1}\left[ 
\begin{array}{c}
u^{a} \\ 
y%
\end{array}%
\right] .
\end{gather*}%
Attributed to Theorem \ref{Theo3-1}, $\left( M,N\right) $ and $\left(
X,Y\right) $ are subject to 
\begin{gather*}
\left[ 
\begin{array}{c}
M \\ 
N%
\end{array}%
\right] =\left[ 
\begin{array}{c}
M_{a} \\ 
N_{a}%
\end{array}%
\right] T_{a},T_{a}\in \mathcal{RH}_{\infty }, \\
\left[ 
\begin{array}{cc}
X & \text{ }Y%
\end{array}%
\right] \left[ 
\begin{array}{c}
M_{a} \\ 
N_{a}%
\end{array}%
\right] =T_{a}^{-1},T_{a}^{-1}\in \mathcal{RH}_{\infty },
\end{gather*}%
for some $T_{a}.$ It follows from Theorem \ref{Theo4-5}%
\begin{gather}
\left( I+\left[ 
\begin{array}{cc}
\Delta _{X} & \text{ }\Delta _{Y}%
\end{array}%
\right] \left[ 
\begin{array}{c}
M \\ 
N%
\end{array}%
\right] \right) ^{-1}=\left( I-\left[ 
\begin{array}{cc}
X & \text{ }Y%
\end{array}%
\right] \Phi ^{a}\left[ 
\begin{array}{c}
M \\ 
N%
\end{array}%
\right] \right) ^{-1}  \notag \\
=\left( I-T_{a}^{-1}\Pi _{u}^{a}\left( I+\left[ 
\begin{array}{c}
M_{a} \\ 
0%
\end{array}%
\right] \Pi _{u}^{a}\right) ^{-1}\left[ 
\begin{array}{c}
M_{a} \\ 
N_{a}%
\end{array}%
\right] T_{a}\right) ^{-1}  \notag \\
=T_{a}^{-1}\left( I-\Pi _{u}^{a}\left( I+\left[ 
\begin{array}{c}
M_{a} \\ 
0%
\end{array}%
\right] \Pi _{u}^{a}\right) ^{-1}\left[ 
\begin{array}{c}
M_{a} \\ 
N_{a}%
\end{array}%
\right] \right) ^{-1}T_{a}.  \label{eq4-47}
\end{gather}%
Consequently, it is sufficient to consider 
\begin{equation}
\left( I-\Pi _{u}^{a}\left( I+\left[ 
\begin{array}{c}
M_{a} \\ 
0%
\end{array}%
\right] \Pi _{u}^{a}\right) ^{-1}\left[ 
\begin{array}{c}
M_{a} \\ 
N_{a}%
\end{array}%
\right] \right) ^{-1}  \label{eq4-50}
\end{equation}%
for examining the system stability. Note that the expression (\ref{eq4-50})
exclusively comprises the transfer functions designed by the attackers. Now,
impose 
\begin{equation*}
\Pi _{u}^{a}\left( I+\left[ 
\begin{array}{c}
M_{a} \\ 
0%
\end{array}%
\right] \Pi _{u}^{a}\right) ^{-1}=\left[ 
\begin{array}{cc}
\Pi _{1}^{a} & \text{ }\Pi _{2}^{a}%
\end{array}%
\right] \left[ 
\begin{array}{cc}
X_{a} & \text{ }Y_{a} \\ 
-\hat{N}_{a} & \text{ }\hat{M}_{a}%
\end{array}%
\right] ,
\end{equation*}%
$\left( \Pi _{1}^{a},\Pi _{2}^{a}\right) \in \mathcal{RH}_{\infty },$
leading to 
\begin{eqnarray*}
\Pi _{u}^{a} &=&\left( I-\left[ 
\begin{array}{cc}
\Pi _{1}^{a} & \text{ }\Pi _{2}^{a}%
\end{array}%
\right] \left[ 
\begin{array}{c}
X_{a} \\ 
-\hat{N}_{a}%
\end{array}%
\right] M_{a}\right) ^{-1}\cdot \\
&&\left[ 
\begin{array}{cc}
\Pi _{1}^{a} & \text{ }\Pi _{2}^{a}%
\end{array}%
\right] \left[ 
\begin{array}{cc}
X_{a} & \text{ }Y_{a} \\ 
-\hat{N}_{a} & \text{ }\hat{M}_{a}%
\end{array}%
\right] .
\end{eqnarray*}%
Hereby, $\left( X_{a},Y_{a}\right) $ and $\left( \hat{M}_{a},\hat{N}%
_{a}\right) $ are the LCPs of the controller and plant set by the attackers,
and $\left( \Pi _{1}^{a},\Pi _{2}^{a}\right) $ are selected so that 
\begin{equation}
\left( I-\left[ 
\begin{array}{cc}
\Pi _{1}^{a} & \text{ }\Pi _{2}^{a}%
\end{array}%
\right] \left[ 
\begin{array}{c}
X_{a} \\ 
-\hat{N}_{a}%
\end{array}%
\right] M_{a}\right) ^{-1}\in \mathcal{RH}_{\infty }.  \label{eq4-48}
\end{equation}%
Consequently, 
\begin{equation}
I-\Pi _{u}^{a}\left( I+\left[ 
\begin{array}{c}
M_{a} \\ 
0%
\end{array}%
\right] \Pi _{u}^{a}\right) ^{-1}\left[ 
\begin{array}{c}
M_{a} \\ 
N_{a}%
\end{array}%
\right] =I-\Pi _{1}^{a}.  \label{eq4-49}
\end{equation}%
We summarize the above results in the following theorem and definition.

\begin{theorem}
\label{Theo4-8}Given the attacked CPCS modelled by (\ref{eq4-44}) in
Definition \ref{De4-3}, if there exists LCPs and RCPs of the plant and
controller, $\left( \hat{M}_{a},\hat{N}_{a}\right) ,\left(
X_{a},Y_{a}\right) ,\left( M_{a},N_{a}\right) ,$ and $\left( \Pi
_{1}^{a},\Pi _{2}^{a}\right) \in \mathcal{RH}_{\infty }$ so that (i) the
condition (\ref{eq4-48}) holds, and (ii) the system $I-\Pi _{1}^{a}$ has a
zero in the closed right half-plane (an unstable zero), the loop dynamic
(\ref{eq4-45}) is unstable. 
\end{theorem}

\begin{proof}
The proof follows from (\ref{eq4-47})-(\ref{eq4-50}) and (\ref{eq4-49}).
\end{proof}

\begin{definition}
\label{De4-4}Given the attacked CPCS model (\ref{eq4-44}) and suppose that
the condition (\ref{eq4-48}) is satisfied, the system $\left( I-\Pi
_{1}^{a}\right) ^{-1}$ is called vulnerability model.
\end{definition}

It is worth emphasizing that all relevant transfer functions in (\ref{eq4-48}%
)-(\ref{eq4-49}) can be arbitrarily designed by the attackers. They pose no
strict conditions, as demonstrated by the following example.

\begin{example}
Consider a stable plant system $G\in \mathcal{RH}_{\infty }^{m\times
p},m\geq p.$ It is apparent that for $F_{a}=0,V_{a}=I,$ 
\begin{equation*}
\left( X_{a},Y_{a}\right) =\left( I,0\right) ,\left( M_{a},N_{a}\right)
=\left( I,G\right) .
\end{equation*}%
Accordingly, 
\begin{equation*}
I-\left[ 
\begin{array}{cc}
\Pi _{1}^{a} & \text{ }\Pi _{2}^{a}%
\end{array}%
\right] \left[ 
\begin{array}{c}
X_{a} \\ 
-\hat{N}_{a}%
\end{array}%
\right] M_{a}=I-\Pi _{1}^{a}+\Pi _{2}^{a}\hat{N}_{a}.
\end{equation*}%
We now select $\Pi _{2}^{a}\in \mathcal{RH}_{\infty }$ so that $\Pi _{2}^{a}%
\hat{N}_{a}$ is subject to 
\begin{equation}
\Pi _{2}^{a}\hat{N}_{a}=\frac{Z_{0}}{P_{0}}I,\frac{Z_{0}}{P_{0}}=\frac{\beta
_{\vartheta }z^{\vartheta }+\cdots +\beta _{0}}{\gamma _{\eta }z^{\eta
}+\cdots +\gamma _{0}},  \label{eq4-51}
\end{equation}%
where $Z_{0}$ contains all zeros of $\hat{N}_{a}$ in the closed right
half-plane, $\eta -\vartheta \geq 0$ (depending on the relative degree of $%
\hat{N}_{a}$), which ensures the existence of (\ref{eq4-51}), and $\gamma
_{i},i=0,\cdots ,\eta ,$ are arbitrarily selectable, so far the stability of 
$\Pi _{2}^{a}\hat{N}_{a}$ is guaranteed. Next, set 
\begin{gather*}
\Pi _{1}^{a}=\frac{\alpha _{\eta }z^{\eta }+\cdots +\alpha _{0}}{P_{0}}I\in 
\mathcal{RH}_{\infty }\Longrightarrow  \\
I-\Pi _{1}^{a}+\Pi _{2}^{a}\hat{N}_{a}=\frac{\Xi _{0}}{P_{0}}I,I-\Pi
_{1}^{a}=\frac{\Pi _{0}}{P_{0}}I, \\
\Xi _{0}=\left( \gamma _{\eta }-\alpha _{\eta }\right) z^{\eta }+\cdots
+\left( \gamma _{\vartheta +1}-\alpha _{\vartheta +1}\right) z^{\vartheta
+1}+ \\
\left( \gamma _{\vartheta }+\beta _{\vartheta }-\alpha _{\vartheta }\right)
z^{\vartheta }+\cdots +\gamma _{0}+\beta _{0}-\alpha _{0}, \\
\Pi _{0}=\left( \gamma _{\eta }-\alpha _{\eta }\right) z^{\eta }+\cdots
+\gamma _{0}-\alpha _{0},
\end{gather*}%
and determine $\alpha _{i},\gamma _{i},i=0,\cdots ,\eta ,$ such that the
zeros of polynomial $\Xi _{0}$ are in the left half-plane, and $\Pi _{0}$
contains zeros in the closed right half-plane. As a result, 
\begin{equation*}
\left( I-\left[ 
\begin{array}{cc}
\Pi _{1}^{a} & \text{ }\Pi _{2}^{a}%
\end{array}%
\right] \left[ 
\begin{array}{c}
X_{a} \\ 
-\hat{N}_{a}%
\end{array}%
\right] M_{a}\right) ^{-1}=\frac{P_{0}}{\Xi _{0}}I\in \mathcal{RH}_{\infty },
\end{equation*}%
and $\left( I-\Pi _{1}^{a}\right) ^{-1}=\frac{P_{0}}{\Pi _{0}}I$ is unstable.
\end{example}

\section{Detection, control, and system privacy}

We now attend to the simultaneous detection of the faults and cyber-attacks,
their tolerant and resilient control. At the end of this section, we will
address system privacy issues as well. For our purpose, the fault model (\ref%
{eq2-24}) and attack model (\ref{eq2-21}) are under consideration.

\subsection{Simultaneous detection of faults and cyber-attacks\label%
{SubsectionV-A}}

In industrial CPCSs, faults and attacks may occur concurrently, and their
effects on system operations are often similar and can be mutually masking.
Consequently, conventional diagnosis methods that detect faults and
cyber-attacks separately may lead to misdiagnosis, false alarms, and even
overlook threats. Congruously, simultaneous detection of faults and
cyber-attacks attracts considerable attention in industrial applications 
\cite{Manandhar2014,Ensansefat2014,ZKPP-CSL-2021,Taheri2024}. On the other
hand, it can be observed that endeavours to research simultaneous detection
of faults and attacks, in comparison with the current enthusiasm for attack
detection, and in particular attack design, are remarkably limited \cite%
{Zaman2022,Baroumand2023,RA-EJC-2023,Xue2025}. In the previous sections, it
is revealed that (i) there exists an attack information potential $IP_{a}$
on the plant and control station sides, (ii) there is no fault information
difference for the two information patterns, since the output residual $%
r_{y} $ identically affects $\left( u,y^{a}\right) $ and $\left(
u^{a},y\right) .$ In particular, Theorem \ref{Theo4-4} clearly showcases
that a simultaneous detection of faults and attacks in the information
pattern $\left( u,y^{a}\right) ,$ i.e. detection on the control station, is
challenging due to the strong coupling of faults and attacks in the
closed-loop data $\left( u,y^{a}\right) .$ These results motivate us to
leverage the unified framework to achieve the simultaneous detection in
different information patterns.

\subsubsection{Simultaneous detection scheme A}

It follows from Theorems \ref{Theo4-1}, \ref{Theo4-3} and (\ref{eq4-3}) that
the closed-loop dynamic under simultaneous existence of the faults and
attacks is governed by 
\begin{equation}
\left[ 
\begin{array}{c}
u^{a} \\ 
y%
\end{array}%
\right] =\left[ 
\begin{array}{c}
M \\ 
N%
\end{array}%
\right] \left( v+a_{r_{u}}\right) +\left[ 
\begin{array}{c}
-\hat{Y} \\ 
\hat{X}%
\end{array}%
\right] \left( 
\begin{array}{c}
\Phi ^{f}\left( \hat{N}_{d}d+\bar{f}\right) \\ 
+\Psi ^{f}\left( v+a_{r_{u}}\right)%
\end{array}%
\right) .  \label{eq4-20}
\end{equation}%
That means, in the information pattern $\left( u^{a},y\right) ,$ (i) the
faults exclusively act on the residual subspace, (ii) the cyber-attacks act
on the image and residual subspaces, whereby the influence on the residual
subspace is coupled with the (multiplicative) fault $\Psi ^{f}$. These
distinguishing features allow us to detect the faults and attacks
simultaneously on the plant side. To be specific, on the assumption that the
reference signal $v$ is known or equal to zero, an assumption mostly adopted
in the literature, detecting the injected cyber-attack $a_{r_{u}}$ by means
of the residual $r_{u},$%
\begin{equation}
r_{u}=Xu^{a}+Yy=v+a_{r_{u}},
\end{equation}%
is straightforward, namely $a_{r_{u}}=r_{u}-v.$ Considering that, during
fault-free operations, the dynamic of the residual generator is subject to 
\begin{equation}
r_{y}=\hat{M}y-\hat{N}u^{a}=\hat{N}_{d}d,
\end{equation}%
where the unknown input vector $d$ is either a stochastic process or an $%
\ell _{2}$-bounded signal and modelled by (\ref{eq2-23a}), an optimal fault
detection is achieved by adding an post-filter $R$ to the residual $r_{y},$%
\begin{equation*}
r=Rr_{y}=R\hat{N}_{d}d,
\end{equation*}%
whose design and the corresponding threshold setting are summarized in
Algorithm 1.

\begin{algorithm}
\textbf{Algorithm 1: Detection scheme A}
\begin{enumerate}
\item Do a co-inner-outer factorization of $\hat{N}_{d},$ 
\begin{equation}
\hat{N}_{d}=\hat{N}_{d}^{co}\hat{N}_{d}^{in};  \label{eq5-30}
\end{equation}
\item set $R$ as 
\begin{equation}
R=\left( \hat{N}_{d}^{co}\right) ^{-1};  \label{eq5-31}
\end{equation}
\item corresponding to the residual evaluation function,%
\begin{equation*}
J=\left\{ 
\begin{array}{l}
r^{T}(k)\Sigma _{r}^{-1}r(k),d=\left[ 
\begin{array}{c}
w \\ 
v%
\end{array}%
\right] \\ 
\left\Vert r\right\Vert _{2},d=\left[ 
\begin{array}{c}
d_{P} \\ 
d_{S}%
\end{array}%
\right] \text{,}%
\end{array}%
\right.
\end{equation*}%
set the threshold $J_{th}$ as%
\begin{equation*}
J_{th}=\left\{ 
\begin{array}{l}
\chi _{\alpha }^{2}\left( m\right) ,d=\left[ 
\begin{array}{c}
w \\ 
v%
\end{array}%
\right] \\ 
\delta _{d},d=\left[ 
\begin{array}{c}
d_{P} \\ 
d_{S}%
\end{array}%
\right]%
\end{array}%
\right.
\end{equation*}%
where $\alpha \in (0,1) $ is pre-defined false alarm rate.
\end{enumerate}
\end{algorithm}The control-theoretic background of the above algorithm is
the application of a co-inner-outer factorization of a transfer function
matrix to optimize an observer-based fault detection system with white
noises or an $\ell _{2}$-bounded unknown input vector. The reader is
referred to \cite{Ding2008,Ding2020} for details. As delineated in \cite%
{Ding2020}, the residual generation system satisfies

\begin{equation*}
\left\{ 
\begin{array}{l}
r\sim \mathcal{N}\left( 0,\Sigma _{r}\right) ,\text{and white},d=\left[ 
\begin{array}{c}
w \\ 
v%
\end{array}%
\right] \\ 
R\hat{N}_{d}\text{ is co-inner, }d=\left[ 
\begin{array}{c}
d_{P} \\ 
d_{S}%
\end{array}%
\right] .%
\end{array}%
\right.
\end{equation*}%
As a result, 
\begin{equation*}
\left\{ 
\begin{array}{l}
r^{T}(k)\Sigma _{r}^{-1}r(k)\sim \chi ^{2}\left( m\right) ,d=\left[ 
\begin{array}{c}
w \\ 
v%
\end{array}%
\right] \\ 
\left\Vert r\right\Vert _{2}\leq d,d=\left[ 
\begin{array}{c}
d_{P} \\ 
d_{S}%
\end{array}%
\right] .%
\end{array}%
\right.
\end{equation*}

\begin{remark}
In the literature, the $\chi ^{2}$ test statistic with the threshold setting 
$\chi _{\alpha }^{2}$ is widely applied for attack detection, even though
the corresponding residual signal is a color noise series. It leads to
higher false alarm rate and lower detectability. The co-inner-outer
factorization described in Algorithm 1 ensures the whiteness of the residual
signal. Correspondingly, the $\chi ^{2}$ test statistic delivers the maximal
detectability \cite{Ding2020}.
\end{remark}

It is worth emphasizing that the generation of the residual pair $\left(
r_{u},r_{y}\right) $ is realized by means of a single observer, as described
in Subsection \ref{SubsecII-A},%
\begin{equation*}
\left\{ 
\begin{array}{l}
\hat{x}(k+1)=\left( A-LC\right) \hat{x}(k)-\left( B-LD\right) u^{a}(k)-Ly(k)
\\ 
\left[ 
\begin{array}{c}
r_{u}(k) \\ 
r_{y}(k)%
\end{array}%
\right] =\left[ 
\begin{array}{c}
F \\ 
C%
\end{array}%
\right] \hat{x}(k)+\left[ 
\begin{array}{c}
I \\ 
-D%
\end{array}%
\right] u^{a}(k)+\left[ 
\begin{array}{c}
0 \\ 
I%
\end{array}%
\right] y(k).%
\end{array}%
\right.
\end{equation*}%
and implemented on the plant side. So far, the associated online
computations are moderate.

\subsubsection{Simultaneous detection scheme B, an improved auxiliary
system-based detection method \label{subsectionV-A-2}}

Roughly speaking, the basic idea of the auxiliary subsystem based detection
methods is to augment the plant with an auxiliary/monitoring subsystem whose
dynamics are unknown to the attacker, but similar to the one of the system.
The auxiliary subsystem is so constructed that deviations induced by an
attack become observable in auxiliary residuals \cite%
{MT-method-CDC2015,Zhang-CDC2017,MT-method-IEEE-TAC2021}. The auxiliary
system-based detection methods are mainly applied for detecting stealthy
attacks. In this regard, it can be plainly seen from the system dynamic (\ref%
{eq4-20}) and Theorem \ref{Theo4-3} that the input residual generator, 
\begin{equation}
r_{u}=Xu+Yy=v+\left[ 
\begin{array}{cc}
X & \text{ }Y%
\end{array}%
\right] \left[ 
\begin{array}{c}
a_{u} \\ 
a_{y}%
\end{array}%
\right] ,  \label{eq5-52}
\end{equation}%
is an ideal auxiliary subsystem for the detection purpose. Recall that on
the control station side, 
\begin{equation*}
\left[ 
\begin{array}{c}
u \\ 
y^{a}%
\end{array}%
\right] =\left[ 
\begin{array}{c}
M \\ 
N%
\end{array}%
\right] v+\left[ 
\begin{array}{c}
-\hat{Y} \\ 
\hat{X}%
\end{array}%
\right] \left( r_{y}-a_{r_{y}}\right) .
\end{equation*}%
Hence, in case of an undetectable faults (stealthy attacks), $a_{r_{y}}=0,$
leading to 
\begin{equation*}
\left[ 
\begin{array}{c}
u \\ 
y^{a}%
\end{array}%
\right] =\left[ 
\begin{array}{c}
M \\ 
N%
\end{array}%
\right] v+\left[ 
\begin{array}{c}
-\hat{Y} \\ 
\hat{X}%
\end{array}%
\right] r_{y}.
\end{equation*}%
Subsequently, simultaneous detection of faults and (stealthy) attacks can be
realized by (i) attack detection on the plant side using the detection
system (\ref{eq5-52}), (ii) the detection information is transmitted to the
control station via a well-encrypted channel, and (iii) fault detection by
means of the output residual system $r_{y}$ and Algorithm 1 running on the
control station.

\subsubsection{Simultaneous detection scheme C}

We now address the situation that attack detection is to be realized on the
control station side. For instance, the latent variable $v$ as a control
reference signal is calculated real-time. In this case, the detection scheme
A cannot be realized. In this regard, Theorem \ref{Theo4-3} plainly depicts
that the attacks exclusively act on the residual subspace and are tightly
coupled with the faults. The consequence is that an attack detection based
on the residual 
\begin{equation*}
\hat{M}y^{a}-\hat{N}u=r_{y}+\hat{M}a_{y}+\hat{N}a_{u}
\end{equation*}%
raises two problems, (i) attackers may easily construct stealthy attacks
when they possess system knowledge $\left( A,B,C,D\right) ,$ (ii) the strong
coupling between the residual $r_{y}$ and the attack function $\hat{M}a_{y}+%
\hat{N}a_{u}$ make it impossible to distinguish the faults from the attacks.
The following detection scheme is proposed to manage these two issues. Let 
\begin{equation*}
u^{a}=u_{0}+Q_{u}r_{u}+a_{u_{0}},Q_{u}\in \mathcal{RH}_{\infty },
\end{equation*}%
where $u_{0}=Ky^{a}+r_{u,0}$ is the input signal computed at the control
station and sent to the plant, $a_{u_{0}}$ is the attack injected into the
control channel, 
\begin{equation*}
r_{u}=Xu^{a}+Yy
\end{equation*}%
is computed on the plant side with $Q_{u}$ to be specified subsequently.
Note that $Q_{u}r_{u}$ is constructed and computed on the plant side and not
transmitted over the network. Moreover, concerning the system nominal
control performance, $Q_{u}r_{u}$ solely changes the system response to the
reference signal $v$, which can be fully recovered by re-setting the
reference (refer to (\ref{eq4-22})-(\ref{eq4-22a})). We now examine the
kernel-based model from the viewpoint of the control station, 
\begin{equation}
\left\{ 
\begin{array}{l}
Xu_{0}+Yy^{a}=Xr_{u,0}=v \\ 
\hat{M}y^{a}-\hat{N}u_{0}=r_{y}+\hat{M}a_{y}+\hat{N}a_{u_{0}}+\hat{N}%
Q_{u}r_{u}.%
\end{array}%
\right.
\end{equation}%
Subsequently, the information pattern at the control station is 
\begin{equation*}
\left[ 
\begin{array}{c}
u_{0} \\ 
y^{a}%
\end{array}%
\right] =\left[ 
\begin{array}{c}
M \\ 
N%
\end{array}%
\right] v+\left[ 
\begin{array}{c}
-\hat{Y} \\ 
\hat{X}%
\end{array}%
\right] \left( r_{y}+\hat{M}a_{y}+\hat{N}a_{u_{0}}+\hat{N}Q_{u}r_{u}\right) ,
\end{equation*}%
and the system dynamic is governed by 
\begin{gather*}
\left[ 
\begin{array}{c}
u^{a} \\ 
y%
\end{array}%
\right] =\left[ 
\begin{array}{c}
u_{0} \\ 
y^{a}%
\end{array}%
\right] +\left[ 
\begin{array}{c}
Q_{u}r_{u} \\ 
0%
\end{array}%
\right] +\left[ 
\begin{array}{c}
a_{u_{0}} \\ 
-a_{y}%
\end{array}%
\right] \\
=\left[ 
\begin{array}{c}
M \\ 
N%
\end{array}%
\right] \left( v+XQ_{u}r_{u}+Xa_{u_{0}}-Ya_{y}\right) +\left[ 
\begin{array}{c}
-\hat{Y} \\ 
\hat{X}%
\end{array}%
\right] r_{y}.
\end{gather*}%
The last equation is attributed to (\ref{eq4-4})-(\ref{eq4-4a}). Observe
that 
\begin{equation*}
r_{u}=Xu+Yy^{a}+Xa_{u_{0}}-Ya_{y}+XQ_{u}r_{u}.
\end{equation*}%
Suppose that $Q_{u}$ is so selected that $\left( I-XQ_{u}\right) ^{-1}\in 
\mathcal{RH}_{\infty }.$ It turns out%
\begin{equation}
r_{u}=\left( I-XQ_{u}\right) ^{-1}\left( Xa_{u_{0}}-Ya_{y}+v\right) .
\end{equation}%
As a result, 
\begin{gather}
\left[ 
\begin{array}{c}
u_{0} \\ 
y^{a}%
\end{array}%
\right] =\left( \left[ 
\begin{array}{c}
M \\ 
N%
\end{array}%
\right] +\left[ 
\begin{array}{c}
-\hat{Y} \\ 
\hat{X}%
\end{array}%
\right] Q\right) v  \notag \\
+\left[ 
\begin{array}{c}
-\hat{Y} \\ 
\hat{X}%
\end{array}%
\right] \left( r_{y}+\hat{M}a_{y}+\hat{N}a_{u_{0}}+Q\left(
Xa_{u_{0}}-Ya_{y}\right) \right) ,  \label{eq4-21} \\
\left[ 
\begin{array}{c}
u^{a} \\ 
y%
\end{array}%
\right] =\left[ 
\begin{array}{c}
M \\ 
N%
\end{array}%
\right] \bar{v}+\left[ 
\begin{array}{c}
-\hat{Y} \\ 
\hat{X}%
\end{array}%
\right] r_{y}  \notag \\
+\left[ 
\begin{array}{c}
M \\ 
N%
\end{array}%
\right] \left( I-XQ_{u}\right) ^{-1}\left( Xa_{u_{0}}-Ya_{y}\right) ,
\label{eq4-22} \\
Q=\hat{N}Q_{u}\left( I-XQ_{u}\right) ^{-1},v=\left( I-XQ_{u}\right) \bar{v}.
\label{eq4-22a}
\end{gather}%
Equation (\ref{eq4-22}) describes the loop dynamic. It is clear that adding $%
Q_{u}r_{u}$ in the control signal results in no change in the system nominal
(attack-free) dynamic, when $\bar{v}$ is set according to (\ref{eq4-22a}) on
the control station side. Subsequently, 
\begin{equation}
r_{y}=\hat{M}y-\hat{N}u^{a}=\Phi ^{f}\left( \hat{N}_{d}d+\bar{f}\right)
+\Psi ^{f}v,  \label{eq4-23}
\end{equation}%
attributed to Theorem \ref{Theo4-1}. Thus, detecting the faults can be
realized by running Algorithm 1 on the plant side, and it is decoupled from
the cyber-attacks. If a fault is detected, an alarm is released and the
corresponding message is sent to the control station.

To detect the cyber-attacks on the control station side, the residual signal 
$r_{u,ad}$ is built, 
\begin{align}
r_{u,ad}& =\left[ 
\begin{array}{cc}
-\hat{N} & \text{ }\hat{M}%
\end{array}%
\right] \left( \left[ 
\begin{array}{c}
u_{0} \\ 
y^{a}%
\end{array}%
\right] -\left( \left[ 
\begin{array}{c}
M \\ 
N%
\end{array}%
\right] +\left[ 
\begin{array}{c}
-\hat{Y} \\ 
\hat{X}%
\end{array}%
\right] Q\right) v\right)  \notag \\
& =\hat{M}y^{a}-\hat{N}u_{0}-Qv  \label{eq5-53}
\end{align}%
which, according to (\ref{eq4-21}), yields%
\begin{equation}
r_{u,ad}=r_{y}+\hat{M}a_{y}+\hat{N}a_{u_{0}}+Q\left(
Xa_{u_{0}}-Ya_{y}\right) .  \label{eq4-24}
\end{equation}%
Concerning stealthy attacks, consider 
\begin{align*}
a_{u,y}& :=\hat{M}a_{y}+\hat{N}a_{u_{0}}+Q\left( Xa_{u_{0}}-Ya_{y}\right) \\
& =\left[ 
\begin{array}{cc}
Q & \text{ }I%
\end{array}%
\right] \left[ 
\begin{array}{cc}
X & \text{ }Y \\ 
-\hat{N} & \text{ }\hat{M}%
\end{array}%
\right] \left[ 
\begin{array}{c}
a_{u_{0}} \\ 
-a_{y}%
\end{array}%
\right] .
\end{align*}%
Although attackers possess system knowledge of $\left( A,B,C,D\right) ,$
they are limited to access the process data $\left( u_{0},y\right) $ rather
than $\left( u,y\right) .$ Consequently, it is hard for attackers to
identify $Q_{u}$ and thus $Q.$ In other words, the missing information about 
$Q$ poses a hard existence condition, i.e.%
\begin{equation*}
a_{u,y}=\hat{M}a_{y}+\hat{N}a_{u_{0}}+Q\left( Xa_{u_{0}}-Ya_{y}\right) =0,
\end{equation*}%
for realizing stealthy attacks.

As far as no alarm for a fault is released, the attacks can be detected
based on the residual signal (\ref{eq4-24}), 
\begin{equation}
r_{u,ad}=\hat{N}_{d}d+a_{u,y},a_{u,y}\neq 0,
\end{equation}%
using Algorithm 1. It is apparent that in the case of a detected fault, it
holds 
\begin{equation}
r_{u,ad}=\Phi ^{f}\left( \hat{N}_{d}d+\bar{f}\right) +\Psi ^{f}v+r_{u,ad}.
\end{equation}%
Thus, it is unrealistic and, from the engineering viewpoint, also
unnecessary to detect the attacks. The implementation of the detection
scheme C is summarized in Algorithm 2.

\begin{algorithm}
\textbf{Algorithm 2: Detection scheme C}
\begin{enumerate}
\item Implementation of $Q_{u}r_{u}$ and $u=u_{0}^{a}+Q_{u}r_{u}$ on the
plant side;
\item Implementation of $r_{y}=\hat{M}y-\hat{N}u^{a}$ and Algorithm 1 for
fault detection, transmission of the fault information once the faults are detected;
\item Implementation of the residual generator (\ref{eq5-53}) and Algorithm
1 for attack detection.
\end{enumerate}
\end{algorithm}It is emphasized again that the core of generating the
residual pair $\left( r_{u,ad},r_{y}\right) $ is two observer-based residual
generators driven by $\left( u^{a},y\right) $ and $\left( u_{0},y^{a}\right) 
$ and realized on the plant and control station sides, respectively, as
described in Subsection \ref{SubsecII-A}.

At the end of this subsection, we would like to remark that (i) in
engineering applications, it is possible to enhance the detection
reliability by a combined use of the proposed detections schemes, (ii)
thanks to the advantage of the use of the identical observer-based residual
generators, the demanded online computation is less costly, (iii) aiming at
fault detection, the (output) residual generated on the plant side is
reliable and capable, in particular, when cyber-attacks are dedicated to the
fault detection system running at the control station (refer to Remark \ref%
{Rem4-2}), and (iv) there exists no change of the nominal system control
performance.

\subsection{Fault-tolerant and attack resilient feedback control \label%
{SubsectionV-B}}

As described in the section of \textit{Problem Formulation,} there are few
publications dedicated to simultaneous fault-tolerant and attack resilient
feedback control. In this subsection, we leverage the unified framework and
robust control methods to propose an integrated fault-tolerant and attack
resilient control scheme\ aiming at mitigating the control performance
degradation caused by the faults and attacks. The basic idea behind this
scheme is to configure the CPCS in the structure form of Variation II
described in Subsection \ref{subsection3-3-1}, which enables us to optimally
make use of (i) the attack information potential $IP_{a}$, (ii) the
additional degree of design freedom, and (iii) the duality between faulty
and attack dynamics, to obviate performance degradation.

\subsubsection{The control scheme}

Consider the attack information potential $IP_{a}$ (\ref{eq4-3b}), the loop
dynamics on the plant side (\ref{eq4-3}) and on the control station side (%
\ref{eq4-5}). It becomes obvious that the impact of $a_{r_{u}}$ on the
information pattern $\left( u^{a},y\right) $ and $a_{r_{y}}$ on the
information pattern $\left( u,y^{a}\right) $ results in the attack
information potential. In order to reduce $IP_{a}$ and the influence of $%
a_{r_{u}}$ on the system dynamic $\left( u^{a},y\right) ,$ compensating $%
a_{r_{u}}$ by feeding back $a_{r_{y}}$ is a potential solution. This is the
basic idea behind the control scheme described below. The controller
comprises two sub-controllers, respectively located on the both sides of the
CPCS,%
\begin{align}
u_{1}& =Ky^{a}+r_{u,0},  \label{eq4-28} \\
K& =-\left( X+Q_{1}\hat{N}\right) ^{-1}\left( Y-Q_{1}\hat{M}\right) ,  \notag
\\
u_{2}& =Q_{2}\left( \hat{M}y-\hat{N}u^{a}\right) ,Q_{1},Q_{2}\in \mathcal{RH}%
_{\infty }.
\end{align}%
The sub-controller $u_{1}$ is implemented on the control station side, while 
$u_{2}$ is embedded in the plant side. The overall controller $u$ acting on
the plant is 
\begin{equation}
u^{a}=u_{1}+u_{2}+a_{u_{1}}.  \label{eq4-26}
\end{equation}%
Notice that the controller (\ref{eq4-26}) is a stabilizing controller and
PnP-configured (refer to Theorem \ref{theo3-2}). During attack-free
operations, it is equivalent to 
\begin{align*}
u& =Ky+r_{u,0}, \\
K& =-\left( X+Q\hat{N}\right) ^{-1}\left( Y-Q\hat{M}\right) , \\
Q& =Q_{1}+\left( X+Q_{1}\hat{N}\right) Q_{2},
\end{align*}%
whose observer-based realization is given by%
\begin{equation*}
u=F\hat{x}+\left( Q_{1}+Q_{2}\right) \left( \hat{M}y-\hat{N}u\right) .
\end{equation*}%
It is of considerable importance to emphasize that two output residuals are
embedded in the controller and they serve for different control purposes.
While the residual $\hat{M}y^{a}-\hat{N}u$ is generated on the control
station side and used for resilient control of cyber-attacks, the residual $%
r_{y}=\hat{M}y-\hat{N}u^{a},$ added to the controller on the plant side, is
applied to enhance the fault-tolerance. To highlight this core idea, we now
examine the closed-loop dynamic.

Consider the kernel-based model of the closed-loop,%
\begin{gather*}
\hat{M}y-\hat{N}u^{a}=r_{y}, \\
\left( X+Q_{1}\hat{N}\right) u^{a}+\left( Y-Q_{1}\hat{M}\right) y \\
=\left[ 
\begin{array}{cc}
X+Q_{1}\hat{N} & \text{ }Y-Q_{1}\hat{M}%
\end{array}%
\right] \left( \left[ 
\begin{array}{c}
u \\ 
y^{a}%
\end{array}%
\right] +\left[ 
\begin{array}{c}
a_{u_{1}} \\ 
a_{y}%
\end{array}%
\right] \right) \\
=v+\bar{Q}_{2}r_{y}+\left( X+Q_{1}\hat{N}\right) a_{u_{1}}-\left( Y-Q_{1}%
\hat{M}\right) a_{y}, \\
v=\left( X+Q_{1}\hat{N}\right) r_{u,0},\bar{Q}_{2}=\left( X+Q_{1}\hat{N}%
\right) Q_{2}.
\end{gather*}%
It follows from (\ref{eq2-6a}) that 
\begin{gather}
\left[ 
\begin{array}{c}
u^{a} \\ 
y%
\end{array}%
\right] =\left[ 
\begin{array}{c}
M \\ 
N%
\end{array}%
\right] v+\left( \left[ 
\begin{array}{c}
-\hat{Y} \\ 
\hat{X}%
\end{array}%
\right] +\left[ 
\begin{array}{c}
M \\ 
N%
\end{array}%
\right] Q\right) r_{y}  \notag \\
+\left[ 
\begin{array}{c}
M \\ 
N%
\end{array}%
\right] \left( \left[ 
\begin{array}{cc}
X & \text{ }Y%
\end{array}%
\right] -Q_{1}\left[ 
\begin{array}{cc}
-\hat{N} & \text{ }\hat{M}%
\end{array}%
\right] \right) \left[ 
\begin{array}{c}
a_{u_{1}} \\ 
-a_{y}%
\end{array}%
\right] .  \label{eq4-27}
\end{gather}%
It is apparent that solving the MMP 
\begin{equation}
J_{1}=\inf_{Q_{1}\in \mathcal{RH}_{\infty }}\left\Vert \left[ 
\begin{array}{cc}
X & \text{ }Y%
\end{array}%
\right] -Q_{1}\left[ 
\begin{array}{cc}
-\hat{N} & \text{ }\hat{M}%
\end{array}%
\right] \right\Vert _{\infty }  \label{eq4-29}
\end{equation}%
leads to the optimal resilience to the attacks, while the fault-tolerance is
enhanced by solving the following MMP 
\begin{equation}
\inf_{Q_{2}\in \mathcal{RH}_{\infty }}\left\Vert \left[ 
\begin{array}{c}
-\hat{Y} \\ 
\hat{X}%
\end{array}%
\right] +\left[ 
\begin{array}{c}
M \\ 
N%
\end{array}%
\right] \left( Q_{1}+\left( X+Q_{1}\hat{N}\right) Q_{2}\right) \right\Vert
_{\infty }  \label{eq4-30}
\end{equation}%
for given $Q_{1}.$ It is of considerable interest to study the case 
\begin{equation}
\left( X+Q_{1}\hat{N}\right) ^{-1}\in \mathcal{RH}_{\infty },  \label{eq4-36}
\end{equation}%
that is, the controller $K$ is stable. Subsequently, for $\bar{Q}_{1}=\left(
X+Q_{1}\hat{N}\right) ^{-1}\left( \bar{Q}_{1}-Q_{1}\right) ,$ the MMP (\ref%
{eq4-30}) is equivalent to 
\begin{equation}
\bar{J}_{1}=\inf_{\bar{Q}_{1}\in \mathcal{RH}_{\infty }}\left\Vert \left[ 
\begin{array}{c}
-\hat{Y} \\ 
\hat{X}%
\end{array}%
\right] +\left[ 
\begin{array}{c}
M \\ 
N%
\end{array}%
\right] \bar{Q}_{1}\right\Vert _{\infty }.  \label{eq4-31}
\end{equation}%
The fault-tolerant and attack resilient controller (\ref{eq4-26}) is
triggered when faults or/and attacks are detected. The system operation
modes and the corresponding controllers are summarized below.

\begin{algorithm}
\textbf{System operation modes and the corresponding controllers}
\begin{enumerate}
\item Fault- and attack-free operation: $u=Ky+r_{u,0},$ running on the
control station side with $K$ subject to (\ref{eq2-18}); 
\item Faulty and attack-free operation: $u=Ky+r_{u,0}+Q_{2}r_{y},K$ is
subject to (\ref{eq2-18}), and $u_{2}=Q_{2}r_{y}$ is added on the plant side;
\item Operation under attacks but fault-free: $u=Ky+r_{u,0},$ running on the
control station side with $K$ subject to (\ref{eq4-28});
\item Operation under attacks and faults: $u=Ky+r_{u,0}+Q_{2}r_{y},K$ is subject
to (\ref{eq4-28}), and $u_{2}=Q_{2}r_{y}$ is added on the plant side.
\end{enumerate}
\end{algorithm}

\subsubsection{On controller design}

Recall that $\left( \hat{M},\hat{N}\right) $ and $\left( M,N\right) $ are
the dual coprime factorization pairs of the plant, while $\left( X,Y\right) $
and $\left( \hat{X},\hat{Y}\right) $ are the dual pairs of the equivalent
"controller", and they are subject to the Bezout identity (\ref{eq2-6}). In
this context, the MMPs (\ref{eq4-29}) and (\ref{eq4-31}) are presented in a
dual form. In other words, the proposed control scheme enables us to solve
the fault-tolerant control and attack resilient control as two dual and
independent MMPs. This fact motivates us to have a close look at these two
dual MMPs. To begin with, a special case with the normalized LCF and RCF of $%
G,\left( \hat{M}_{0},\hat{N}_{0}\right) $ and $\left( M_{0},N_{0}\right) ,$
is considered. The corresponding LCF and RCF of the controller are denoted
by $\left( X_{0},Y_{0}\right) $ and $\left( \hat{X}_{0},\hat{Y}_{0}\right) ,$
respectively. Observe that 
\begin{gather*}
\left\Vert \left[ 
\begin{array}{cc}
X_{0} & \text{ }Y_{0}%
\end{array}%
\right] -Q_{1}\left[ 
\begin{array}{cc}
-\hat{N}_{0} & \text{ }\hat{M}_{0}%
\end{array}%
\right] \right\Vert _{\infty } \\
=\left\Vert \left( \left[ 
\begin{array}{cc}
X_{0} & \text{ }Y_{0}%
\end{array}%
\right] -Q_{1}\left[ 
\begin{array}{cc}
-\hat{N}_{0} & \text{ }\hat{M}_{0}%
\end{array}%
\right] \right) \left[ 
\begin{array}{cc}
-\hat{N}_{0}^{\ast } & M_{0} \\ 
\hat{M}_{0}^{\ast } & N_{0}%
\end{array}%
\right] \right\Vert _{\infty } \\
=\left\Vert \left[ 
\begin{array}{cc}
Y_{0}\hat{M}_{0}^{\ast }-X_{0}\hat{N}_{0}^{\ast }-Q_{1} & \text{ }I%
\end{array}%
\right] \right\Vert _{\infty }, \\
\left\Vert \left[ 
\begin{array}{c}
-\hat{Y} \\ 
\hat{X}%
\end{array}%
\right] +\left[ 
\begin{array}{c}
M_{0} \\ 
N_{0}%
\end{array}%
\right] \bar{Q}_{1}\right\Vert _{\infty } \\
=\left\Vert \left[ 
\begin{array}{cc}
M_{0}^{\ast } & \text{ }N_{0}^{\ast } \\ 
-\hat{N}_{0} & \text{ }\hat{M}_{0}%
\end{array}%
\right] \left( \left[ 
\begin{array}{c}
-\hat{Y}_{0} \\ 
\hat{X}_{0}%
\end{array}%
\right] +\left[ 
\begin{array}{c}
M_{0} \\ 
N_{0}%
\end{array}%
\right] \right) \bar{Q}_{1}\right\Vert _{\infty } \\
=\left\Vert \left[ 
\begin{array}{c}
N_{0}^{\ast }\hat{X}_{0}-M_{0}^{\ast }\hat{Y}_{0}+\bar{Q}_{1} \\ 
I%
\end{array}%
\right] \right\Vert _{\infty }.
\end{gather*}%
Hence, it holds%
\begin{align*}
J_{1}& =\left( J_{0}^{2}+1\right) ^{1/2},J_{0}=\inf_{Q_{1}\in \mathcal{RH}%
_{\infty }}\left\Vert Y_{0}\hat{M}_{0}^{\ast }-X_{0}\hat{N}_{0}^{\ast
}-Q_{1}\right\Vert _{\infty }, \\
\bar{J}_{1}& =\left( \bar{J}_{0}^{2}+1\right) ^{1/2},\bar{J}_{0}=\inf_{\bar{Q%
}_{1}\in \mathcal{RH}_{\infty }}\left\Vert N_{0}^{\ast }\hat{X}%
_{0}-M_{0}^{\ast }\hat{Y}_{0}+\bar{Q}_{1}\right\Vert _{\infty }.
\end{align*}%
Next, we repeat the above procedure to examine the general case. It follows
from Corollary \ref{Co3-1} that%
\begin{gather}
\left( \left[ 
\begin{array}{cc}
X & \text{ }Y%
\end{array}%
\right] -Q_{1}\left[ 
\begin{array}{cc}
-\hat{N} & \text{ }\hat{M}%
\end{array}%
\right] \right) \left[ 
\begin{array}{cc}
-\hat{N}_{0}^{\ast } & M_{0} \\ 
\hat{M}_{0}^{\ast } & N_{0}%
\end{array}%
\right]  \notag \\
=\left( \left[ 
\begin{array}{cc}
X_{0} & \text{ }Y_{0}%
\end{array}%
\right] -Q_{0}\left[ 
\begin{array}{cc}
-\hat{N}_{0} & \text{ }\hat{M}_{0}%
\end{array}%
\right] \right) \left[ 
\begin{array}{cc}
-\hat{N}_{0}^{\ast } & M_{0} \\ 
\hat{M}_{0}^{\ast } & N_{0}%
\end{array}%
\right]  \notag \\
=\left[ 
\begin{array}{cc}
Y_{0}\hat{M}_{0}^{\ast }-X_{0}\hat{N}_{0}^{\ast }-Q_{0} & \text{ }I%
\end{array}%
\right] ,Q_{0}=Q_{1}R_{0}^{-1}+T_{0}^{-1}\bar{T}_{0},  \label{eq4-35a} \\
\left[ 
\begin{array}{cc}
M_{0}^{\ast } & \text{ }N_{0}^{\ast } \\ 
-\hat{N}_{0} & \text{ }\hat{M}_{0}%
\end{array}%
\right] \left( \left[ 
\begin{array}{c}
-\hat{Y} \\ 
\hat{X}%
\end{array}%
\right] +\left[ 
\begin{array}{c}
M \\ 
N%
\end{array}%
\right] \bar{Q}_{1}\right)  \notag \\
=\left[ 
\begin{array}{cc}
M_{0}^{\ast } & \text{ }N_{0}^{\ast } \\ 
-\hat{N}_{0} & \text{ }\hat{M}_{0}%
\end{array}%
\right] \left( \left[ 
\begin{array}{c}
-\hat{Y}_{0} \\ 
\hat{X}_{0}%
\end{array}%
\right] +\left[ 
\begin{array}{c}
M_{0} \\ 
N_{0}%
\end{array}%
\right] \bar{Q}_{0}\right)  \notag \\
=\left[ 
\begin{array}{c}
N_{0}^{\ast }\hat{X}_{0}-M_{0}^{\ast }\hat{Y}_{0}+\bar{Q}_{0} \\ 
I%
\end{array}%
\right] ,\bar{Q}_{0}=\bar{R}_{0}R_{0}^{-1}+T_{0}^{-1}\bar{Q}_{1},
\label{eq4-35b}
\end{gather}%
where $R_{0}^{-1},\bar{R}_{0},T_{0},\bar{T}_{0}$ are defined in Corollary %
\ref{Co3-1}. Accordingly, the closed-loop dynamic (\ref{eq4-27}) is
equivalent to 
\begin{gather}
\left[ 
\begin{array}{c}
u^{a} \\ 
y%
\end{array}%
\right] =\left[ 
\begin{array}{c}
M_{0} \\ 
N_{0}%
\end{array}%
\right] v_{0}+\left( \left[ 
\begin{array}{c}
-\hat{Y}_{0} \\ 
\hat{X}_{0}%
\end{array}%
\right] +\left[ 
\begin{array}{c}
M_{0} \\ 
N_{0}%
\end{array}%
\right] \bar{Q}_{0}\right) r_{0}  \notag \\
+\left[ 
\begin{array}{c}
M_{0} \\ 
N_{0}%
\end{array}%
\right] \left( \left[ 
\begin{array}{cc}
X_{0} & \text{ }Y_{0}%
\end{array}%
\right] -Q_{0}\left[ 
\begin{array}{cc}
-\hat{N}_{0} & \text{ }\hat{M}_{0}%
\end{array}%
\right] \right) \left[ 
\begin{array}{c}
a_{u_{1}} \\ 
-a_{y}%
\end{array}%
\right] ,  \label{eq4-34} \\
v_{0}=\left[ 
\begin{array}{cc}
X_{0} & \text{ }Y_{0}%
\end{array}%
\right] \left[ 
\begin{array}{c}
u^{a} \\ 
y%
\end{array}%
\right] ,r_{0}=\left[ 
\begin{array}{cc}
-\hat{N}_{0} & \text{ }\hat{M}_{0}%
\end{array}%
\right] \left[ 
\begin{array}{c}
u^{a} \\ 
y%
\end{array}%
\right] .  \notag
\end{gather}%
Since the reference signal and the residual signal can be constructed by the
designer, the dynamic (\ref{eq4-34}) is identical with (\ref{eq4-27}). The
above results are summarized as in the form of a theorem.

\begin{theorem}
\label{Theo5-0}Given the control system model (\ref{eq2-20}) with the
controller $u$ (\ref%
{eq4-26}), the optimal controller design problems (%
\ref{eq4-29}) and (\ref%
{eq4-31}) are equivalent to 
\begin{align}%
J_{1}& =\left( J_{0}^{2}+1\right) ^{1/2},J_{0}=\inf_{Q_{0}\in \mathcal{RH}%
_{\infty }}\left\Vert Y_{0}\hat{M}_{0}^{\ast }-X_{0}\hat{N}%
_{0}^{\ast
}+Q_{0}\right\Vert _{\infty },  \label{eq4-37} \\
\bar{J}_{1}&
=\left( \bar{J}_{0}^{2}+1\right) ^{1/2},\bar{J}_{0}=\inf_{\bar{Q%
}_{0}\in 
\mathcal{RH}_{\infty }}\left\Vert N_{0}^{\ast }\hat{X}%
_{0}-M_{0}^{\ast }%
\hat{Y}_{0}+\bar{Q}_{0}\right\Vert _{\infty }.
\label{eq4-38}
\end{align}
\end{theorem}

\begin{proof}
The proof follows from (\ref{eq4-34}) and (\ref{eq4-35a})-(\ref{eq4-35b}).
\end{proof}

\subsubsection{System performance}

We analyze the system performance under two aspects, control difference and
stability. For the simplicity and without loss of application generality, it
is supposed that the condition (\ref{eq4-36}), i.e. a stable controller,
holds, and the optimal fault-tolerant and attack-resilient controller is
applied, namely, $Q_{0}$ and $\bar{Q}_{0}$ are the solutions of the
optimization problems (\ref{eq4-37}) and (\ref{eq4-38}). We begin with
systems under additive faults and cyber-attacks. It follows from (\ref%
{eq4-34}), the fault and attack models (\ref{eq2-21}) and (\ref{eq2-24})
that the control difference is given by 
\begin{gather}
\left[ 
\begin{array}{c}
e_{u} \\ 
e_{y}%
\end{array}%
\right] =\left[ 
\begin{array}{c}
u^{a} \\ 
y%
\end{array}%
\right] -\left[ 
\begin{array}{c}
M_{0} \\ 
N_{0}%
\end{array}%
\right] v_{0}=  \notag \\
\left[ 
\begin{array}{c}
-\hat{Y}_{0,\bar{Q}_{0}} \\ 
\hat{X}_{0,\bar{Q}_{0}}%
\end{array}%
\right] \left( \hat{N}_{d,0}d+f\right) +\left[ 
\begin{array}{c}
M_{0} \\ 
N_{0}%
\end{array}%
\right] \left[ 
\begin{array}{cc}
X_{0,Q_{0}} & \text{ \ }Y_{0,Q_{0}}%
\end{array}%
\right] \left[ 
\begin{array}{c}
\eta _{u} \\ 
-\eta _{y}%
\end{array}%
\right]  \label{eq5-1} \\
\left[ 
\begin{array}{c}
-\hat{Y}_{0,\bar{Q}_{0}} \\ 
\hat{X}_{0,\bar{Q}_{0}}%
\end{array}%
\right] =\left[ 
\begin{array}{c}
-\hat{Y}_{0} \\ 
\hat{X}_{0}%
\end{array}%
\right] +\left[ 
\begin{array}{c}
M_{0} \\ 
N_{0}%
\end{array}%
\right] \bar{Q}_{0},  \notag \\
\left[ 
\begin{array}{cc}
X_{0,Q_{0}} & \text{ \ }Y_{0,Q_{0}}%
\end{array}%
\right] =\left[ 
\begin{array}{cc}
X_{0} & \text{ \ }Y_{0}%
\end{array}%
\right] -Q_{0}\left[ 
\begin{array}{cc}
-\hat{N}_{0} & \text{ \ }\hat{M}_{0}%
\end{array}%
\right] ,  \notag \\
\hat{N}_{d,0}=\hat{M}_{0}G_{d}=\left(
A-L_{0}C,E_{d}-F_{d}L_{0},W_{0}C,W_{0}F_{d}\right) .  \notag
\end{gather}

\begin{theorem}
\label{Theo5-1}Given the system model (\ref{eq2-20}) with the normalized LCR 
$\left( \hat{M}_{0},\hat{N}_{0}\right) $, RCF $\left( M_{0},N_{0}\right) $
of the plant, the corresponding LCP, RCP $\left( X_{0},Y_{0}\right) ,\left( 
\hat{X}_{0},\hat{Y}_{0}\right) ,$ and the controller (\ref{eq4-26}), then
the control difference (\ref{eq5-1}) satisfies%
\begin{gather}
\left\Vert \left[ 
\begin{array}{c}
e_{u} \\ 
e_{y}%
\end{array}%
\right] \right\Vert _{2}^{2}=\left\Vert \left[ 
\begin{array}{c}
e_{u,f} \\ 
e_{y,f}%
\end{array}%
\right] \right\Vert _{2}^{2}+\left\Vert \left[ 
\begin{array}{c}
e_{u,a} \\ 
e_{y,a}%
\end{array}%
\right] \right\Vert _{2}^{2}  \label{eq5-2} \\
\leq \left( J_{0}^{2}+1\right) \left\Vert \left[ 
\begin{array}{c}
\eta _{u} \\ 
-\eta _{y}%
\end{array}%
\right] \right\Vert _{2}^{2}+\left( \bar{J}_{0}^{2}+1\right) \left\Vert \hat{%
N}_{d,0}d+f\right\Vert _{2}^{2},  \label{eq5-2a} \\
\left[ 
\begin{array}{c}
e_{u,f} \\ 
e_{y,f}%
\end{array}%
\right] =\left[ 
\begin{array}{c}
-\hat{Y}_{0,\bar{Q}_{0}} \\ 
\hat{X}_{0,\bar{Q}_{0}}%
\end{array}%
\right] \left( \hat{N}_{d,0}d+f\right) ,  \notag \\
\left[ 
\begin{array}{c}
e_{u,a} \\ 
e_{y,a}%
\end{array}%
\right] =\left[ 
\begin{array}{c}
M_{0} \\ 
N_{0}%
\end{array}%
\right] \left[ 
\begin{array}{cc}
X_{0,Q_{0}} & \text{ \ }Y_{0,Q_{0}}%
\end{array}%
\right] \left[ 
\begin{array}{c}
\eta _{u} \\ 
-\eta _{y}%
\end{array}%
\right] .  \notag
\end{gather}
\end{theorem}

\begin{proof}
We firstly prove that 
\begin{gather*}
\left[ 
\begin{array}{c}
-\hat{Y}_{0} \\ 
\hat{X}_{0}%
\end{array}%
\right] +\left[ 
\begin{array}{c}
M_{0} \\ 
N_{0}%
\end{array}%
\right] \bar{Q}_{0}^{\ast }\text{ and }\left[ 
\begin{array}{c}
M_{0} \\ 
N_{0}%
\end{array}%
\right] , \\
\bar{Q}_{0}^{\ast }=\arg \inf_{\bar{Q}_{0}\in \mathcal{RH}_{\infty
}}\left\Vert Y_{0}\hat{M}_{0}^{\ast }-X_{0}\hat{N}_{0}^{\ast }+\bar{Q}%
_{0}\right\Vert _{\infty }
\end{gather*}%
build two orthogonal subspaces. As described in Subsection \ref{Subsec3-1},
let 
\begin{equation*}
\mathcal{P}_{\mathcal{I}_{G}}:=\left[ 
\begin{array}{c}
M_{0} \\ 
N_{0}%
\end{array}%
\right] \mathcal{P}_{\mathcal{H}_{2}}\left[ 
\begin{array}{c}
M_{0} \\ 
N_{0}%
\end{array}%
\right] ^{\ast }:\mathcal{H}_{2}\rightarrow \mathcal{H}_{2},
\end{equation*}%
be the operator of an orthogonal projection onto the system image subspace $%
\mathcal{I}_{G}$. Examining 
\begin{gather*}
\left[ 
\begin{array}{c}
M_{0} \\ 
N_{0}%
\end{array}%
\right] \mathcal{P}_{\mathcal{H}_{2}}\left[ 
\begin{array}{c}
M_{0} \\ 
N_{0}%
\end{array}%
\right] ^{\ast }\left( \left[ 
\begin{array}{c}
-\hat{Y}_{0} \\ 
\hat{X}_{0}%
\end{array}%
\right] +\left[ 
\begin{array}{c}
M_{0} \\ 
N_{0}%
\end{array}%
\right] \bar{Q}_{0}\right) \\
=\left[ 
\begin{array}{c}
M_{0} \\ 
N_{0}%
\end{array}%
\right] \left( \mathcal{P}_{\mathcal{H}_{2}}\left( N_{0}^{\ast }\hat{X}%
_{0}-M_{0}^{\ast }\hat{Y}_{0}\right) +\bar{Q}_{0}\right) =0
\end{gather*}%
verifies the claim. Here, the last equation is attributed to the optimal
solution of (\ref{eq4-38}), 
\begin{equation*}
\bar{Q}_{0}^{\ast }=-\mathcal{P}_{\mathcal{H}_{2}}\left( N_{0}^{\ast }\hat{X}%
_{0}-M_{0}^{\ast }\hat{Y}_{0}\right) .
\end{equation*}%
As a result of the orthogonality, the Pythagorean equation (\ref{eq4-50}),
and (\ref{eq4-37})-(\ref{eq4-38}), we have (\ref{eq5-2}). The inequality (%
\ref{eq5-2a}) follows from Theorem \ref{Theo5-0}.
\end{proof}

Theorem \ref{Theo5-1} gives the minimum $\mathcal{L}_{2}$-norm of the
control difference $\left( e_{u},e_{y}\right) $ achieved by the optimal
controller. It is obvious that the additive faults and attacks do not affect
the system stability.

Next, the system performance under the multiplicative faults and attacks is
closely examined. At first, the fault-free systems under the attacks are
considered. Analogous to the analysis and computations in Theorem \ref%
{Theo4-3} and Corollary \ref{Co3-1}, some routine computations yield%
\begin{gather}
\left[ 
\begin{array}{c}
u^{a} \\ 
y%
\end{array}%
\right] =\left[ 
\begin{array}{c}
M_{0} \\ 
N_{0}%
\end{array}%
\right] v_{0}+\left[ 
\begin{array}{c}
-\hat{Y}_{0,\bar{Q}_{0}} \\ 
\hat{X}_{0,\bar{Q}_{0}}%
\end{array}%
\right] r_{y,0}  \notag \\
+\left[ 
\begin{array}{c}
M_{0} \\ 
N_{0}%
\end{array}%
\right] \left[ 
\begin{array}{cc}
X_{0,Q_{0}} & \text{ }Y_{0,Q_{0}}%
\end{array}%
\right] \left( \Phi ^{a}\left[ 
\begin{array}{c}
u^{a} \\ 
y%
\end{array}%
\right] +\left[ 
\begin{array}{c}
\bar{\eta}_{u} \\ 
-\bar{\eta}_{y}%
\end{array}%
\right] \right) \Longrightarrow  \notag \\
\left[ 
\begin{array}{c}
u^{a} \\ 
y%
\end{array}%
\right] =\left[ 
\begin{array}{c}
M_{0} \\ 
N_{0}%
\end{array}%
\right] \Lambda _{Q_{0}}^{a}v_{0}+\bar{\Lambda}_{Q_{0}}^{a}\left[ 
\begin{array}{c}
-\hat{Y}_{0,\bar{Q}_{0}} \\ 
\hat{X}_{0,\bar{Q}_{0}}%
\end{array}%
\right] r_{y,0}  \notag \\
+\left[ 
\begin{array}{c}
M_{0} \\ 
N_{0}%
\end{array}%
\right] \Psi _{Q_{0}}^{a}\left[ 
\begin{array}{c}
\bar{\eta}_{u} \\ 
-\bar{\eta}_{y}%
\end{array}%
\right] ,  \label{eq5-4} \\
\bar{\Lambda}_{Q_{0}}^{a}=\left( I-\left[ 
\begin{array}{c}
M_{0} \\ 
N_{0}%
\end{array}%
\right] \left[ 
\begin{array}{cc}
X_{0,Q_{0}} & \text{ }Y_{0,Q_{0}}%
\end{array}%
\right] \Phi ^{a}\right) ^{-1},  \notag \\
\Lambda _{Q_{0}}^{a}=\left( I-\left[ 
\begin{array}{cc}
X_{0,Q_{0}} & \text{ }Y_{0,Q_{0}}%
\end{array}%
\right] \Phi ^{a}\left[ 
\begin{array}{c}
M_{0} \\ 
N_{0}%
\end{array}%
\right] \right) ^{-1},  \notag \\
\Psi _{Q_{0}}^{a}=\Lambda _{Q_{0}}^{a}\left[ 
\begin{array}{cc}
X_{0,Q_{0}} & \text{ }Y_{0,Q_{0}}%
\end{array}%
\right] .  \notag
\end{gather}%
It turns out%
\begin{gather}
\left[ 
\begin{array}{c}
e_{u} \\ 
e_{y}%
\end{array}%
\right] =\left[ 
\begin{array}{c}
M_{0} \\ 
N_{0}%
\end{array}%
\right] \left( v_{0,\Delta }+r_{y,0,\Delta }+\Psi _{Q_{0}}^{a}\left[ 
\begin{array}{c}
\bar{\eta}_{u} \\ 
-\bar{\eta}_{y}%
\end{array}%
\right] \right)  \notag \\
+\left[ 
\begin{array}{c}
-\hat{Y}_{0,\bar{Q}_{0}} \\ 
\hat{X}_{0,\bar{Q}_{0}}%
\end{array}%
\right] r_{y,0}  \label{eq5-3a} \\
v_{0,\Delta }=\Psi _{Q_{0}}^{a}\Phi ^{a}\left[ 
\begin{array}{c}
M_{0} \\ 
N_{0}%
\end{array}%
\right] v_{0},r_{y,0,\Delta }=\Psi _{Q_{0}}^{a}\Phi ^{a}\left[ 
\begin{array}{c}
-\hat{Y}_{0,\bar{Q}_{0}} \\ 
\hat{X}_{0,\bar{Q}_{0}}%
\end{array}%
\right] r_{y,0}.  \notag
\end{gather}

\begin{theorem}
\label{Theo5-2}Given the system as described in Theorem \ref{Theo5-1}, then
the closed-loop dynamic under the multiplicative attacks (\ref{eq2-21}) is
stable, if%
\begin{equation}
\left\Vert \Phi ^{a}\right\Vert _{\infty }=\left\Vert \Pi ^{a}\left( I+\left[
\begin{array}{cc}
I & \text{ }0 \\ 
0 & \text{ }0%
\end{array}%
\right] \Pi ^{a}\right) ^{-1}\right\Vert _{\infty }<\left(
J_{0}^{2}+1\right) ^{-1/2}.  \label{eq5-3}
\end{equation}%
Moreover, 
\begin{gather}
\left\Vert \left[ 
\begin{array}{c}
e_{u} \\ 
e_{y}%
\end{array}%
\right] \right\Vert _{2}^{2}\leq \left( \bar{J}_{0}^{2}+1\right) \left\Vert
r_{y,0}\right\Vert _{2}^{2}+  \notag \\
\left( 
\begin{array}{c}
\left\Vert \Psi _{Q_{0}}^{a}\Phi ^{a}\right\Vert _{\infty }\left( \left( 
\bar{J}_{0}^{2}+1\right) ^{1/2}\left\Vert r_{y,0}\right\Vert _{2}+\left\Vert
v_{0}\right\Vert _{2}\right) + \\ 
\left\Vert \Psi _{Q_{0}}^{a}\right\Vert _{\infty }\left\Vert \left[ 
\begin{array}{c}
\bar{\eta}_{u} \\ 
-\bar{\eta}_{y}%
\end{array}%
\right] \right\Vert _{2}%
\end{array}%
\right) ^{2}.  \label{eq5-5}
\end{gather}
\end{theorem}

\begin{proof}
According to Small Gain Theorem (SGT), the closed-loop system (\ref{eq5-4})
is stable, if 
\begin{gather*}
\left\Vert \left[ 
\begin{array}{c}
M_{0} \\ 
N_{0}%
\end{array}%
\right] \left[ 
\begin{array}{cc}
X_{0,Q_{0}} & \text{ }Y_{0,Q_{0}}%
\end{array}%
\right] \Phi ^{a}\right\Vert _{\infty }=\left\Vert \left[ 
\begin{array}{cc}
X_{0,Q_{0}} & \text{ }Y_{0,Q_{0}}%
\end{array}%
\right] \Phi ^{a}\right\Vert _{\infty } \\
\leq \left\Vert \left[ 
\begin{array}{cc}
X_{0,Q_{0}} & \text{ }Y_{0,Q_{0}}%
\end{array}%
\right] \right\Vert _{\infty }\left\Vert \Phi ^{a}\right\Vert _{\infty }<1.
\end{gather*}%
It follows from (\ref{eq4-37}) that%
\begin{gather*}
\left\Vert \left[ 
\begin{array}{cc}
X_{0,Q_{0}} & \text{ }Y_{0,Q_{0}}%
\end{array}%
\right] \right\Vert _{\infty }\left\Vert \Phi ^{a}\right\Vert _{\infty
}=\left( J_{0}^{2}+1\right) ^{1/2}\left\Vert \Phi ^{a}\right\Vert _{\infty }
\\
\Longrightarrow \left\Vert \left[ 
\begin{array}{cc}
X_{0,Q_{0}} & \text{ }Y_{0,Q_{0}}%
\end{array}%
\right] \right\Vert _{\infty }\left\Vert \Phi ^{a}\right\Vert _{\infty }<1 \\
\Longleftrightarrow \left\Vert \Phi ^{a}\right\Vert _{\infty }<\left(
J_{0}^{2}+1\right) ^{-1/2}.
\end{gather*}%
Concerning inequality (\ref{eq5-5}), it is a straightforward result of (\ref%
{eq5-2}) on account of (\ref{eq4-37})-(\ref{eq4-38}).
\end{proof}

Now, we analyze the system dynamic simultaneously under the faults and
attacks. For our purpose, the dynamic of the residual $r_{y,0}$ with respect
to the faults is firstly analyzed. It turns out, 
\begin{equation*}
r_{y,0}=\left[ 
\begin{array}{cc}
-\hat{N}_{0} & \text{ }\hat{M}_{0}%
\end{array}%
\right] \left[ 
\begin{array}{c}
u^{a} \\ 
y%
\end{array}%
\right] =\Pi ^{f}\left[ 
\begin{array}{c}
{u}^{a} \\ 
y%
\end{array}%
\right] +{\bar{f}}.
\end{equation*}%
It leads to%
\begin{gather}
\left[ 
\begin{array}{c}
u^{a} \\ 
y%
\end{array}%
\right] =\left[ 
\begin{array}{c}
M_{0} \\ 
N_{0}%
\end{array}%
\right] v_{0}+\left[ 
\begin{array}{cc}
M_{0} & \text{ }-\hat{Y}_{0,\bar{Q}_{0}} \\ 
N_{0} & \text{ }\hat{X}_{0,\bar{Q}_{0}}%
\end{array}%
\right] \left[ 
\begin{array}{c}
\bar{\eta} \\ 
{\bar{f}}%
\end{array}%
\right]  \notag \\
+\left[ 
\begin{array}{cc}
M_{0} & \text{ }-\hat{Y}_{0,\bar{Q}_{0}} \\ 
N_{0} & \text{ }\hat{X}_{0,\bar{Q}_{0}}%
\end{array}%
\right] \left[ 
\begin{array}{c}
\bar{\Phi}^{a} \\ 
\Pi ^{f}%
\end{array}%
\right] \left[ 
\begin{array}{c}
u^{a} \\ 
y%
\end{array}%
\right] \Longrightarrow  \notag \\
\left[ 
\begin{array}{c}
u^{a} \\ 
y%
\end{array}%
\right] =\left( I-\left[ 
\begin{array}{cc}
M_{0} & \text{ }-\hat{Y}_{0,\bar{Q}_{0}} \\ 
N_{0} & \text{ }\hat{X}_{0,\bar{Q}_{0}}%
\end{array}%
\right] \left[ 
\begin{array}{c}
\bar{\Phi}^{a} \\ 
\Pi ^{f}%
\end{array}%
\right] \right) ^{-1}\left[ 
\begin{array}{c}
M_{0} \\ 
N_{0}%
\end{array}%
\right] v_{0}+  \notag \\
\left( I-\left[ 
\begin{array}{cc}
M_{0} & \text{ }-\hat{Y}_{0,\bar{Q}_{0}} \\ 
N_{0} & \text{ }\hat{X}_{0,\bar{Q}_{0}}%
\end{array}%
\right] \left[ 
\begin{array}{c}
\bar{\Phi}^{a} \\ 
\Pi ^{f}%
\end{array}%
\right] \right) ^{-1}\left[ 
\begin{array}{cc}
M_{0} & \text{ }-\hat{Y}_{0,\bar{Q}_{0}} \\ 
N_{0} & \text{ }\hat{X}_{0,\bar{Q}_{0}}%
\end{array}%
\right] \left[ 
\begin{array}{c}
\bar{\eta} \\ 
{\bar{f}}%
\end{array}%
\right]  \label{eq5-6} \\
\bar{\Phi}^{a}=\left[ 
\begin{array}{cc}
X_{0,Q_{0}} & \text{ }Y_{0,Q_{0}}%
\end{array}%
\right] \Phi ^{a},  \notag \\
\bar{\eta}=\left[ 
\begin{array}{cc}
X_{0,Q_{0}} & \text{ }Y_{0,Q_{0}}%
\end{array}%
\right] \left[ 
\begin{array}{c}
\bar{\eta}_{u} \\ 
-\bar{\eta}_{y}%
\end{array}%
\right] .  \notag
\end{gather}

\begin{theorem}
\label{Theo5-3}Given the system as described in Theorem \ref{Theo5-1}, then
the closed-loop dynamic under the faults (\ref{eq2-24}) and attacks (\ref%
{eq2-21}) is stable, if%
\begin{equation}
b:=\left( \left( J_{0}^{2}+1\right) \left\Vert \Phi ^{a}\right\Vert _{\infty
}^{2}+\left( \bar{J}_{0}^{2}+1\right) \left\Vert \Pi ^{f}\right\Vert
_{\infty }^{2}\right) ^{1/2}<1.  \label{eq5-7}
\end{equation}%
Moreover, 
\begin{gather}
\left\Vert \left[ 
\begin{array}{c}
e_{u} \\ 
e_{y}%
\end{array}%
\right] \right\Vert _{2}\leq \frac{b}{1-b}\left\Vert v_{0}\right\Vert _{2}+ 
\notag \\
\frac{1}{1-b}\left( \left( J_{0}^{2}+1\right) \left\Vert \left[ 
\begin{array}{c}
\bar{\eta}_{u} \\ 
-\bar{\eta}_{y}%
\end{array}%
\right] \right\Vert _{2}^{2}+\left( \bar{J}_{0}^{2}+1\right) \left\Vert {%
\bar{f}}\right\Vert _{2}^{2}\right) ^{1/2}.  \label{eq5-8}
\end{gather}
\end{theorem}

\begin{proof}
The proof of (\ref{eq5-7}) follows from (\ref{eq5-6}) using SGT, that is,
the closed-loop system (\ref{eq5-6}) is stable, if 
\begin{gather*}
\left\Vert \left[ 
\begin{array}{cc}
M_{0} & \text{ }-\hat{Y}_{0,\bar{Q}_{0}} \\ 
N_{0} & \text{ }\hat{X}_{0,\bar{Q}_{0}}%
\end{array}%
\right] \left[ 
\begin{array}{c}
\bar{\Phi}^{a} \\ 
\Pi ^{f}%
\end{array}%
\right] \right\Vert _{\infty } \\
=\left( \left\Vert \left[ 
\begin{array}{c}
M_{0} \\ 
N_{0}%
\end{array}%
\right] \bar{\Phi}^{a}\right\Vert _{\infty }^{2}+\left\Vert \left[ 
\begin{array}{c}
-\hat{Y}_{0,\bar{Q}_{0}} \\ 
\hat{X}_{0,\bar{Q}_{0}}%
\end{array}%
\right] \Pi ^{f}\right\Vert _{\infty }^{2}\right) ^{1/2} \\
=\left( \left( J_{0}^{2}+1\right) \left\Vert \Phi ^{a}\right\Vert _{\infty
}^{2}+\left( \bar{J}_{0}^{2}+1\right) \left\Vert \Pi ^{f}\right\Vert
_{\infty }^{2}\right) ^{1/2}<1.
\end{gather*}%
Notice that under the condition (\ref{eq5-7}), 
\begin{gather*}
\left\Vert \left( I-\left[ 
\begin{array}{cc}
M_{0} & \text{ }-\hat{Y}_{0,\bar{Q}_{0}} \\ 
N_{0} & \text{ }\hat{X}_{0,\bar{Q}_{0}}%
\end{array}%
\right] \left[ 
\begin{array}{c}
\bar{\Phi}^{a} \\ 
\Pi ^{f}%
\end{array}%
\right] \right) ^{-1}-I\right\Vert _{\infty } \\
\leq \frac{\left\Vert \left[ 
\begin{array}{cc}
M_{0} & \text{ }-\hat{Y}_{0,\bar{Q}_{0}} \\ 
N_{0} & \text{ }\hat{X}_{0,\bar{Q}_{0}}%
\end{array}%
\right] \left[ 
\begin{array}{c}
\bar{\Phi}^{a} \\ 
\Pi ^{f}%
\end{array}%
\right] \right\Vert _{\infty }}{1-\left\Vert \left[ 
\begin{array}{cc}
M_{0} & \text{ }-\hat{Y}_{0,\bar{Q}_{0}} \\ 
N_{0} & \text{ }\hat{X}_{0,\bar{Q}_{0}}%
\end{array}%
\right] \left[ 
\begin{array}{c}
\bar{\Phi}^{a} \\ 
\Pi ^{f}%
\end{array}%
\right] \right\Vert _{\infty }}\leq \frac{b}{1-b}, \\
\left\Vert \left( I-\left[ 
\begin{array}{cc}
M_{0} & \text{ }-\hat{Y}_{0,\bar{Q}_{0}} \\ 
N_{0} & \text{ }\hat{X}_{0,\bar{Q}_{0}}%
\end{array}%
\right] \left[ 
\begin{array}{c}
\bar{\Phi}^{a} \\ 
\Pi ^{f}%
\end{array}%
\right] \right) ^{-1}\right\Vert _{\infty }\leq \frac{1}{1-b}.
\end{gather*}%
It yields 
\begin{gather*}
\left\Vert \left[ 
\begin{array}{c}
e_{u} \\ 
e_{y}%
\end{array}%
\right] \right\Vert _{2}\leq \frac{1}{1-b}\left( 
\begin{array}{c}
b\left\Vert \left[ 
\begin{array}{c}
M_{0} \\ 
N_{0}%
\end{array}%
\right] v_{0}\right\Vert _{2} \\ 
+\left\Vert \left[ 
\begin{array}{cc}
M_{0} & \text{ }-\hat{Y}_{0,\bar{Q}_{0}} \\ 
N_{0} & \text{ }\hat{X}_{0,\bar{Q}_{0}}%
\end{array}%
\right] \left[ 
\begin{array}{c}
\bar{\eta} \\ 
{\bar{f}}%
\end{array}%
\right] \right\Vert _{2}%
\end{array}%
\right) \\
\leq \frac{1}{1-b}\left( b\left\Vert v_{0}\right\Vert _{2}+\sqrt{\left\Vert %
\left[ 
\begin{array}{c}
M_{0} \\ 
N_{0}%
\end{array}%
\right] \bar{\eta}\right\Vert _{2}^{2}+\left\Vert \left[ 
\begin{array}{c}
-\hat{Y}_{0,\bar{Q}_{0}} \\ 
\hat{X}_{0,\bar{Q}_{0}}%
\end{array}%
\right] {\bar{f}}\right\Vert _{2}^{2}}\right) .
\end{gather*}%
Observe that%
\begin{equation*}
\left\Vert \left[ 
\begin{array}{c}
M_{0} \\ 
N_{0}%
\end{array}%
\right] \bar{\eta}\right\Vert _{2}=\left\Vert \left[ 
\begin{array}{cc}
X_{0,Q_{0}} & \text{ }Y_{0,Q_{0}}%
\end{array}%
\right] \left[ 
\begin{array}{c}
\bar{\eta}_{u} \\ 
-\bar{\eta}_{y}%
\end{array}%
\right] \right\Vert _{2}.
\end{equation*}%
Hence, following Theorem \ref{Theo5-0} we finally obtain (\ref{eq5-8}).
\end{proof}

Theorem \ref{Theo5-3} provides us with a sufficient condition of the
closed-loop stability under the worst case operation condition, the system
simultaneously under attacks and faults. Notice that, in the condition (\ref%
{eq5-7}), $J_{0}$ and $\bar{J}_{0}$ are the only two values designed by the
designer, which reach the minimum, when the optimal feedback gains $Q_{0}$
and $\bar{Q}_{0}$ are adopted. In this case, a tighter upper-bound of the $%
\mathcal{L}_{2}$-norm of the control difference $\left( e_{u},e_{y}\right) $
is reached as well.

\subsection{A subspace-based privacy-preserving scheme \label{SubsectionV-C}}

The last few years have seen a sharp rise in interest in privacy-preserving
control, driven by the growing use of networked, cloud-based, and
data-driven control architectures. These settings expose internal states,
system parameters, and user data to potential leakage, making system-level
privacy a core concern. Congruously, new research directions have emerged
that integrate differential privacy, encryption, and privacy-aware
controller design into classical control theory. This has produced a rapidly
expanding body of work aiming to ensure reliable control while protecting
sensitive system information \cite%
{HanPrivacy2018,LUreviewARC2019,Kawano2020TAC,Survey-Segovia-Ferreira2024}.
For instance, recent advances in privacy-performance trade-off
characterizations \cite{KAWANO2021automatica} and co-design of controllers
with privacy mechanisms \cite{Kawano2020TAC}.

In this subsection, we present a subspace-based privacy-preserving scheme in
the unified framework of control and detection. To protect sensitive system
information, the system transfer function $G$ and its state space
representation $\left( A,B,C,D\right) $, we introduce a privacy filter,
which generates an auxiliary signal added to the system output. The modified
output defines an augmented image subspace whose deviation from the original
subspace quantifies privacy via the gap metric. This approach yields a
system-theoretic privacy layer that geometrically characterizes possible
information leakage and, importantly, fully preserves the control
performance.

\subsubsection{Auxiliary signal and privacy filter}

We now consider the system modelled by (\ref{eq2-1}) whose kernel-based
closed-loop model is subject to\textbf{\ }%
\begin{equation}
\left\{ 
\begin{array}{c}
r_{y}(z)=\hat{M}\left( z\right) y(z)-\hat{N}(z)u(z)=\hat{M}\left( z\right)
r_{y,0}\left( z\right) \\ 
r_{u}(z)=X(z)u\left( z\right) +Y\left( z\right) y\left( z\right)
=X(z)r_{u,0}\left( z\right) .%
\end{array}%
\right.  \label{eq5-09}
\end{equation}%
For our purpose, an auxiliary signal $p$ generated by the privacy filter
implemented on the plant side, 
\begin{equation}
p=\Pi \left[ 
\begin{array}{c}
u \\ 
y%
\end{array}%
\right] ,\Pi \in \mathcal{RH}_{\infty },  \label{eq5-11}
\end{equation}%
is added to the plant output $y$, yielding the signal $y^{p},$ 
\begin{equation*}
y^{p}:=y+p=y+\Pi \left[ 
\begin{array}{c}
u \\ 
y%
\end{array}%
\right] .
\end{equation*}%
$y^{p}$ (instead of $y$) is transmitted to the control station over the
communication network. It follows from Theorem \ref{Theo4-2} that 
\begin{gather}
y^{p}=G_{\Delta }u+\bar{r}_{y}  \label{eq5-12} \\
G_{\Delta }=\left( N+\Delta _{N}\right) \left( M+\Delta _{M}\right) ^{-1} \\
=\left( \hat{M}+\Delta _{\hat{M}}\right) ^{-1}\left( \hat{N}+\Delta _{\hat{N}%
}\right) ,  \label{eq5-13} \\
\left( \Delta _{M},\Delta _{N}\right) =\left( -\hat{Y}\Psi ,\hat{X}\Psi
\right) ,\left[ 
\begin{array}{cc}
\Delta _{\hat{N}} & -\Delta _{\hat{M}}%
\end{array}%
\right] =\Pi ,  \notag \\
\Psi =\left( I-\Pi \left[ 
\begin{array}{c}
-\hat{Y} \\ 
\hat{X}%
\end{array}%
\right] \right) ^{-1}\Pi \left[ 
\begin{array}{c}
M \\ 
N%
\end{array}%
\right] ,  \notag \\
\bar{r}_{y}=\left( \hat{M}+\Delta _{\hat{M}}\right) ^{-1}\hat{N}_{d}d. 
\notag
\end{gather}%
To ensure the system stability, $\Pi $ is to be selected so that 
\begin{equation*}
\left( I-\Pi \left[ 
\begin{array}{c}
-\hat{Y} \\ 
\hat{X}%
\end{array}%
\right] \right) ^{-1}\in \mathcal{RH}_{\infty }.
\end{equation*}%
It is natural that the the injected signal $p$ will cause degradation of the
system performance, which is now specified. Let's recast the closed-loop
dynamic under the injected signal $p,$ 
\begin{equation}
\left\{ 
\begin{array}{l}
\hat{M}y=\hat{N}u+r_{y} \\ 
Xu+Yy^{p}=v,%
\end{array}%
\right.
\end{equation}%
into the form%
\begin{equation*}
\left[ 
\begin{array}{cc}
X & \text{ }Y \\ 
-\hat{N} & \text{ }\hat{M}%
\end{array}%
\right] \left[ 
\begin{array}{c}
u \\ 
y%
\end{array}%
\right] =\left[ 
\begin{array}{c}
v-Y\Pi \left[ 
\begin{array}{c}
u \\ 
y%
\end{array}%
\right] \\ 
r_{y}%
\end{array}%
\right] .
\end{equation*}%
It turns out%
\begin{gather*}
\left[ 
\begin{array}{c}
u \\ 
y%
\end{array}%
\right] =\left[ 
\begin{array}{c}
M \\ 
N%
\end{array}%
\right] \left( v-Y\Pi \left[ 
\begin{array}{c}
u \\ 
y%
\end{array}%
\right] \right) +\left[ 
\begin{array}{c}
-\hat{Y} \\ 
\hat{X}%
\end{array}%
\right] r_{y}\Longrightarrow \\
\left[ 
\begin{array}{c}
u \\ 
y%
\end{array}%
\right] =\left[ 
\begin{array}{c}
M \\ 
N%
\end{array}%
\right] \Phi v+\left( I+\left[ 
\begin{array}{c}
M \\ 
N%
\end{array}%
\right] Y\Pi \right) ^{-1}\left[ 
\begin{array}{c}
-\hat{Y} \\ 
\hat{X}%
\end{array}%
\right] r_{y}, \\
\Phi =\left( I+Y\Pi \left[ 
\begin{array}{c}
M \\ 
N%
\end{array}%
\right] \right) ^{-1}.
\end{gather*}%
Thus, the control difference, 
\begin{gather*}
\left[ 
\begin{array}{c}
e_{u} \\ 
e_{y}%
\end{array}%
\right] =\left[ 
\begin{array}{c}
u \\ 
y%
\end{array}%
\right] -\left[ 
\begin{array}{c}
M \\ 
N%
\end{array}%
\right] v-\left[ 
\begin{array}{c}
-\hat{Y} \\ 
\hat{X}%
\end{array}%
\right] r_{y} \\
=-\left[ 
\begin{array}{c}
M \\ 
N%
\end{array}%
\right] \Phi Y\Pi \left[ 
\begin{array}{c}
M \\ 
N%
\end{array}%
\right] v \\
-\left( I+\left[ 
\begin{array}{c}
M \\ 
N%
\end{array}%
\right] Y\Pi \right) ^{-1}\left[ 
\begin{array}{c}
M \\ 
N%
\end{array}%
\right] Y\Pi \left[ 
\begin{array}{c}
-\hat{Y} \\ 
\hat{X}%
\end{array}%
\right] r_{y},
\end{gather*}%
models the control performance degradation.

\subsubsection{Performance-preserving controller}

In order to fully recover the system performance, the controller (on the
plant side) is extended by adding a feedback of the residual signal $%
r_{y}^{p}$ as 
\begin{align}
Xu& =-Yy^{p}+v+Qr_{y}^{p}  \label{eq5-10} \\
& =\left( Q\hat{M}-Y\right) y^{p}-Q\hat{N}u+v,Q\in \mathcal{RH}_{\infty }, 
\notag \\
r_{y}^{p}& =\hat{M}y^{p}-\hat{N}u=r_{y}+\hat{M}\Pi \left[ 
\begin{array}{c}
u \\ 
y%
\end{array}%
\right] .  \notag
\end{align}%
It yields%
\begin{gather*}
\left[ 
\begin{array}{c}
u \\ 
y%
\end{array}%
\right] =\left[ 
\begin{array}{c}
M \\ 
N%
\end{array}%
\right] \left( v+\left( Q\hat{M}-Y\right) \Pi \left[ 
\begin{array}{c}
u \\ 
y%
\end{array}%
\right] \right) \\
+\left( \left[ 
\begin{array}{c}
-\hat{Y} \\ 
\hat{X}%
\end{array}%
\right] +\left[ 
\begin{array}{c}
M \\ 
N%
\end{array}%
\right] Q\right) r_{y} \\
\Longrightarrow \left[ 
\begin{array}{c}
u \\ 
y%
\end{array}%
\right] =\left[ 
\begin{array}{c}
M \\ 
N%
\end{array}%
\right] \left( I-\left( Q\hat{M}-Y\right) \Pi \left[ 
\begin{array}{c}
M \\ 
N%
\end{array}%
\right] \right) ^{-1}v \\
+\left( I-\left[ 
\begin{array}{c}
M \\ 
N%
\end{array}%
\right] \left( Q\hat{M}-Y\right) \Pi \right) ^{-1}\left( \left[ 
\begin{array}{c}
M \\ 
N%
\end{array}%
\right] Q+\left[ 
\begin{array}{c}
-\hat{Y} \\ 
\hat{X}%
\end{array}%
\right] \right) r_{y}.
\end{gather*}%
Some routine computations lead to%
\begin{gather}
\left[ 
\begin{array}{c}
u \\ 
y%
\end{array}%
\right] =\left[ 
\begin{array}{c}
-\hat{Y} \\ 
\hat{X}%
\end{array}%
\right] r_{y}+\left[ 
\begin{array}{c}
M \\ 
N%
\end{array}%
\right] \Phi ^{p}v  \notag \\
+\left[ 
\begin{array}{c}
M \\ 
N%
\end{array}%
\right] \Phi ^{p}\left( Q+\left( Q\hat{M}-Y\right) \Pi \left[ 
\begin{array}{c}
-\hat{Y} \\ 
\hat{X}%
\end{array}%
\right] \right) r_{y},  \label{eq5-9} \\
\Phi ^{p}=\left( I-\left( Q\hat{M}-Y\right) \Pi \left[ 
\begin{array}{c}
M \\ 
N%
\end{array}%
\right] \right) ^{-1}.  \notag
\end{gather}

\begin{theorem}
\label{Theo5-4}Given the plant model (\ref{eq2-1}) and controller (\ref%
{eq5-10}), setting 
\begin{gather}
\bar{v}=\left( I-\left( Q\hat{M}-Y\right) \Pi \left[ 
\begin{array}{c}
M \\ 
N%
\end{array}%
\right] \right) ^{-1}v,  \label{eq5-11a} \\
Q=\left( I+\hat{M}\Pi \left[ 
\begin{array}{c}
-\hat{Y} \\ 
\hat{X}%
\end{array}%
\right] \right) ^{-1}Y\Pi \left[ 
\begin{array}{c}
-\hat{Y} \\ 
\hat{X}%
\end{array}%
\right] ,  \label{eq5-11b}
\end{gather}%
and $\Pi $ ensuring 
\begin{eqnarray}
\left( I-\Pi \left[ 
\begin{array}{c}
-\hat{Y} \\ 
\hat{X}%
\end{array}%
\right] \right) ^{-1} &\in &\mathcal{RH}_{\infty },  \label{eq5-13a} \\
\left( I+\hat{M}\Pi \left[ 
\begin{array}{c}
-\hat{Y} \\ 
\hat{X}%
\end{array}%
\right] \right) ^{-1} &\in &\mathcal{RH}_{\infty },  \label{eq5-13b} \\
\left( I-\left( Q\hat{M}-Y\right) \Pi \left[ 
\begin{array}{c}
M \\ 
N%
\end{array}%
\right] \right) ^{-1} &\in &\mathcal{RH}_{\infty }  \label{eq5-13c}
\end{eqnarray}%
results in%
\begin{equation}
\left[ 
\begin{array}{c}
u \\ 
y%
\end{array}%
\right] =\left[ 
\begin{array}{c}
M \\ 
N%
\end{array}%
\right] \bar{v}+\left[ 
\begin{array}{c}
-\hat{Y} \\ 
\hat{X}%
\end{array}%
\right] r_{y}.  \label{eq5-14}
\end{equation}
\end{theorem}

\begin{proof}
The proof follows from (\ref{eq5-9}) immediately.
\end{proof}

Theorem \ref{Theo5-4} showcases that the controller (\ref{eq5-10}) enables
fully recovering the system performance, as far as we are able to design the
privacy filter $\Pi $ so that the conditions (\ref{eq5-13a})-(\ref{eq5-13c})
are satisfied. Herewith, it is worth remarking that the conditions (\ref%
{eq5-13a})-(\ref{eq5-13b}) are necessary conditions for the system stability
and the realization of the residual feedback controller $Qr_{y}^{p},$ while
the necessity of the condition (\ref{eq5-13c}) depends on the use of the
reference signal $v.$ For instance, for $v=0,$ the existence of (\ref%
{eq5-13c}) becomes nonrelevant. Below, the realization of the
privacy-preserving system with a performance-preserving controller is
summarized.

\begin{algorithm}
\textbf{Realization of the privacy-preserving system}
\begin{enumerate}
\item Implement the privacy system (\ref{eq5-11}) (on the plant side), add
the auxiliary signal $p$ to $y$ and send $y^{p}$ to the control station;
\item Set the controller (on the control station side) for given $v,$ either
as an observer-based state feedback controller%
\begin{equation*}
u=F\hat{x}+\bar{v}+Qr_{y}^{p} \label{eq5-10a},
\end{equation*}%
or equivalently 
\begin{align*}
u& =Ky^{p}+\left( X+Q\hat{N}\right) ^{-1}\bar{v}+Qr_{y}^{p}, \\
K& =-\left( X+Q\hat{N}\right) ^{-1}\left( Y-Q\hat{M}\right) ,
\end{align*}%
where $\bar{v}$ is set according to (\ref{eq5-11a}).
\end{enumerate}
\end{algorithm}

\subsubsection{System-level privacy and privacy filter design}

We now attend to designing $\Pi $ in the regard of system-level privacy. To
this end, a measure of system privacy in terms of the similarity between $G$
and $G_{\Delta }$ or equivalently between the image subspaces $\mathcal{I}%
_{G}$ and $\mathcal{I}_{G_{\Delta }},$ 
\begin{equation*}
\mathcal{I}_{G_{\Delta }}=\left\{ \left[ 
\begin{array}{c}
u \\ 
y%
\end{array}%
\right] :\left[ 
\begin{array}{c}
u \\ 
y%
\end{array}%
\right] =\left[ 
\begin{array}{c}
M-\hat{Y}\Psi \\ 
N+\hat{X}\Psi%
\end{array}%
\right] v_{\Delta },v_{\Delta }\in \mathcal{H}_{2}\right\}
\end{equation*}%
is introduced. As introduced in Subsection \ref{Subsec3-1}, the gap metric
between $\mathcal{I}_{G}$ and $\mathcal{I}_{G_{\Delta }}$ defined by%
\begin{equation*}
\delta \left( \mathcal{I}_{G},\mathcal{I}_{G_{\Delta }}\right) =\left\Vert 
\mathcal{P}_{\mathcal{I}_{G}}-\mathcal{P}_{\mathcal{I}_{G_{\Delta
}}}\right\Vert .
\end{equation*}%
is a well-established mathematical tool serving for such a purpose.

\begin{definition} \label{Def5-1}
Given the plant model (\ref{eq2-1}) and controller (\ref{eq5-10a}), the
system-level privacy induced by signal $p$ given in (\ref{eq5-11}) is
defined as the gap metric $\delta \left( \mathcal{I}_{G},\mathcal{I}%
_{G_{\Delta }}\right) $.
\end{definition}

It follows from (\ref{eq3-51a}) and (\ref{eq3-51b}) immediately that (i)
when $\Pi =0,$ privacy level is zero, (ii) if $\mathcal{I}_{G}\bot \mathcal{I%
}_{G_{\Delta }},$ then the maximal privacy level is achieved, i.e. $\delta
\left( \mathcal{I}_{G},\mathcal{I}_{G_{\Delta }}\right) =1.$ Notice that $%
\mathcal{I}_{G}\bot \mathcal{I}_{G_{\Delta }}$ implies $\forall \left[ 
\begin{array}{c}
u \\ 
y%
\end{array}%
\right] \in \mathcal{I}_{G},\left[ 
\begin{array}{c}
u \\ 
y^{p}%
\end{array}%
\right] \in \mathcal{I}_{G_{\Delta }},$ 
\begin{equation}
\left\langle \left[ 
\begin{array}{c}
u \\ 
y%
\end{array}%
\right] ,\left[ 
\begin{array}{c}
u \\ 
y^{p}%
\end{array}%
\right] \right\rangle =\dsum\limits_{k=0}^{\infty }\left(
u^{T}(k)u(k)+y^{T}(k)y^{p}(k)\right) =0.  \label{eq5-54}
\end{equation}%
Consequently, it is impossible to identify the model $G=NM^{-1},$ or
equivalently $\left( A,B,C,D\right) ,$ by means of the data $\left(
u,y^{p}\right) $ belonging to $\mathcal{I}_{G_{\Delta }}.$ Generally
speaking, the gap metric between system subspaces is a system-theoretic
privacy metric, quantifying how much additional information that the
augmented signal $y^{p}$ reveals about the system dynamic.

\begin{remark}
Generally speaking, system identification algorithms are computations using
process data collected over a finite time interval. On the other hand, the
inner product (\ref{eq5-54}) is defined over $\left[ 0,\infty \right) .$ In
order to preserve the privacy level, an inner product over finite time
interval $\left[ k_{0},k_{1}\right] ,$%
\begin{equation*}
\left\langle \left[ 
\begin{array}{c}
u \\ 
y%
\end{array}%
\right] ,\left[ 
\begin{array}{c}
u \\ 
y^{p}%
\end{array}%
\right] \right\rangle =\dsum\limits_{k=k_{0}}^{k_{1}}\left(
u^{T}(k)u(k)+y^{T}(k)y^{p}(k)\right) ,
\end{equation*}%
is to be defined. As a result, time-varying (normalized) coprime
factorizations that lead to time-varying post- and pre-filters are applied.
This issue will not be addressed in this work. The reader is referred to 
\cite{Feintuch_book,LDZ2024} for the corresponding concepts and computation
algorithms. 
\end{remark}

Next, we study the computation of the gap metric induced privacy-level and
design of the privacy filter (\ref{eq5-11}). It is known in the literature
that 
\begin{gather}
\vec{\delta}\left( \mathcal{I}_{G},\mathcal{I}_{G_{\Delta }}\right)
=\inf_{Q\in \mathcal{RH}_{\infty }}\left\Vert I_{G_{0}}-I_{G_{0,\Delta
}}Q\right\Vert _{\infty },  \label{eq5-16} \\
I_{G_{0}}=\left[ 
\begin{array}{c}
M_{0} \\ 
N_{0}%
\end{array}%
\right] ,I_{G_{0,\Delta }}=\left[ 
\begin{array}{c}
M_{0}-\hat{Y}_{0}\Psi _{0} \\ 
N_{0}+\hat{X}_{0}\Psi _{0}%
\end{array}%
\right] \Theta _{0},
\end{gather}%
where $\left( M_{0},N_{0}\right) $ is the normalized RCP of $G,\left( \hat{X}%
_{0},\hat{Y}_{0}\right) $ the corresponding RCP of the controller, and 
\begin{equation*}
\Psi _{0}=\left( I-\Pi \left[ 
\begin{array}{c}
-\hat{Y}_{0} \\ 
\hat{X}_{0}%
\end{array}%
\right] \right) ^{-1}\Pi \left[ 
\begin{array}{c}
M_{0} \\ 
N_{0}%
\end{array}%
\right]
\end{equation*}%
with the pre-filter $\Theta _{0}\in \mathcal{RH}_{\infty }$ normalizing $%
I_{G_{0,\Delta }},$ i.e. $\ I_{G_{0,\Delta }}^{\ast }I_{G_{0,\Delta }}=I.$

The optimization (\ref{eq5-16}) is an MMP whose solution requires
considerable computational effort. Remembering our final goal of designing
the privacy filter (\ref{eq5-11}), solving the MMP (\ref{eq5-16}) remarkably
complicates the design procedure. For our purpose, the subsequent known
results are presented, which provides us with a tight lower bound of the gap
metric $\delta \left( \mathcal{I}_{G},\mathcal{I}_{G_{\Delta }}\right) .$ On
account of the fact that the privacy level $\delta \left( \mathcal{I}_{G},%
\mathcal{I}_{G_{\Delta }}\right) $ should be larger than zero and $\delta
\left( \mathcal{I}_{G},\mathcal{I}_{G_{\Delta }}\right) $ is smaller than
one due to the relation between $\mathcal{I}_{G}$ and $\mathcal{I}%
_{G_{\Delta }},$ below only the case $0<\delta \left( \mathcal{I}_{G},%
\mathcal{I}_{G_{\Delta }}\right) <1$ is addressed.

\begin{lemma}
\label{Le5-1}\cite{Vinnicombe-book} Given the normalized SIRs and SKRs,%
\begin{equation*}
I_{G_{i}}=\left[ 
\begin{array}{c}
M_{i} \\ 
N_{i}%
\end{array}%
\right] ,K_{G_{i}}=\left[ 
\begin{array}{cc}
-\hat{N}_{i} & \text{ }\hat{M}_{i}%
\end{array}%
\right] ,i=1,2,
\end{equation*}%
and assume that $0<\delta \left( \mathcal{I}_{G_{i}},\mathcal{I}%
_{G_{j}}\right) <1,i\neq j,$ then it holds%
\begin{equation}
\delta \left( \mathcal{I}_{G_{i}},\mathcal{I}_{G_{j}}\right) \geq \left\Vert
K_{G_{i}}I_{G_{j}}\right\Vert _{\infty }=\inf_{Q\in \mathcal{RL}_{\infty
}}\left\Vert I_{G_{i}}-I_{G_{j}}Q\right\Vert _{\infty }.  \label{eq5-17}
\end{equation}
\end{lemma}

It is known \cite{Vinnicombe-book} that $\inf_{Q\in \mathcal{RL}_{\infty
}}\left\Vert I_{G_{i}}-I_{G_{j}}Q\right\Vert _{\infty }$ is the gap metric
measuring the similarity between two $\ell _{2}$ image subspaces 
\begin{equation*}
\mathcal{I}_{G_{i}}^{\mathcal{L}_{2}}=\left\{ \left[ 
\begin{array}{c}
u \\ 
y%
\end{array}%
\right] :\left[ 
\begin{array}{c}
u \\ 
y%
\end{array}%
\right] =\left[ 
\begin{array}{c}
M_{i} \\ 
N_{i}%
\end{array}%
\right] v,v\in \ell _{2}\right\} ,
\end{equation*}%
$i=1,2.$ We denote it by $\delta ^{\mathcal{L}_{2}}\left( \mathcal{I}%
_{G_{i}},\mathcal{I}_{G_{j}}\right) .$ The equation in (\ref{eq5-17}) means
that 
\begin{equation*}
\delta ^{\mathcal{L}_{2}}\left( \mathcal{I}_{G_{i}},\mathcal{I}%
_{G_{j}}\right) =\left\Vert K_{G_{i}}I_{G_{j}}\right\Vert _{\infty }.
\end{equation*}%
In this context, the concept of system-level $\mathcal{L}_{2}$-privacy is
introduced.

\begin{definition}
\label{Def5-2}Given the plant model (\ref{eq2-1}) and controller (\ref%
{eq5-10}), the gap metric,%
\begin{equation*}
\delta ^{\mathcal{L}_{2}}\left( \mathcal{I}_{G},\mathcal{I}_{G_{\Delta
}}\right) =\inf_{Q\in \mathcal{RL}_{\infty }}\left\Vert
I_{G_{0}}-I_{G_{0,\Delta }}Q\right\Vert _{\infty },
\end{equation*}%
is called $\mathcal{L}_{2}$-privacy induced by signal $p$ given in (\ref%
{eq5-11}). 
\end{definition}

\begin{theorem}
\label{Theo5-5}Given the plant model (\ref{eq2-1}) and controller (\ref%
{eq5-10}), the $\mathcal{L}_{2}$-privacy preserved by means of the privacy
filter (\ref{eq5-11}) is given by%
\begin{gather}
\delta ^{\mathcal{L}_{2}}\left( \mathcal{I}_{G},\mathcal{I}_{G_{\Delta
}}\right) =\left\Vert \Pi _{0}\left[ 
\begin{array}{c}
M_{0} \\ 
N_{0}%
\end{array}%
\right] \right\Vert _{\infty },  \label{eq5-24} \\
\Pi _{0}=R_{N}\left( \left[ 
\begin{array}{cc}
-\hat{N}_{0} & \text{ }\hat{M}_{0}%
\end{array}%
\right] -\Pi \right) ,  \label{eq5-25} \\
R_{N}^{*}R_{N}=\left( \left( \left[ 
\begin{array}{cc}
-\hat{N}_{0} & \text{ }\hat{M}_{0}%
\end{array}%
\right] -\Pi \right) \left( \left[ 
\begin{array}{cc}
-\hat{N}_{0} & \text{ }\hat{M}_{0}%
\end{array}%
\right] -\Pi \right) ^{*}\right) ^{-1}.  \notag
\end{gather}
\end{theorem}

\begin{proof}
Denote the normalized SKR of $G_{\Delta }$ by $K_{G_{0,\Delta }}.$ According
to Lemma \ref{Le5-1}, 
\begin{equation*}
\delta ^{\mathcal{L}_{2}}\left( \mathcal{I}_{G},\mathcal{I}_{G_{\Delta
}}\right) =\left\Vert K_{G_{0,\Delta }}I_{G_{0}}\right\Vert _{\infty }.
\end{equation*}%
We now determine $K_{G_{0,\Delta }}.$ Observe that 
\begin{equation*}
\left[ 
\begin{array}{c}
M_{0}-\hat{Y}_{0}\Psi _{0} \\ 
N_{0}+\hat{X}_{0}\Psi _{0}%
\end{array}%
\right] =\left( I-\left[ 
\begin{array}{c}
-\hat{Y}_{0} \\ 
\hat{X}_{0}%
\end{array}%
\right] \Pi \right) ^{-1}\left[ 
\begin{array}{c}
M_{0} \\ 
N_{0}%
\end{array}%
\right] .
\end{equation*}%
It is apparent that 
\begin{align*}
K_{G_{\Delta }}& :=R_{N}\left[ 
\begin{array}{cc}
-\hat{N}_{0} & \text{ }\hat{M}_{0}%
\end{array}%
\right] \left( I-\left[ 
\begin{array}{c}
-\hat{Y}_{0} \\ 
\hat{X}_{0}%
\end{array}%
\right] \Pi \right) \\
& =R_{N}\left( \left[ 
\begin{array}{cc}
-\hat{N}_{0} & \text{ }\hat{M}_{0}%
\end{array}%
\right] -\Pi \right)
\end{align*}%
is an SKR of $G_{\Delta },$ where $R_{N}\in \mathcal{RH}_{\infty }$ is a
post-filter to be selected. Now, let $R_{N}$ satisfy 
\begin{align*}
R_{N}^{*}R_{N}& =\left( \left( \left[ 
\begin{array}{cc}
-\hat{N}_{0} & \text{ }\hat{M}_{0}%
\end{array}%
\right] -\Pi \right) \left( \left[ 
\begin{array}{cc}
-\hat{N}_{0} & \text{ }\hat{M}_{0}%
\end{array}%
\right] -\Pi \right) ^{*}\right) ^{-1}, \\
\Pi _{0}& =R_{N}\left( \left[ 
\begin{array}{cc}
-\hat{N}_{0} & \text{ }\hat{M}_{0}%
\end{array}%
\right] -\Pi \right) =:K_{G_{0,\Delta }}
\end{align*}%
so that $\Pi _{0}$ is the normalized SKR of $G_{\Delta }$, i.e. $\Pi _{0}\Pi
_{0}^{*}=I.$ As a result, 
\begin{equation*}
\left\Vert K_{G_{0,\Delta }}I_{G_{0}}\right\Vert _{\infty }=\left\Vert \Pi
_{0}I_{G_{0}}\right\Vert _{\infty }.
\end{equation*}
\end{proof}

\begin{remark}
As an SKR of $G_{\Delta },$ $\left[ 
\begin{array}{cc}
-\hat{N}_{0} & \text{ }\hat{M}_{0}%
\end{array}%
\right] -\Pi $ is stable and right invertible. By a spectral factorization 
\cite{Zhou96} 
\begin{equation*}
\left( \left[ 
\begin{array}{cc}
-\hat{N}_{0} & \text{ }\hat{M}_{0}%
\end{array}%
\right] -\Pi \right) \left( \left[ 
\begin{array}{cc}
-\hat{N}_{0} & \text{ }\hat{M}_{0}%
\end{array}%
\right] -\Pi \right) ^{*}=\Theta _{0}\Theta _{0}^{*},
\end{equation*}%
where $\Theta _{0}$ is invertible over $\mathcal{RH}_{\infty },$ we have $R_{N}=\Theta _{0}^{-1}$. 
\end{remark}

Theorem \ref{Theo5-5} considerably simplifies the design of the privacy
filter (\ref{eq5-11}) and enables us to perform the selection of $\Pi _{0}.$
To this end, let us firstly parametrize $\Pi _{0}$ as%
\begin{equation*}
\Pi _{0}=\left[ 
\begin{array}{cc}
\Pi _{0,1} & \text{ }\Pi _{0,2}%
\end{array}%
\right] \left[ 
\begin{array}{cc}
X_{0} & \text{ }Y_{0} \\ 
-\hat{N}_{0} & \text{ }\hat{M}_{0}%
\end{array}%
\right] .
\end{equation*}%
It turns out 
\begin{equation*}
\Pi _{0}I_{G_{0}}=\Pi _{0,1},\Pi _{0}\left[ 
\begin{array}{c}
-\hat{Y}_{0} \\ 
\hat{X}_{0}%
\end{array}%
\right] =\Pi _{0,2}.
\end{equation*}%
Recall the conditions (\ref{eq5-13a})-(\ref{eq5-13c}) in Theorem \ref%
{Theo5-4} for the system stability, and settings of $Q,\bar{v}$. Attributed
to Corollary \ref{Co3-1}, they are further written as%
\begin{gather*}
\left( I-\Pi \left( \left[ 
\begin{array}{c}
-\hat{Y}_{0} \\ 
\hat{X}_{0}%
\end{array}%
\right] R_{0}+\left[ 
\begin{array}{c}
M_{0} \\ 
N_{0}%
\end{array}%
\right] \bar{R}_{0}\right) \right) ^{-1}\in \mathcal{RH}_{\infty }, \\
\left( I+\hat{M}\Pi \left( \left[ 
\begin{array}{c}
-\hat{Y}_{0} \\ 
\hat{X}_{0}%
\end{array}%
\right] R_{0}+\left[ 
\begin{array}{c}
M_{0} \\ 
N_{0}%
\end{array}%
\right] \bar{R}_{0}\right) \right) ^{-1}\in \mathcal{RH}_{\infty }, \\
\left( I-\left( Q\hat{M}-Y\right) \Pi \left[ 
\begin{array}{c}
M_{0} \\ 
N_{0}%
\end{array}%
\right] T_{0}^{-1}\right) ^{-1}\in \mathcal{RH}_{\infty }
\end{gather*}%
with $\bar{R}_{0},R_{0}$ and $T_{0}$ defined in Corollary \ref{Co3-1}.
Substituting 
\begin{gather}
\Pi =\left( I-\bar{\Pi}_{0,2}\right) \left[ 
\begin{array}{cc}
-\hat{N}_{0} & \text{ }\hat{M}_{0}%
\end{array}%
\right] -\bar{\Pi}_{0,1}\left[ 
\begin{array}{cc}
X_{0} & \text{ }Y_{0}%
\end{array}%
\right] ,  \label{eq5-27} \\
\left[ 
\begin{array}{cc}
\bar{\Pi}_{0,1} & \text{ }\bar{\Pi}_{0,2}%
\end{array}%
\right] =R_{N}^{-1}\left[ 
\begin{array}{cc}
\Pi _{0,1} & \text{ }\Pi _{0,2}%
\end{array}%
\right] ,  \notag \\
\hat{\Pi}_{0,2}=\left( I-\bar{\Pi}_{0,2}\right) R_{0}
\end{gather}%
into them allow us to express the conditions (\ref{eq5-13a})-(\ref{eq5-13c})
respectively by%
\begin{gather}
\left( I-\hat{\Pi}_{0,2}-\bar{\Pi}_{0,1}\bar{R}_{0}\right) ^{-1}\in \mathcal{%
RH}_{\infty },  \label{eq5-23} \\
\left( I+\hat{M}\left( \hat{\Pi}_{0,2}-\bar{\Pi}_{0,1}\bar{R}_{0}\right)
\right) ^{-1}\in \mathcal{RH}_{\infty },  \label{eq5-24} \\
\left( I+\left( Q\hat{M}-Y\right) \bar{\Pi}_{0,1}T_{0}^{-1}\right) ^{-1}\in 
\mathcal{RH}_{\infty }.  \label{eq5-25}
\end{gather}%
Consequently, the design of the privacy filter (\ref{eq5-11}) is
equivalently formulated as%
\begin{gather}
\sup_{\Pi _{0,1}}\delta ^{\mathcal{L}_{2}}\left( \mathcal{I}_{G},\mathcal{I}%
_{G_{\Delta }}\right) =\sup_{\Pi _{0,1}}\left\Vert \Pi _{0,1}\right\Vert
_{\infty }  \label{eq5-19} \\
\text{s.t. (\ref{eq5-23})-(\ref{eq5-25}).}  \label{eq5-26}
\end{gather}%
It is of interest to notice that $\delta ^{\mathcal{L}_{2}}\left( \mathcal{I}%
_{G},\mathcal{I}_{G_{\Delta }}\right) $ reaches its maximum for $\Pi
_{0,2}=0.$ In that case, $\Pi _{0,1}$ normalizes $\left[ 
\begin{array}{cc}
X_{0} & \text{ }Y_{0}%
\end{array}%
\right] ,$ i.e. 
\begin{equation*}
\Pi _{0,1}\left[ 
\begin{array}{cc}
X_{0} & \text{ }Y_{0}%
\end{array}%
\right] \left[ 
\begin{array}{c}
X_{0}^{\ast } \\ 
Y_{0}^{\ast }%
\end{array}%
\right] \Pi _{0,1}^{\ast }=I.\text{ }
\end{equation*}%
On the other hand, setting $\Pi _{0,2}=0$ leads to, for instance, 
\begin{equation*}
I-\left( I-\bar{\Pi}_{0,2}\right) R_{0}-\bar{\Pi}_{0,1}\bar{R}_{0}=I-R_{0}-%
\bar{\Pi}_{0,1}\bar{R}_{0},
\end{equation*}%
which implies, the condition (\ref{eq5-23}) is probably not satisfied.
Nevertheless, this observation inspires the subsequent algorithm for a
practical sub-optimal solution of the design problem (\ref{eq5-19})-(\ref%
{eq5-26}).

\begin{algorithm}
\textbf{Privacy filter design: a sub-optimal solution}
\begin{enumerate}
\item Select a sufficiently small $0<\varepsilon <<1$ and set%
\begin{equation}
\bar{\Pi}_{0,2}=\varepsilon I\Longrightarrow \Pi _{0,2}=\varepsilon R_{N}
\label{eq5-21}
\end{equation}%
so that the conditions (\ref{eq5-23})-(\ref{eq5-25}) are satisfied;

\item Compute $\hat{\Pi}_{0,1}$ that normalizes $\left[ 
\begin{array}{cc}
X_{0} & \text{ }Y_{0}%
\end{array}%
\right] $ by a spectral factorization, 
\begin{equation*}
\left[ 
\begin{array}{cc}
X_{0} & \text{ }Y_{0}%
\end{array}%
\right] \left[ 
\begin{array}{cc}
X_{0} & \text{ }Y_{0}%
\end{array}%
\right] ^{\ast }=\hat{\Pi}_{0,1}^{-1}\left( \hat{\Pi}_{0,1}^{-1}\right)
^{\ast }
\end{equation*}%
and set 
\begin{equation}
\Pi _{0,1}=\sqrt{(1-\varepsilon ^{2})}\hat{\Pi}_{0,1};  \label{eq5-22}
\end{equation}

\item Set $\Pi $ according to (\ref{eq5-27}).
\end{enumerate}
\end{algorithm}

We now examine the condition $\Pi _{0}\Pi _{0}^{\ast }=I$ and (\ref{eq5-19}%
). It holds 
\begin{gather*}
\Pi _{0}\Pi _{0}^{\ast }=\varepsilon ^{2}R_{N}\left[ 
\begin{array}{cc}
-\hat{N}_{0} & \text{ }\hat{M}_{0}%
\end{array}%
\right] \left[ 
\begin{array}{c}
-\hat{N}_{0}^{\ast } \\ 
\hat{M}_{0}^{\ast }%
\end{array}%
\right] R_{N}^{\ast } \\
+(1-\varepsilon ^{2})\hat{\Pi}_{0,1}\left[ 
\begin{array}{cc}
X_{0} & \text{ }Y_{0}%
\end{array}%
\right] \left[ 
\begin{array}{c}
X_{0}^{\ast } \\ 
Y_{0}^{\ast }%
\end{array}%
\right] \hat{\Pi}_{0,1}^{\ast } \\
+\varepsilon \sqrt{(1-\varepsilon ^{2})}\left( \Xi +\Xi ^{\ast }\right)
=I+\varepsilon \sqrt{(1-\varepsilon ^{2})}\left( \Xi +\Xi ^{\ast }\right) ,
\\
\Xi =\hat{\Pi}_{0,1}\left[ 
\begin{array}{cc}
X_{0} & \text{ }Y_{0}%
\end{array}%
\right] \left[ 
\begin{array}{c}
-\hat{N}_{0}^{\ast } \\ 
\hat{M}_{0}^{\ast }%
\end{array}%
\right] R_{N}^{\ast }.
\end{gather*}%
Considering that 
\begin{equation*}
\left\Vert \left[ 
\begin{array}{cc}
\hat{\Pi}_{0,1}X_{0} & \text{ }\hat{\Pi}_{0,1}Y_{0} \\ 
-\hat{N}_{0} & \text{ }\hat{M}_{0}%
\end{array}%
\right] \right\Vert _{\infty }\leq \sqrt{2},
\end{equation*}%
for sufficiently small $\varepsilon ,\Pi _{0}\Pi _{0}^{\ast }$ can be well
approximated by $\Pi _{0}\Pi _{0}^{\ast }\approx I.$ As a result, we have a
sub-optimal solution%
\begin{equation*}
\delta ^{\mathcal{L}_{2}}\left( \mathcal{I}_{G},\mathcal{I}_{G_{\Delta
}}\right) =\sqrt{(1-\varepsilon ^{2})}\left\Vert \hat{\Pi}_{0,1}\right\Vert
_{\infty }.
\end{equation*}%
Concerning the conditions (\ref{eq5-23})-(\ref{eq5-25}), whether they hold
depends on $\left( \hat{M}_{0},R_{0},\bar{R}_{0},T_{0}\right) ,$ or
equivalently the parameters $\left( F,L,V,W\right) $ of the applied
controller. To illustrate this claim, we present the following example.

\begin{example}
Suppose that $\left( F,L,V,W\right) $ are set equal to%
\begin{gather*}
\left( F,L,T,W\right) =\left( F_{0},L_{0},V_{0},W_{0}\right)  \\
\Longrightarrow \left( M,N\right) =\left( M_{0},N_{0}\right) ,\left( \hat{M},%
\hat{N}\right) =\left( \hat{M}_{0},\hat{N}_{0}\right) .
\end{gather*}%
Here, we would like to mention the interpretation of the normalized SIR and
SKR as an LQ controller and an LS estimator, respectively \cite{Tay1998}. It
follows from Corollary \ref{Co3-1}, Lemma \ref{Le2-1} and (\ref{eq5-21})
that 
\begin{gather*}
\bar{R}_{0}=0,R_{0}=I,T_{0}=I,\hat{M}=\hat{M}_{0},\hat{\Pi}%
_{0,2}=(1+\varepsilon )I\Longrightarrow  \\
\left( I-\left( I-\bar{\Pi}_{0,2}\right) R_{0}-\bar{\Pi}_{0,1}\bar{R}%
_{0}\right) ^{-1}=\varepsilon ^{-1}I\in \mathcal{RH}_{\infty }, \\
I+\hat{M}\left( I-\bar{\Pi}_{0,2}\right) R_{0}-\bar{\Pi}_{0,1}\bar{R}%
_{0}^{-1}=I+\hat{M}_{0}(1-\varepsilon ).
\end{gather*}%
Since $\left\Vert \hat{M}_{0}\right\Vert _{\infty }\leq 1\Longrightarrow
\left\Vert \hat{M}_{0}(1-\varepsilon )\right\Vert _{\infty }<1,$ it holds $%
\left( I+\hat{M}_{0}(1-\varepsilon )\right) ^{-1}\in \mathcal{RH}_{\infty }.$
As a result, the conditions (\ref{eq5-23})-(\ref{eq5-24}) are guaranteed,
which assures the system stability and the existence of the controller $%
Qr_{y}^{p}.$ Concerning the condition (\ref{eq5-25}), it is equivalent to%
\begin{gather}
\left( I-\left( Q\hat{M}-Y\right) \Pi \left[ 
\begin{array}{c}
M_{0} \\ 
N_{0}%
\end{array}%
\right] T_{0}^{-1}\right) ^{-1}  \notag \\
=\left( I+\left( Q\hat{M}_{0}-Y_{0}\right) R_{N}^{-1}\Pi _{0,1}\right)
^{-1}\in \mathcal{RH}_{\infty }.  \label{eq5-28}
\end{gather}%
Thus, when (\ref{eq5-28}) is true, $v$ is arbitrarily selectable. Otherwise, 
$v$ is to be selected properly so that $\bar{v}$ can be realized according
to (\ref{eq5-11a}). 
\end{example}

\section{Conclusions}

In this work, we have leveraged the unified framework of control and
detection to explore analysis, simultaneous detection, fault-tolerant and
resilient control of cyber-physical control systems under attacks and
faults. In a three-step procedure, we have firstly built the
control-theoretic foundation for the subsequent work. Specifically, we have
introduced, among others, the image, kernel and residual subspaces in the
system input-output data space $\left( u,y\right) ,$ gap metric as the
distance between two subspaces, and the Bezout identity-based
transformations between $\left( u,y\right) $ and $\left( r_{u},r_{y}\right) $
as well as the image and residual subspaces, where Theorem \ref{Theo3-1} and
Corollary \ref{Co3-1} function as a basic mathematical tool. A further
result is the PnP structure that builds an essential system configuration
adopted in the proposed fault-tolerant and attack-resilient control schemes.
We have then reviewed the MTD, watermark methods, the auxiliary system-based
detection method, and characterized zero dynamics attacks in the unified
framework. On this basis, the alternative and enhanced forms of these
methods have been proposed.

The major part of our endeavours has been devoted to establishing a
control-theoretic paradigm, in which analysis, simultaneous detection of
faults and attacks, fault-tolerant and attack resilient control of CPCSs are
unified addressed. Concerning system analysis, three essential information
and structural properties have been revealed. They are the existence of (i)
an attack information potential between the plant and control station and
(ii) the duality between the faulty dynamic and attack dynamic, and (iii)
the interpretation of cyber-attacks as uncertainty in the controller. Aiming
at a general and detection method-independent definition of stealthy
attacks, the concept of attack detectability has been introduced, which is
complemented by the concept of actuability of attacks. The latter is
overlooked in the literature, but important for addressing attack design
issues. The proof of the existence conditions of undetectable and
actuable/unactuable attacks have been provided as well. An immediate
application of the aforementioned results is the introduction of image
attacks as a general form of undetectable (stealthy) attacks that provides a
detection method-independent platform for (stealthy) attacks design. The
last of part of the system analysis has been dedicated to system
vulnerability. To this end, the concept of attack stability margin as a
measure of the system vulnerability has been introduced, and a design
procedure of image attacks and the proof of a sufficient condition of the
system vulnerability under image attacks have been provided.

As the last part of our work, simultaneous detection of faults and attacks
as well as fault-tolerant and attack-resilient control have been explored.
Specifically, three simultaneous detection schemes have been presented. In
the PnP structure, a fault-tolerant and attack-resilient control system has
been designed, whose control performance and stability conditions have been
analyzed. Concerning system privacy-preserving, we have proposed a
subspace-based privacy-preserving scheme whose core is a privacy filter
generating an auxiliary signal. For our purpose, the concepts of
system-level privacy and $\mathcal{L}_{2}$-privacy based on the similarity
of two system subspaces and measured by gap metric have been introduced. In
this regard, the $\mathcal{L}_{2}$-privacy-filter design problem has been
formulated and solved in the form of an optimization problem.

At the end of this paper, we would like to highlight three key properties
that have been for the first time revealed and utilized in this work,

\begin{itemize}
\item the attack information potential between the control station and plant,

\item the duality between the faulty dynamic and attack dynamic, and

\item the two-site control and detection architecture for configuring CPCSs.
\end{itemize}

Most of the significant results presented in this paper, including
simultaneous detection of faults and attack, fault-tolerant and
attack-resilient control, and subspace-based privacy-preserving,
vulnerability analysis-aided image attack design schemes, have been achieved
based on these three properties. Note that the last two design schemes are
indeed the consequence of the duality between the faulty and attack
dynamics, although they have not been explicitly introduced in this regard.
We are convinced that these three system properties are of enormous
potentials to facilitate the development of advanced methods towards secure
and safe CPCSs.


\end{document}